\documentclass[openany]{now}

\usepackage{amsmath}
\usepackage{amsfonts}
\usepackage{amssymb}
\usepackage{dsfont}
\usepackage{paralist}
\usepackage{color}

\usepackage{bm}
\usepackage{nicefrac}
\renewcommand{\mathbf}[1]{\boldsymbol{#1}}
\renewcommand{\vec}[1]{\mathbf{#1}}

\newcommand{\inc}[1]{\left(#1\right)}

\makeatletter
\newcommand*{\defeq}{\mathrel{\rlap{%
                     \raisebox{0.3ex}{$\m@th\cdot$}}%
                     \raisebox{-0.3ex}{$\m@th\cdot$}}%
                    =}
\makeatother

\DeclareMathOperator{\tr}{tr}

\DeclareMathOperator{\poly}{poly}

\newcommand{\beq}{\begin{equation}}
\newcommand{\eeq}{\end{equation}}
\newcommand{\beqn}{\begin{equation*}}
\newcommand{\eeqn}{\end{equation*}}

\newcommand{\Span}[1]{\text{Span}\left( #1 \right)}

\newcommand{\C}{\ensuremath{\mathbb{C}}}

\newcommand{\setft}[1]{\mathrm{#1}}
\newcommand{\lin}[1]{\mathcal{L}\left(#1\right)}
\newcommand{\herm}[1]{\mathcal{H}\left(#1\right)}

\newcommand{\unitary}[1]{\mathcal{U}\left(#1\right)}

\newcommand{\pos}[1]{\setft{Pos}\left(#1\right)}
\newcommand{\rankpi}{\operatorname{rank}(\Pi_{ij})}
\def\X{\mathcal{X}}
\def\Y{\mathcal{Y}}

\def\poly{\textup{poly}}

\newtheorem{Question}{Question}
\newtheorem{Hint}{Hint}

\newcommand{\pr}{\operatorname{Pr}}
\newcommand{\hout}{H_{\rm out}}
\newcommand{\hstab}{H_{\rm stab}}

\newcommand{\psiprod}{\ket{\psi_{\rm prod}}}
\newcommand{\class}[1]{\textup{#1}}

\newcommand{\lmin}[1] {\lambda_{\operatorname{min}}(#1)}

\newcommand{\lh}{\operatorname{LH}}
\newcommand{\flh}{\operatorname{5-LH}}
\newcommand{\tlh}{\operatorname{2-LH}}
\newcommand{\thlh}{\operatorname{3-LH}}
\newcommand{\klhh}{\operatorname{k-LH}}
\newcommand{\qma}{\operatorname{QMA}}

\newcommand{\heff}{H_{\rm eff}}
\newcommand{\tH}{\widetilde{H}}
\newcommand{\tA}{\widetilde{A}}
\newcommand{\tB}{\widetilde{B}}
\newcommand{\se}{\Sigma_-(z)}

\newcommand{\trace}{{\rm Tr}}

\newcommand{\comment}[1]{}
\newcommand{\sX}{\spa{X}}
\newcommand{\sY}{\spa{Y}}
\newcommand{\sYi}{\spa{Y}_i}
\newcommand{\sYio}{\spa{Y}_{i1}}
\newcommand{\sYit}{\spa{Y}_{i2}}
\newcommand{\sZ}{\spa{Z}}
\newcommand{\rhoeq}{\rho_{\rm eq}}

\newcommand{\norm}[1]{\left\|\,#1\,\right\|}       
\newcommand{\enorm}[1]{\norm{#1}_{\mathrm{2}}}      
\newcommand{\trnorm}[1]{\norm{#1}_{\mathrm {tr}}}  
\newcommand{\fnorm}[1]{\norm{#1}_{\mathrm {F}}}    
\newcommand{\snorm}[1]{\norm{#1}_{\mathrm {\infty}}}    

\newcommand{\set}[1]{{\left\{#1\right\}}}    
\newcommand{\abs}[1]{\left\lvert #1 \right\rvert}

\newcommand{\complex}{{\mathbb C}}
\newcommand{\reals}{{\mathbb R}}
\newcommand{\ints}{{\mathbb Z}}
\newcommand{\B}{\complex^2}

\newcommand{\Bd}{\complex^d}
\newcommand{\Bdi}[1]{(\Bd)^{\otimes {#1}}}

\def\ket#1{ | #1 \rangle}
\def\bra#1{{\langle #1 | }}
\newcommand{\ketbra}[2]{\ket{#1}\!\bra{#2}}        
\newcommand{\braket}[2]{\mbox{$\langle #1  | #2 \rangle$}}

\newcommand{\LL}{\mathcal{L}}
\newcommand{\DD}{\mathcal{D}}
\newcommand{\HH}{\mathcal{H}}
\newcommand{\UU}{\mathcal{U}}
\newcommand{\hin}{H_{\rm in}}
\newcommand{\hprop}{H_{\rm prop}}

\newcommand{\nl} {\mathcal{L}_1}
\newcommand{\nll} {\mathcal{L}_2}
\newcommand{\ve}[1]{\mathbf{#1}}
\newcommand{\tq}{$2$-QSAT}
\newcommand{\phiab}{\ket{\phi}_{ab}}
\newcommand{\phibc}{\ket{\phi}_{bc}}

\newcommand{\ayes}{A_{\rm yes}} 
\newcommand{\ano}{A_{\rm no}} 

\newcommand{\spa}[1]{\mathcal{#1}}
\newcommand{\enc}[1]{\langle #1 \rangle}

\newtheorem{obs}[theorem]{Observation}
\newtheorem{problem}[theorem]{Problem}

\title{Quantum Hamiltonian Complexity}

\author{
Sevag Gharibian \\
Simons Institute for the Theory of Computing,\\
University of California, Berkeley
\and Yichen Huang\\
University of California, Berkeley \and
Zeph Landau\\
Simons Institute for the Theory of Computing,\\ University of California, Berkeley
\and Seung Woo Shin\\
University of California, Berkeley
}

\date{\today}
\begin{document}

\copyrightowner{S.~Gharibian, Y.~Huang, Z.~Landau and S.~W.~Shin}
\volume{10}
\issue{3}
\pubyear{2014}
\copyrightyear{2015}
\isbn{978-1-68083-006-4}
\doi{10.1561/0400000066}
\firstpage{159}
\lastpage{282}

\frontmatter  

\maketitle

\tableofcontents

\mainmatter

\begin{abstract}
Constraint satisfaction problems are a central pillar of modern computational complexity theory. This survey provides an introduction to the rapidly growing field of Quantum Hamiltonian Complexity, which includes the study of quantum constraint satisfaction problems. Over the past decade and a half, this field has witnessed fundamental breakthroughs, ranging from the establishment of a  ``Quantum Cook-Levin Theorem'' to deep insights into the structure of $1$D low-temperature quantum systems via so-called area laws. Our aim here is to provide a computer science-oriented introduction to the subject in order to help bridge the language barrier between computer scientists and physicists in the field. As such, we include the following in this survey: (1) The motivations and history of the field, (2) a glossary of condensed matter physics terms explained in computer-science friendly language, (3) overviews of central ideas from condensed matter physics, such as indistinguishable particles, mean field theory, tensor networks, and area laws, and (4) brief expositions of selected computer science-based results in the area. For example, as part of the latter, we provide a novel information theoretic presentation of Bravyi's polynomial time algorithm for Quantum 2-SAT.
\end{abstract}

\chapter{Introduction}\label{scn:intro}
\begin{quote}
\emph{``Computers are physical objects, and computations are
physical processes. What computers can or cannot compute is
determined by the laws of physics alone\ldots''\\--- David Deutsch~\cite{NC00}}
\end{quote}
The Cook-Levin Theorem~\cite{C72,L73}, which states that the SATISFIABILITY problem is NP-complete, is one of the cornerstones of modern computational complexity theory~\cite{AB90}. One of its implications is the following simple, yet powerful, statement: Computation is, in a well-defined sense, \emph{local}. Yet, as David Deutsch's quote above perhaps foreshadows, this is not the end of the story, but rather its beginning. Indeed, just as a sequence of computational steps on a Turing machine can be encoded into local classical constraints (as in the Cook-Levin theorem), the quantum world around us also evolves ``locally'', and this quantum evolution can be encoded into an analogous notion of local \emph{quantum} constraints. The study of such quantum constraint systems underpins an emerging field at the intersection of condensed matter physics, computer science, and mathematics, known as \emph{Quantum Hamiltonian Complexity (QHC)}.\\

At the heart of QHC lies a central object of study: The notion of a \emph{local Hamiltonian} $H$, which can intuitively be thought of as a quantum constraint system (in this introduction, we will keep our discussion informal in order to convey high-level ideas; all formal definitions, including an introduction to quantum information, are given in Chapter~\ref{scn:preliminaries}). To introduce local Hamiltonians, we begin with the fact that the state of a quantum system $S$ on $n$ qudits is described by some $d^n$-dimensional complex unit vector $\ket{\psi}\in (\complex^d)^{\otimes n}$. How can we describe the evolution of the state $\ket{\psi}$ of $S$ as time elapses? This is given by the Schr\"{o}dinger equation, which says that after time $t$, the new state of our system is $e^{-iHt}\ket{\psi}$, where $H$ is a $d^n\times d^n$-dimensional complex (more precisely, Hermitian) operator called a \emph{Hamiltonian}. Here, the precise definition of the matrix exponential $e^{iHt}$ is irrelevant; what \emph{is} important is the dependence of the Schr\"{o}dinger equation on $H$. In other words, Hamiltonians are intricately tied to the evolution of quantum systems. We thus arrive at a natural question: Which classes of Hamiltonians correspond to \emph{actual} quantum evolutions for systems occurring in nature? It turns out that typically, only a special class of Hamiltonians is physically relevant: These are known as \emph{local} Hamiltonians.

Roughly, a \emph{$k$-local} Hamiltonian is a Hermitian matrix which has a succinct representation of the form
\[
	H=\sum_i H_i,
\]
where each $H_i$ acts ``non-trivially'' only on some subset of $k$ qudits. Here, each $H_i$ should be thought of as a ``quantum constraint'' or ``clause'', analogous to the notion of a $k$-local clause in classical constraint satisfaction problems. For example, just as a classical clause such as $(x_i\vee x_j\vee x_k)$ for $x_i,x_j,x_k\in\set{0,1}$ forces its bits to lie in set $x_ix_jx_k\in\set{001,010,011,100,101,110,111}$ (where $\vee$ denotes logical OR), a quantum clause $H_i$ restricts the state of the $k$ qudits it acts on to lie in a certain \emph{subspace} of $(\complex^d)^{\otimes n}$. Moreover, each clause $H_i$ requires $O(k)$ bits express (assuming all matrix entries are specified to constant precision). This is because each $H_i$ is given as a $d^k\times d^k$ complex matrix (this is made formal in Section~\ref{sscn:basicsqhc}). As a result, although $H$ itself is a matrix of dimension $d^n\times d^n$, i.e. $H$ has dimension exponential in $n$ the number of qudits, the description of $H$ in terms of local clauses $\set{H_i}$ has size \emph{polynomial} in $n$.

Since local Hamiltonians are intricately tied to the time evolution of quantum systems in nature, the goal of QHC is to study properties of local Hamiltonians $H$. Common computational tasks include estimating the \emph{ground state energy} (smallest eigenvalue) of $H$, or computing features of $H$'s \emph{ground state} (eigenvector corresponding to the smallest eigenvalue). Intuitively, the ground state can be thought of as the vector $\ket{\psi}$ which ``maximally satisfies'' the constraints $\set{H_i}$ (i.e. the ``optimal solution'' to the quantum constraint system), and is of particular interest as it encodes the state of the corresponding quantum system when cooled to low temperature. In fact, any classical Constraint Satisfaction Problem (CSP) of arity $k$ can be embedded into a $k$-local Hamiltonian, such that determining the ground state of the Hamiltonian yields the optimal solution to the CSP. (This connection is made explicit in \S\ref{sscn:basicsqhc}.) Thus, ground states are interesting from a complexity theoretic perspective.

Let us also motivate ground states from a physics perspective. Consider the case of helium-4: When cooled to near absolute zero, helium-4 relaxes to a state $\ket{\psi}$ which is the ground state of some local Hamiltonian $H$ (the precise form of $H$ is beyond the scope of this introduction). This ground state exhibits an exotic phase of matter known as \emph{superfluidity} --- it acts like a fluid with zero viscosity. (See~\cite{SFH} for a video demonstrating this remarkable phenomenon.) Ideally, we would like to understand the properties of the superfluid phase demonstrated by $\ket{\psi}$, so that, for example, we can in turn use this knowledge to design new, advanced materials. In this direction, QHC might ask questions such as: Which quantum systems in nature have a ground state with a succinct classical representation? Can we run efficient classical simulations to predict when a quantum system will exhibit interesting phenomena, such as a phase transition? Can we quantify the hardness of determining certain properties of local Hamiltonians by establishing connections to computational complexity theory? In the context of helium-4, for example, the first of these questions is particularly relevant --- to the best of our knowledge, a closed form for the ground state energy or the ground state of helium-4 remain elusive. (Heuristic approximations based on variational methods, however, have long been known; see, e.g.~\cite{SV67}.)

This state of affairs illustrates the formidable challenge facing QHC: Namely, we are interested in computing properties of $k$-local Hamiltonians $H$, which are matrices of dimension $d^n\times d^n$, whereas an efficient algorithm must run in time polynomial in $n$, the number of qudits $H$ acts on. Despite this challenge, QHC has proven a very fruitful area of research. For example, in 1999 Kitaev established~\cite{KSV02} a quantum version of the celebrated Cook-Levin theorem~\cite{C72,L73} for local Hamiltonian systems. In 2006, Bravyi gave a polynomial time algorithm for solving the quantum analogue of $2$-SATISFIABILITY, known as Quantum $2$-SAT~\cite{B06}. And though the heuristic approach of White~\cite{W92,W93} (known as ``Density Matrix Renormalization Group'') was known to solve $1$-dimensional (gapped) Hamiltonians in practice efficiently, Hastings' 1D area law in 2007~\cite{Ha07} helped explain the efficacy of this heuristic by strongly characterizing the entanglement structure of such $1$-dimensional systems. This survey aims to review a select subset of such fundamental results in QHC.

To help make this survey accessible to computer scientists with little or no background in quantum information, we begin in \S\ref{sscn:basicsqi} with a review of basic quantum information. We next establish some of the fundamental definitions of QHC in \S\ref{sscn:basicsqhc}, including an explicit sketch of how an instance of $3$-CSP can be encoded into a local Hamiltonian. With this basic background in place, we finally proceed in Chapter~\ref{scn:roadmap} to give a roadmap for the remainder of this survey.

\chapter{Preliminaries}\label{scn:preliminaries}
We now review the basics of quantum information (\S\ref{sscn:basicsqi}) and introduce fundamental concepts in QHC (\S\ref{sscn:basicsqhc}). For readers interested in further details on quantum information in general, please see the text of Nielsen and Chuang~\cite{NC00} (see also~\cite{KSV02,KLM07} for alternate textbooks) or the thesis of Gharibian~\cite{G13} for a self-contained brief introduction. (We remark that \S\ref{sscn:basicsqi} here is essentially a condensed version of Chapter 1.4 of the thesis of Gharibian~\cite{G13}.) For a review of quantum complexity theory, see the survey of Watrous~\cite{W09_2}. For a physics-oriented introduction to Hamiltonian complexity, we refer the reader to the survey of Osborne~\cite{O11}.

\section{Basics of quantum information}\label{sscn:basicsqi}

Let us first set notation and state a number of useful linear algebraic definitions; intuitively, the latter are necessary, as the mathematical language behind quantum mechanics is linear algebra. \\

\noindent\textbf{Notation.} The symbols $\complex$, $\reals$, and $\ints$ denote the sets of complex, real, and integer numbers, respectively. For $m$ a positive integer, the notation $[m]$ indicates the set $\{1, \ldots, m\}$. The terms $\lin{\spa{X}}$, $\unitary{\spa{X}}$, $\herm{\spa{X}}$, and $\pos{\spa{X}}$ denote the sets of linear, unitary, Hermitian, and positive semidefinite operators acting on complex Euclidean space $\spa{X}$, respectively. For the purposes of this survey, our choice of $\spa{X}$ will always be $\complex^d$ for some positive integer $d$. Recall that an operator $A$ is unitary if $AA^\dagger=A^\dagger A=I$, $A$ is Hermitian if $A=A^\dagger$ (where $A^\dagger$ denotes the conjugate transpose of $A$), and $A$ is positive semidefinite if $\ve{x}^\dagger A\ve{x}\geq 0$ for all complex vectors $\ve{x}$. The orthogonal projector onto vector space $\spa{X}$ is denoted $\Pi_{\spa{X}}$; by definition, $\Pi_{\spa{X}}^2=\Pi_{\spa{X}}$.  The notation $A\succeq B$ means operator $A-B$ is positive semidefinite. The smallest (largest) eigenvalue of $A\in \HH(\X)$ is given by $\lambda_{\operatorname{min}}(A)$ ($\lambda_{\operatorname{max}}(A)$). The trace, Frobenius, and spectral (or operator) norms of $A\in \LL(\spa{X})$ are defined as
\begin{eqnarray*}
 \trnorm{A}&:=&\trace\left(\sqrt{A^\dagger A}\right),\quad\quad\\
 \fnorm{A}&:=&\sqrt{\trace(A^\dagger A)},\quad\quad\\
 \snorm{A}&:=&\max_{\ket{x}\in\spa{X}\mbox{ s.t. } \enorm{x}= 1}\enorm{A\ket{x}},
\end{eqnarray*}
respectively, where $:=$ denotes a definition, and where $\trace(A):=\sum_iA(i,i)$ for the $(m,n)$th entry of $A$ given by $A(m,n)$. Unless otherwise noted, all logarithms are taken to base two.

\subsection{The four postulates of quantum mechanics}\label{ssscn:postulates}

With our notation in place, we can now provide a brief introduction to quantum information. For this, we introduce the four postulates of quantum mechanics, which intuitively prescribe the following concepts: How a quantum state is described, how one ``reads'' or measures a quantum state, what operations can be performed on a quantum state, and finally, how one describes multiple quantum systems jointly.\\

\noindent\textbf{1. Describing quantum states.}
Let $\spa{X}=\complex^d$ for $d\geq 2$. In a nutshell, a linear operator $\rho\in\LL(\spa{X})$ describes a $d$-dimensional quantum state if and only if $\rho$ is positive semi-definite and has trace $1$, i.e. $\rho\in\pos{\spa{X}}$ and $\trace(\rho)=1$. Let us now provide some intuition as to how this statement comes about.

In classical computing, the basic unit of information is a bit, which takes on values in the set $\set{0,1}$. One can equivalently encode a bit using the set $\set{\ket{0},\ket{1}}$, where $\set{\ket{0},\ket{1}}\subseteq\complex^2$ is the standard basis for $\complex^2$, i.e.\ $\ket{0}=(1,0)^T$ and $\ket{1}=(0,1)^T$. Here, the notation $\ket{\psi}$ (called a ``ket'') denotes a column vector labeled by $\psi$. The key difference between classical bits and quantum bits (or \emph{qubits}) is that in the quantum world, one can ``interpolate'' between the two discrete values $\ket{0}$ and $\ket{1}$ by taking a \emph{superposition}, i.e.\ the vector
\begin{equation}
    \ket{\psi}=\alpha\ket{0}+\beta\ket{1}
\end{equation}
describes a valid quantum state if $\abs{\alpha}^2+\abs{\beta}^2=1$. In other words, any unit vector in $\complex^2$ describes a quantum bit, or \emph{qubit}.

More generally, any unit vector $\ket{\psi}\in\complex^d$ describes a $d$-dimensional quantum state, sometimes dubbed a qu\emph{d}it. Such vectors are called \emph{pure} states, as they describe the state of a quantum system which is not subject to external ``noise'', or more generally is evolving in isolation from its environment. Sometimes, however, interaction between our system and its environment may be inevitable (such as when we wish to measure our system); this will inherently inject some ``noise'' into our system, and we thus need a mathematical approach to model this. To do so, we simply permit probabilistic mixtures of pure states, more generally referred to as \emph{mixed} states. Such probabilistic mixtures are described in the following straightforward manner, known as the \emph{density matrix} formalism.

Associated with any probabilistic mixture is an \emph{ensemble} $E$,
\begin{equation}
E=\set{\set{p_i}_{i=1}^k,\set{\ketbra{\psi_i}{\psi_i}}_{i=1}^k},
\end{equation}
where $\set{p_i}_{i=1}^k$ forms a probability distribution and $\set{\ket{\psi_i}}\subseteq\spa{X}$ is a set of unit vectors. Here, the notation $\bra{\psi}$ (called a ``bra'') denotes the row vector corresponding to the conjugate transpose of $\ket{\psi}$. Thus, if $\ket{\psi}\in \complex^d$, then $\ketbra{\psi}{\psi}$ is a rank-$1$ $d\times d$ complex matrix. Returning to our discussion, the mixed quantum state $\rho_E$ corresponding to the ensemble $E$ is given by:
\begin{equation}
    \rho_E = \sum_{i=1}^k p_i \ketbra{\psi_i}{\psi_i}.\label{0_eqn:density}
\end{equation}
Here, $\rho_E$ is called the \emph{density matrix} describing the underlying quantum system. We denote the set of density operators acting on $\X$ as $\DD(\X)$.

Let us now tie the density operator formalism back into the statement made at the beginning of this section: That $\rho$ represents a quantum state if and only if $\rho\succeq 0$ and $\trace(\rho)=1$. Note that since in Equation~\ref{0_eqn:density}, $\rho$ is a non-negative sum of positive semidefinite operators, we have $\rho\succeq 0$. Moreover, by applying the cyclic property of the trace (i.e. that $\trace(ABC)=\trace(CAB)$ for all $A,B,C$), we have
 \[
    \trace(\rho)=\sum_{i=1}^k p_i \trace(\ketbra{\psi_i}{\psi_i})=\sum_{i=1}^k p_i\braket{\psi_i}{\psi_i}=\sum_{i=1}^k p_i=1,
 \]
  where $\braket{v}{w}$ denotes the inner product between vectors $\ket{v}$ and $\ket{w}$. Conversely, we can now intuitively see why any $\rho\in\mathcal{X}$ with $\rho\succeq 0$ and $\trace(\rho)=1$ describes a valid quantum state --- recall that any positive semi-definite operator $A$ can be diagonalized via its \emph{spectral decomposition} $A=\sum_{i}\lambda_i(A)\ketbra{\lambda_i(A)}{\lambda_i(A)}$, where $\lambda_i(A)\geq 0$ are the eigenvalues of $A$, and $\ket{\lambda_i(A)}$ are the corresponding eigenvectors. Then, given $\rho$ satisfying $\rho\succeq 0$ and $\trace(\rho)=1$, taking its spectral decomposition allows us to immediately recover an ensemble $\set{\set{p_i},\set{\ketbra{\psi_i}{\psi_i}}}$. In this case, the $p_i$, which correspond to the eigenvalues of $\rho$, sum to $1$ due to the trace constraint on $\rho$.

 Finally, note that although we have attempted to present a simple exposition of how quantum states are classically described via density operators, in reality the precise interpretation of what a density operator \emph{means} is highly non-trivial and remains a subject of intense debate.\\

\noindent\textbf{2. Measuring quantum states.}
Now that we have an approach for describing quantum states, we require a formalism for modeling how a quantum state is ``observed'' or measured. For this, let $\rho\in\DD(\X)$ be a density matrix. Then, a quantum \emph{measurement} is formalized by a set of operators $\set{M_i}\subseteq\LL(\X)$ satisfying
\begin{equation}
    \sum_i M_i^\dagger M_i=I,
\end{equation}
where the latter is called the \emph{completeness relation}. The act of measuring $\rho$ with $\set{M_i}$ is in general an inherently probabilistic process, \emph{even if} $\rho$ corresponds to a pure state (unlike in the classical case of bits). Specifically, when measuring $\rho$ with respect to $\set{M_i}$, we obtain outcome $i$ with probability given by
\begin{equation}
    \pr(\mbox{outcome }i \mid \rho)=\trace(M_i\rho M_i^\dagger).
\end{equation}
Once a particular outcome $i$ is observed, the state $\rho$ ``collapses'' to a new state $\rho'$ consistent with this outcome, i.e.\
\begin{equation}
    \rho' = \frac{M_i\rho M_i^\dagger}{\pr(\mbox{outcome }i \mid \rho)}.
\end{equation}
Note that the denominator above serves the role of normalizing $\rho'$ so that $\trace(\rho')=1$. Such a measurement $\set{M_i}$ is the most general possible.

Often, however, we are interested in a special type of measurement in which each $M_i$ is an orthogonal projection operator additionally satisfying $M_iM_j=\delta_{ij}M_i$. Such measurements are called \emph{projective} or \emph{von Neumann} measurements. An example of such a measurement used frequently is a measurement in the \emph{computational} or \emph{standard} basis. For $\X=\complex^d$ and standard basis vectors $\ket{i}$ (which have a $1$ in position $i$ and zero elsewhere), such a measurement is specified via $\set{M_i}$ for $M_i=\ketbra{i}{i}$. More generally, given any orthonormal basis $B=\set{\ket{\psi_i}}$, measuring in basis $B$ is formalized by setting $M_i=\ketbra{\psi_i}{\psi_i}$.

There is a useful alternate representation of von Neumann measurements via \emph{observables} $M\in\HH(\X)$. Specifically, given an observable $M$, to recover the underlying von Neumann measurement, we take the spectral decomposition of $M$ to obtain $M=\sum_i\lambda_i\Pi_i$, where $\lambda_i\neq \lambda_j$ for $i\neq j$ and each $\Pi_i$ is a projection operator (of rank possibly greater than one). Then, the measurement operators are defined as $M_i=\Pi_i$, and each distinct eigenvalue $\lambda_i$ corresponds to a distinct label for the corresponding measurement outcome $\Pi_i$. The reason this representation via observables proves useful is because the \emph{expected value} of a measurement, denoted $\mathbb{E}_M$, takes a very simple form in terms of the observable $M$:
\begin{eqnarray*}
    \mathbb{E}_M(\rho)&=&\sum_i\lambda_i\pr(\mbox{outcome }i \mid \rho)\\&=&\sum_i\lambda_i \trace(\Pi_i\rho\Pi_i^\dagger)\\&=&\sum_i\lambda_i \trace(\Pi_i\rho)\\&=&\trace(M\rho),
\end{eqnarray*}
where the third equality follows from the cyclic property of the trace and since $\Pi$ is Hermitian and is a projector, and the last equality since the trace is a linear map.\\


\noindent\textbf{3. Evolution of quantum states.}
We now know how to describe a quantum state $\rho\in\DD(\X)$, as well how to model a measurement of $\rho$. The next question we ask is: What kind of operations (i.e. gates) can we perform on $\rho$? For example, to a classical bit, we can apply a NOT gate to flip its value. What can we do to a \emph{qu}bit?

Just as in the case of describing quantum states, there are two scenarios to consider here: When a quantum system evolves in a manner isolated from its environment (i.e. a \emph{closed} system), and when a system interacts with an environment during its evolution (i.e. an \emph{open} system). Beginning with the former, the set of valid operations on a {closed} quantum system with state $\rho\in\DD(\X)$ is the set of unitary operators $U\in\UU(\X)$. Specifically, $U$ maps $\rho$ to
\begin{equation}
    \rho' := U\rho U^\dagger.
\end{equation}

For example, for $\rho\in\DD(\complex^2)$, i.e.\ a single qubit, a frequently used set of unitary operators are the \emph{Pauli} operators (where $i:=\sqrt{-1}\in\complex$)
\begin{equation*}
    X=\left(
        \begin{array}{cc}
          0 & 1 \\
          1 & 0 \\
        \end{array}
      \right)\quad\quad\quad
    Y=\left(
        \begin{array}{cc}
          0 & -i \\
          i & 0 \\
        \end{array}
      \right)\quad\quad\quad
    Z=\left(
        \begin{array}{cc}
          1 & 0 \\
          0 & -1 \\
        \end{array}
      \right).
\end{equation*}
Note, for example, that the Pauli $X$ plays the role of a quantum NOT gate, i.e.\ $X\ket{0}=\ket{1}$ and $X\ket{1}=\ket{0}$. The Pauli operators are also referred to as $\sigma_x=X$, $\sigma_y=Y$, and $\sigma_z=Z$. In this survey, we will use the fact that the operators above can be generalized to act on $\complex^3$, i.e. on \emph{qutrits}, as
\begin{eqnarray}\label{eqn:spin1pauli}
\sigma_x&=&\frac{1}{\sqrt2}\left(
                \begin{array}{ccc}
                  0 & 1 & 0\\
                  1 & 0 & 1\\
                  0 & 1 & 0\\
                \end{array}
              \right),\\
\sigma_x&=&\frac{1}{\sqrt2}\left(
                \begin{array}{ccc}
                  0 & -i & 0\\
                  i & 0 & -i\\
                  0 & i & 0\\
                \end{array}
              \right),\\
\sigma_z&=&\left(
                \begin{array}{ccc}
                  1 & 0 & 0\\
                  0 & 0 & 0\\
                  0 & 0 & -1 \\
                \end{array}
              \right).
\end{eqnarray}
Moreover, it is common in the physics literature to define vector $\overrightarrow{\sigma_i}:=(\sigma_i^x,\sigma_i^y,\sigma_i^z)$ (sometimes also denoted $\overrightarrow{S_i}$), and dot-product
\begin{equation}\label{eqn:sigma}
\overrightarrow{\sigma_i}\cdot\overrightarrow{\sigma_j} := \sigma_i^x\sigma_{j}^x+\sigma_i^y\sigma_{j}^y+\sigma_i^z\sigma_{j}^z,
\end{equation}
where there is an implicit tensor product between (e.g.) $\sigma_i^x$ and $\sigma_{j}^x$ above. (The tensor product is defined under ``Composite quantum systems'' below.)\\

Let us next discuss the time evolution of an \emph{open} quantum system (i.e. one which interacts with its environment). In this case, the set of allowed operations strictly contains $\UU(\X)$, and is in fact the set of \emph{trace-preserving completely-positive} (TPCP) maps, which we henceforth refer to as \emph{admissible} maps or operations. Here, a linear map $\Phi:\LL(\X)\mapsto\LL(\Y)$ is \emph{trace-preserving} if $\trace(\Phi(A))=\trace(A)$ for all $A\in\LL(\X)$. To define ``completely positive'', let us first define ``positive'' --- $\Phi$ is positive if $\Phi(A)\succeq 0$ whenever $A\succeq 0$. Then, $\Phi$ is \emph{completely positive} if $\Phi$ is positive even when acting only some subsystem of some larger joint system, i.e. if $I\otimes \Phi$ is positive, where $I\in\LL(\spa{X})$. Although the notion of TPCP maps plays an important role in quantum information, we remark that there is no loss of generality in restricting oneself to the set of unitary maps. This is because any valid TPCP operation on an open quantum system $A$ can be simulated by moving to a larger closed joint system $AB$, evolving $AB$ via a unitary operator, and subsequently tracing out part of $AB$. (More on joint systems and tracing parts of them out will be said shortly.)

To further distinguish the cases of open versus closed systems, let us give a concrete example. Suppose one wishes to perform a measurement on $A$. In order to do so in a lab, one introduces a measurement apparatus, which we think of as system $B$. To complete the actual measurement, $B$ must interact with $A$, implying $A$ is an open system. Thus, if we look at $A$ alone, the action of the measurement on $A$ is not described by a unitary operator, but by a TPCP map. However, if we instead look at $AB$ as a whole, this joint system is now closed, and hence its evolution is described by a unitary operator.

\paragraph{Aside: Hamiltonians, and the connection to unitary operations.} We are now in a position to understand where the notion of a Hamiltonian comes from, which will be fundamental for this survey.

Recall that any closed quantum system evolves in time according to some unitary operation $U\in\UU(\spa{X})$. Now, any unitary $U$ can be written as $U=\exp(i H)$ for some Hermitian $H\in\HH(\X)$. To see this, let us first define the notion of an \emph{operator function}. Specifically, for any operator $A\in\LL(\X)$ admitting a spectral decomposition\footnote{Such operators $A$ are called \emph{normal}, and satisfy $AA^\dagger=A^\dagger A$.}, and any function $f:\complex\mapsto\complex$, we define $f(A):=\sum_i f(\lambda_i)\ketbra{\lambda_i}{\lambda_i}$, where $\sum_i \lambda_i\ketbra{\lambda_i}{\lambda_i}$ is the spectral decomposition of $A$. Then, to verify our claim about unitary operators, take the spectral decomposition $U=\sum_j e^{i\theta_j}\ketbra{\psi_j}{\psi_j}$ (where recall a unitary operator has its eigenvalues on the unit circle), and observe that defining
\begin{equation}
    H=\sum_j \theta_j\ketbra{\psi_j}{\psi_j}
\end{equation}
yields $U=e^{iH}$. The operator $H$ is called a \emph{Hamiltonian}.

We conclude that corresponding to each $U\in\UU(\X)$, there exists a Hamiltonian $H\in\HH(\X)$. Where does $H$ come from? It turns out that the time evolution of a closed quantum system $\ket{\psi}$ according to $H$ is given by the celebrated Schr\"{o}dinger equation\footnote{Here, we state the time-independent version of Schr\"{o}dinger's equation, and consider only the corresponding class of time-independent Hamiltonians. More generally, one can consider evolution via Hamiltonians which themselves change with time, i.e. are time-dependent.} ,
\begin{equation}
    i\hbar\frac{d\ket{\psi}}{dt}=H\ket{\psi},
\end{equation}
where $\hbar$ denotes Planck's constant. For a quantum system evolving from time $t_1$ to $t_2$, the solution to this equation is given by
\begin{equation}\label{0_eqn:schroedinger}
    \ket{\psi(t_2)}=\exp\left(-i\frac{t_2-t_1}{\hbar}H\right)\ket{\psi(t_1)},
\end{equation}
i.e. $\ket{\psi(t_1)}$ evolves to $\ket{\psi(t_2)}$ via the unitary $\exp\left(-i\frac{t_2-t_1}{\hbar}H\right)$! In other words, the notion of evolution under unitary operations arises precisely because Schr\"{o}dinger's equation maps any ``input'' Hamiltonian $H$ to a unitary of the form $e^{-itH}$.\\

\noindent\textbf{4. Composite quantum systems.} We have stated that a quantum state on $\X$ is described by a density matrix $\rho\in\DD(\X)$. Suppose now we have two quantum systems $A$ and $B$ --- how do we describe their joint state $AB$? It turns out that if $A$ and $B$ correspond to complex Euclidean spaces $\X=\complex^{d_x}$ and $\Y=\complex^{d_y}$, then the joint system $AB$ corresponds to the space $\X\otimes \Y=\complex^{d_x+d_y}$. Here, $\otimes$ denotes the \emph{tensor product}, and is discussed further shortly.

Before we do so, however, let us make an important observation. Consider a set of $n$ single-qubit systems $\set{\X_i}_{i=1}^n$. Then, the joint state of all $n$ qubits is described by a density operator acting on $\bigotimes_{i=1}^n \X_i=\complex^{2^n}$. In other words, in stark contrast to a classical system of $n$ bits, in order to describe the joint state of $n$ qubits, one requires an \emph{exponential}-size density matrix, i.e. $\rho\in\DD(\complex^{2^n})$! This was essentially the reason why Richard Feynman originally proposed~\cite{F82,F85} the concept of quantum computing (see also Benioff~\cite{B80,B82_1,B82_2}) --- he believed quantum computers might enable an efficient study of quantum systems (which is otherwise generally intractable on a classical computer).

We now further discuss the tensor product $\otimes$, as it is used throughout this survey. First, for vectors $\ve{u}\in\X$ and $\ve{y}\in\Y$, we have that $\ve{u}\otimes\ve{v}\in\X\otimes \Y$, such that for all $i\in[d_x]$ and $j\in[d_y]$
\begin{equation}
    (\ve{u}\otimes \ve{v})(i,j):=u(i)v(j).
\end{equation}
For linear operators $A\in\LL(\X)$, $B\in\LL(\Y)$, we similarly have that $A\otimes B\in\LL(\X\otimes\Y)$, such that $A\otimes B$ is a complex matrix whose index sets are given by $([d_x]\times [d_y],[d_x]\times [d_y])$ satisfying
\begin{equation}
    (A\otimes B)((i_1,j_1),(i_2,j_2)):=A(i_1,i_2)B(j_1,j_2)
\end{equation}
for all $i_1,i_2\in[d_x]$ and $j_1,j_2\in[d_y]$. The tensor product has the following properties for any $A,C\in \X$, $B,D\in\Y$, $c\in\complex$:
\begin{eqnarray}
    (A+C)\otimes B &=& A\otimes B + C\otimes B\\
    A\otimes (B+D) &=& A\otimes B + A\otimes D\\
    c(A\otimes B) &=& (cA)\otimes B = A\otimes (cB)\\
    (A\otimes B)(C\otimes D) &=& AC\otimes BD\\
    \trace(A\otimes B) &=& \trace(A)\trace(B)\\
    (A\otimes B)^\dagger &=& A^\dagger\otimes B^\dagger.
\end{eqnarray}
These properties hold analogously in the vector setting.



Finally, given a description $\rho$ of the state of a joint system $AB$, we now require a method for describing the marginal state on $A$ (or $B$) alone. Specifically, given a composite system $\rho\in\DD(\X\otimes\Y)$, the \emph{reduced state} $\rho_A$ on $A$ (analogously, $\rho_B$ on $B$) is prescribed via the linear map known as the \emph{partial trace}, i.e.
\begin{equation}
    \rho_A = \trace_B(\rho).
\end{equation}
The partial trace is defined as follows: For any $A\otimes B\in \LL(\X\otimes \Y)$, we have that $\trace_\X (A\otimes B)\in\Y$ such that
\begin{equation}
    \trace_\X (A\otimes B):=\trace(A)B.
\end{equation}
A second equivalent definition is as follows: For any orthonormal basis $\set{\ve{v}_i}_{i=1}^d$ for $\X$, we have that for any $C\in\LL(\X\otimes \Y)$
\begin{equation}
    \trace_\X(C) = \sum_{i=1}^d \left(\ve{v}_i^\dagger \otimes I\right) C \left(\ve{v}_i \otimes I\right).
\end{equation}
For example, $\trace_B(\rho_A\otimes\rho_B)$ is simply $\rho_A$, and $\trace_B(\ketbra{\phi^+}{\phi^+})=I/2$, where
\[
	\ket{\phi^+}=\frac{1}{\sqrt{2}}\ket{00}+\frac{1}{\sqrt{2}}\ket{11}\in \complex^2\otimes\complex^2
\]
is a special two-qubit state known as a \emph{Bell state}.

Having introduced the concept of partial trace, we arrive at another fundamental point about joint quantum systems: The possibility of exotic \emph{correlations} between a pair of quantum systems $A$ and $B$. On one end of the spectrum lies the notion of a \emph{tensor product} state, i.e. a state of the form $\rho_A\otimes\rho_B$. Such states are \emph{uncorrelated}, and as such, tracing out system $B$ simply yields $\rho_A$. At the other extreme lie states exhibiting a uniquely quantum form of correlations known as \emph{quantum entanglement}. A canonical example of such a state is the Bell state $\ket{\phi^+}$ from above --- in contrast to a tensor product state, it is \emph{impossible} to express $\ket{\phi^+}\in\X\otimes\Y$ as the tensor product of a pair of states $\ket{\psi_1}\in\X$ and $\ket{\psi_2}\in\Y$! The consequence of this is that even though the joint state $\ket{\phi^+}$ of both qubits is fully known (i.e. is pure), tracing out $B$ yields a state which is noisy or mixed. (In fact, in this case $\trace_B(\ketbra{\phi^+}{\phi^+})=I/2$ is \emph{maximally mixed}, in that it yields \emph{no} information about the state of $A$.) Entanglement is one of the key characteristics distinguishing the quantum world from the classical one~\cite{S35}, and is in fact a necessary resource for (pure state) quantum computation to exponentially outperform classical computation~\cite{jozsa03a}.

Delving into entanglement further, it is useful to note that any bipartite pure state $\ket{\psi_{AB}}\in\complex^{d_x}\otimes \complex^{d_y}$ can be written in terms of its \emph{Schmidt decomposition}, such that
\begin{equation}
    \ket{\psi_{AB}}=\sum_{i=1}^{\min(d_x,d_y)}\alpha_i\ket{\psi_i}\otimes\ket{\phi_i}.
\end{equation}
Here, the real $\alpha_i\geq 0$ are called \emph{Schmidt coefficients}, and the sets $\set{\ket{\psi_i}}$ and $\set{\ket{\phi_i}}$ are orthonormal bases for $\complex^{d_x}$ and $\complex^{d_y}$, respectively, known as the \emph{Schmidt bases}. The connection to entanglement is simple in the pure state case: The state $\ket{\psi}$ is entangled if and only if it has at least two non-zero Schmidt coefficients $\alpha_i$. For this reason, the number of non-zero $\alpha_i$ has a special name --- it is called the \emph{Schmidt rank} of $\ket{\psi}$. Thus, we say that $\ket{\psi}$ is entangled if and only if it has Schmidt rank at least $2$. But we can do more than simply state whether $\ket{\psi}$ is entangled or not; in the pure state case, we can also quantify the \emph{amount} of entanglement, achieved via the \emph{entropy of entanglement}:
\begin{equation}\label{eqn:eentropy}
	S(\trace_B(\ketbra{\psi}{\psi}))=H(\set{\alpha_i}),
\end{equation}
where $S:\DD(\X)\mapsto\reals$ is the \emph{von Neumann entropy} function defined as $S(\rho)=-\trace(\rho \log(\rho))$ (recall that here we treat $\log$ as an operator function, i.e. it is applied to the eigenvalues of $\rho$, and we define $0\cdot\log 0=0$), and $H$ denotes the classical Shannon entropy of a probability distribution defined by $H(\set{p_i})=-\sum_i p_i\log p_i$. Notions such as the Schmidt rank and entropy of entanglement will play important roles in our discussions on matrix product states (\S\ref{sscn:MPS}) and area laws (\S\ref{sscn:arealaw}).

We close this discussion with two final remarks. First, the partial trace is employed here as it is the unique function which correctly produces the measurement statistics for arbitrary observables $M$ measured on $A$ alone. Second, one can also define a notion of entanglement for mixed quantum states (which will not be relevant to this survey). Unlike the case of pure states, however, determining whether a mixed state $\rho\in\DD(\X\otimes Y)$ is entangled is highly non-trivial; in fact, it is strongly NP-hard~\cite{G03,G10}.\\

\section{Basics of Quantum Hamiltonian Complexity}\label{sscn:basicsqhc}

With our background on quantum information in place, we can now introduce some fundamental concepts in Quantum Hamiltonian Complexity (QHC). We begin with a complexity class which has played a central role in the development of the field.

Specifically, the class NP can be generalized to a bounded-error quantum variant known as \emph{Quantum-Merlin Arthur (QMA)}. The intuition behind QMA is analogous to NP, except we now replace the classical prover and verifier by a quantum prover and verifier, respectively, and stipulate that given a quantum proof $\ket{\psi}$ from the prover, the verifier is allowed to err with bounded probability at most (say) $1/3$ as to whether the input instance is a YES or NO instance. More formally, we have the following definition, where a quantum circuit is simply the quantum analogue of a classical circuit in which classical gates are replaced with quantum ones such as the Pauli $X$, $Y$, and $Z$ gates.

\begin{definition}[QMA]\label{0_def:QMA}
    A promise problem $A=(\ayes,\ano)$ is in QMA if and only if there exist polynomials $p$, $q$ and a polynomial-time uniform family of quantum circuits $\set{Q_n}$, where $Q_n$ takes as input a string $x\in\Sigma^*$ with $\abs{x}=n$, a quantum proof $\ket{y}\in (\complex^2)^{\otimes p(n)}$, and $q(n)$ ancilla qubits in state $\ket{0}^{\otimes q(n)}$, such that:
    \begin{itemize}
    \item (Completeness) If $x\in\ayes$, then there exists a proof $\ket{y}\in (\complex^2)^{\otimes p(n)}$ such that $Q_n$ accepts $(x,\ket{y})$ with probability at least $2/3$.
    \item (Soundness) If $x\in\ano$, then for all proofs $\ket{y}\in (\complex^2)^{\otimes p(n)}$, $Q_n$ accepts $(x,\ket{y})$ with probability at most $1/3$.
    \end{itemize}
\end{definition}

Here, the quantum proof is given by $\ket{y}$, and the verification circuit by the circuit family $\set{Q_n}$. We remark that the completeness and soundness parameters of $(2/3,1/3)$ can be made exponentially close to $1$ and $0$ in the input size, respectively, via two approaches. The first approach is to apply the verifier's circuit many times in parallel to many copies of the proof $\ket{y}$. The disadvantage of this technique is that it increases the proof size. (Note that it is not entirely trivial that this approach should work; namely, the prover could try to cheat as follows. Instead of sending many copies of $\ket{y}$ in tensor product, it could send some arbitrary entangled state across all proof registers. A simple analysis~\cite{AN02} shows, however, that this type of parallel repetition indeed correctly reduces the error.) The second approach for error reduction is due to Marriot and Watrous~\cite{MW05}, who give a non-trivial procedure for exponentially reducing the error \emph{without} increasing the proof size (albeit at the expense of increasing the verification circuit's size).

With QMA defined, we now discuss a canonical QMA-complete problem which plays a role analogous to SATISFIABILITY for NP~\cite{C72,L73}: The Local Hamiltonian Problem. To define the latter, recall that a $k$-local Hamiltonian acting on $n$ qudits is a Hermitian operator $H=\sum_i H_i$, where each $H_i$ acts non-trivially on $k$ qudits (i.e. to be more accurate, if $H_i$ acts on a set $S$ of qudits, then the $i$th constraint $H_i$ should actually be written as $(H_i)_S\otimes I_{[n]\backslash S}$). The motivation for this problem stems partly from the fact that, as discussed in \S\ref{ssscn:postulates}, Hamiltonians are intricately tied to the time evolution of quantum systems. Moreover, in nature, such evolution is typically governed by a \emph{local} Hamiltonian.

\begin{problem}[$k$-Local Hamiltonian Problem ($\klhh$)~\cite{KSV02}] \label{0_def:localhamiltonianproblem}Given as input a $k$-local Hamiltonian $H$ acting on $n$ qudits, specified as a collection of constraints $\set{H_i}_{i=1}^r\subseteq \herm{\complex^d}^{\otimes k}$ where $k,d\in \Theta(1)$, and threshold parameters $a,b\in\reals$, such that $0\leq a < b$ and $(b-a)\geq 1$, decide, with respect to the complexity measure $\enc{H} + \enc{a} + \enc{b}$:
    \begin{enumerate}
        \item If $\lmin{H}\leq a$, output YES.
        \item If $\lmin{H} \geq b$, output NO.
    \end{enumerate}
\end{problem}
\noindent Here, $\enc{A}$ denotes the encoding length of object $A$ in bits, and $\lmin{A}$ denotes the smallest eigenvalue of $A$. The value $\lmin{A}$ is the \emph{ground state energy} of $H$, and its corresponding eigenvector (or eigenspace) is known as the \emph{ground state} (or \emph{ground space}). Note that often $\klhh$ is phrased with $(b-a)\geq 1/p(n)$ for some polynomial $p$; such an inverse polynomial gap can straightforwardly be boosted to the constant $1$ above by defining $H$ to have $p(n)$ many copies of each local term $H_j$~\cite{W09_2}.

To highlight the connection between $\klhh$ and classical satisfiability problems, let us demonstrate how an instance of $k$-CSP can be embedded into $\klhh$. Specifically, let $\phi$ denote an instance of $3$-CSP with clauses $c_i$ which are arbitrary Boolean functions on $3$ bits. Then, corresponding to each clause $c_i$, we add to our Hamiltonian $H$ a diagonal local constraint $H_i\in \herm{\complex^2}^{\otimes 3}$ which penalizes all non-satisfying assignments, i.e.
\[
    H_i=\sum_{\substack{x\in\set{0,1}^3\\\text{s.t. }c_i(x)=0}}\ketbra{x}{x}.
\]
In other words, suppose our assignment is $\ket{x}$ for $x\in\set{0,1}^n$, such that $c_i(x)=0$. Then, we have $\trace(H_i \ketbra{x}{x})=1$. Conversely, if $c_i(x)=1$, then $\trace(H_i \ketbra{x}{x})=0$. We conclude that a product state $\ket{x}\in (\complex^2)^{\otimes n}$ representing a satisfying binary assignment to $\phi$ will achieve energy $0$ on $H=\sum_i H_i$, i.e.
\[
    \trace(H\ketbra{x}{x})=0.
\]
On the other hand, if $\phi$ is unsatisfiable, then for all $\ket{x}$ with $x\in\set{0,1}^n$, we have $\trace(H\ketbra{x}{x})\geq 1$. In fact, we can conclude the same property holds for \emph{any} $\ket{\psi}\in(\complex^2)^{\otimes n}$ (i.e. not just for binary string assignments); this is because all $H_i$ simultaneously diagonalize in the standard basis, and thus without loss of generality, one can choose the ground state as a binary string. We conclude that $\klhh$ is a generalization of $k$-CSP, and is thus at least NP-hard. Of course, we know that $\klhh$ is expected to be much harder; in \S\ref{sscn:5LH} and \S\ref{sscn:2LH}, we discuss the proofs of QMA-completeness of $5$-LH and $2$-LH, respectively.

\paragraph{The Simulation Problem.} Much of this survey focuses on the Local Hamiltonian Problem, and as such, it is important to place $\klhh$ into the context of QHC as a whole. It turns out that $\klhh$ is a special case of a more general problem capturing the essence of QHC, known as the Simulation Problem~\cite{O11}. Intuitively, the latter asks how difficult it is to simulate a physical system. More formally, in the Simulation Problem one is given as input a description of a Hamiltonian $H$, an initial state $\rho$, an observable $M$, and a time $t\in\complex$, and the task is to output an estimate of the expectation
\begin{equation}\label{eqn:sim}
    \trace\left[M \frac{(e^{iHt})^\dagger \rho e^{iHt}}{\trace\left((e^{iHt})^\dagger \rho e^{iHt}\right)}\right].
\end{equation}
Note here that the denominator, $\trace\left((e^{iHt})^\dagger \rho e^{iHt}\right)$, is not redundant, since for $t\in \complex\setminus\reals$, $e^{iHt}$ is not necessarily unitary. Indeed, this fact is crucial to recovering the local Hamiltonian problem, which is obtained by choosing $H$ as a local Hamiltonian, setting $M=H$, $\rho=I/\trace(I)$, and considering $t=i\beta$ for $\beta\in\reals$. Then, Equation~\ref{eqn:sim} reduces to
\[
        \trace\left[H \frac{e^{-2\beta H}}{\trace\left(e^{-2\beta H}\right)}\right].
\]
Taking the limit $\beta\rightarrow+\infty$, the expression $(e^{-2\beta H})/\trace(e^{-2\beta H})$ approaches the projector onto the ground space of $H$ (where here, the normalization given by the denominator is crucial to obtain a norm $1$ operator), as desired.

Finally let us make an aside: In order to recover the local Hamiltonian problem above, we chose a \emph{complex} time $t$. There is another sense in which complex times $t$ are physically important: For complex $t$ which is neither purely real or imaginary, the simulation problem addresses the scenario in which a system dynamically evolves at \emph{finite} temperature~\cite{O11}. The latter is particularly natural, as in a lab it is often impossible to cool a quantum system all the way down to its actual ground state.\\

\chapter{Roadmap and Organization}\label{scn:roadmap}

With our primer on quantum information (\S\ref{sscn:basicsqi}) and QHC (\S\ref{sscn:basicsqhc}) covered, we are ready to lay out the roadmap for the remainder of this survey. QHC has evolved into a field with numerous areas of study, some of the most fundamental of which we shall attempt to cover here. Our first step in this journey is to present the reader with a brief history of the field from both computer science (\S\ref{sscn:cshistory}) and physics (\S\ref{sscn:physicshistory}) perspectives.

The remainder of this survey can then be thought of as consisting of two parts: The first half (\S\ref{scn:motivation} and~\S\ref{scn:physicsforCS}) explains concepts from condensed matter physics in a computer science-friendly language, and the second half (Chapter~\ref{scn:results}) discusses selected computer science-inspired results in the field, beginning with Kitaev's celebrated proof that the LH problem is QMA-complete. Let us elaborate on these two parts briefly here.

We begin in the first part by asking one of the most important questions in any field of study: \emph{Why?} Why is the topic of study interesting or useful? What are the connections between the formal model being studied and the underlying reality it is intended to represent? This is precisely the purpose of Chapter~\ref{scn:motivation}. In particular, here we introduce the motivations for QHC from a physics perspective, explaining key questions of interest such as the study of time evolution versus thermal equilibrium, the origin of local Hamiltonians, and fundamental concepts such as indistinguishable particles (i.e. bosons and fermions).

Chapter~\ref{scn:physicsforCS} then elaborates further on these physics ideas. It begins with a glossary in \S\ref{sscn:terms} of common terms in the physics literature (such as frequently studied local Hamiltonian models including the Ising and Heisenberg models). We then discuss selected significant physics-based contributions to QHC (some of which predate QHC by decades), such as mean field theory (\S\ref{sscn:meanfield}), tensor networks (\S\ref{sscn:tensor}), Density Matrix Renormalization Group (\S\ref{sscn:DMRG}), and Multi-Scale Entanglement Renormalization Ansatz (\S\ref{sscn:MERA}). These are techniques used to classically represent, simulate, and compute properties of quantum systems occurring in nature. We close this section with an overview of Area Laws (\S\ref{sscn:arealaw_overview}), which are a set of conjectures and theorems about the structure of entanglement present in ground states of physically relevant quantum systems.

Moving to the second part of the survey, in Chapter~\ref{scn:results}, we review selected computer science-inspired results in QHC, beginning with Kitaev's proof that $5$-LH is QMA-complete (\S\ref{sscn:5LH}). This is arguably the founding work of QHC. We then discuss Kempe, Kitaev, and Regev's use of perturbation theory gadgets to show that even $2$-LH is QMA-complete (\S\ref{sscn:2LH}); their techniques are powerful, and have since found numerous uses  showing hardness results for physically motivated local Hamiltonian models. We next consider Hamiltonian models which \emph{a priori} seem ``more classical'', namely those with commuting local constraints. In this direction, we discuss Bravyi and Vyalyi's results that the commuting variant of $2$-LH is in NP, and thus unlikely to be QMA-complete (\S\ref{sscn:CLH}). This is interesting given that commuting Hamiltonians can nevertheless have ground states demonstrating exotic forms of entanglement. Moving on, we give a new information-theoretic presentation of Bravyi's polynomial time algorithm for Quantum 2-SAT, which we hope makes the algorithm more accessible to a computer science audience. Finally, we review Arad, Kitaev, Landau and Vazirani's combinatorial proof of a 1D area law for gapped systems (\S\ref{sscn:arealaw}), which compliments and strengthens Hastings' original physics-inspired proof.

\chapter{A Brief History}\label{scn:history}
The history of quantum Hamiltonian complexity has its roots in both physics and computer science. In this section, we give a brief survey of both perspectives, beginning with the latter in \S\ref{sscn:cshistory}, and following with the former in \S\ref{sscn:physicshistory}. In \S\ref{sscn:recenthistory}, we briefly list selected recent developments to both perspectives since this survey first appeared.

\section{The computer science perspective}\label{sscn:cshistory}

\textbf{The general LH problem.} In 1999, Alexei Kitaev presented~\cite{K99,KSV02} what is regarded as the quantum analogue of the Cook-Levin theorem, proving that $\klhh$ is in QMA for $k\geq 1$ and QMA-hard for $k\geq 5$. His proof is based on a clever combination of the ideas behind the Cook-Levin theorem and early ideas for a quantum computer of Feynman~\cite{F85}, and is surveyed in \S\ref{sscn:5LH}. The fact that $3$-$\lh$ is also QMA-complete was shown subsequently by Kempe and Regev~\cite{KR03} (an alternate proof was later given by Nagaj and Mozes~\cite{NM06}). Kempe, Kitaev, and Regev then showed~\cite{KKR06} that $2$-$\lh$ is QMA-complete; see \S\ref{sscn:2LH} for an exposition of the proof. Note that $1$-$\lh$ is in P, since one can simply optimize for each $1$-local term independently.

From a physicist's perspective, however, the Hamiltonians involved in the QMA-hardness reductions above are arguably not ``physical'', i.e. occurring in nature. To address this, Oliveira and Terhal next showed~\cite{OT05} that $\lh$ with the Hamiltonians restricted to nearest-neighbor interactions on a 2D grid is still QMA-complete. Biamonte and Love proved~\cite{BL08} that the XY model with certain linear terms is QMA-complete, and Schuch and Verstraete~\cite{SV09} showed a similar result for the Heisenberg model (see Chapter~\ref{scn:physicsforCS} for definitions of these models). Next, in stark contrast to the classical case of MAX-2-CSP on the line (which is in P), Aharonov, Gottesman, Irani and Kempe~\cite{AGIK09} showed that $2$-LH with nearest-neighbor interactions on the line is also QMA-complete if the local systems have dimension at least $12$. (Note: Reference \cite{HNN13} later pointed out a small error in~\cite{AGIK09}, and argued that the error can be fixed using ideas from~\cite{AGIK09}, but at the cost of increasing the local dimension from $12$ to $13$.) The latter was improved to $11$~\cite{N08} and subsequently to $8$~\cite{HNN13}. Gottesman and Irani~\cite{GI09} obtained related results for translationally invariant 1D systems; see also Kay~\cite{K07} for results regarding the latter setting.

Finally, very recently, Cubitt and Montanaro established~\cite{CM13} a quantum variant of Schaefer's Dichotomy Theorem~\cite{S78} for the setting of $2$-LH on qubits, classifying the complexity of a very general version of LH based on which set of $2$-qubit quantum constraints one incorporates in the constraint system. Their classification contains the following levels: Problems are either in P, NP-complete, TIM-complete, or QMA-complete, where TIM is defined~\cite{CM13} as the set of problems which are polynomial-time equivalent to solving the general Ising model with transverse magnetic fields. Note, however, that due to the perturbation theory techniques employed in~\cite{CM13}, the use of large and possibly negative weights on local constraints is required; for this reason, the results of~\cite{CM13} do not capture the complexity of certain physical models such as the Heisenberg anti-ferromagnet.\\

\noindent\textbf{Quantum SAT.} To be precise, the LH problem does not generalize $k$-CSP, but rather (the decision version of) its optimization variant MAX-$k$-CSP. One can ask how the complexity of LH changes if we instead focus on a restricted version intended to generalize $k$-CSP. In this direction, in 2006 Bravyi~\cite{B06} defined Quantum $k$-SAT ($k$-QSAT), in which all local constraints are positive semidefinite, and the question is whether the ground state energy is zero (in this case, $H$ is called \emph{frustration-free}, in that the optimal assignment lies in the null space of every interaction term), or bounded away from zero (i.e. the Hamiltonian is \emph{frustrated}). He showed that $2$-QSAT is in P (see \S\ref{sscn:Q2SAT} for an exposition), and that $k$-QSAT is ${\rm QMA}_1$-complete for $k\geq 4$, where ${\rm QMA}_1$ is QMA with perfect completeness. Recently, Gosset and Nagaj showed~\cite{GN13} that $3$-QSAT is also ${\rm QMA}_1$-complete.\\

\noindent\textbf{Stoquastic LH.} A natural special case of LH is that of \emph{stoquastic} local Hamiltonians, in which the local constraints have only non-positive off-diagonal matrix elements in the computational basis. This class of LH does not suffer from the so-called ``sign problem'' in quantum Monte Carlo simulations (quantum Monte Carlo is not a quantum algorithm; it is just a classical Monte Carlo algorithm applied to quantum systems) and is thus heuristically expected to be easier than the general LH problem. Indeed, the Stoquastic $k$-SAT problem, defined as the stoquastic variant of Quantum $k$-SAT, was shown to be in Merlin-Arthur (MA) for $k\geq 1$ and MA-complete for $k\geq 6$ by Bravyi, Bessen, and Terhal~\cite{BBT06} and Bravyi and Terhal~\cite{BT09}. (Incidentally, this was the first non-trivial example of an MA-complete promise problem.) The problem Stoquastic LH-MIN, defined as $\klhh$ with stoquastic Hamiltonians, was shown to be contained in AM~\cite{BDOT08} and complete for the class StoqMA~\cite{BBT06} for $k\geq 2$. Here, StoqMA is a variant of QMA in which the verifier is restricted to preparing qubits in the states $\ket{0}$ and $\ket{+}$, performing classical reversible gates, and measuring in the Hadamard (i.e.\ $\ket{+},\ket{-}$) basis.  Finally, Jordan, Gosset, and Love showed that computing the \emph{largest} eigenvalue of a stoquastic local Hamiltonian is QMA-complete~\cite{JGL10}.\\

\noindent\textbf{Commuting LH.} Unlike classical constraint satisfaction problems, quantum constraints do not necessarily pairwise commute. It is thus natural to ask how crucial this non-commuting property is to the QMA-completeness of LH. In this direction, Bravyi and Vyalyi showed~\cite{BV05} that commuting $2$-LH on qudits is in NP. Aharonov and Eldar subsequently showed~\cite{AE11} that $3$-LH Hamiltonian on \emph{qubits} is in NP, as well as for qutrits on ``nearly Euclidean'' interaction graphs. Schuch showed that $4$-LH on qubits arranged in a square lattice is in NP~\cite{s11}. Hastings proved that special cases of $4$-LH on $d$-dimensional qudits on a square lattice is in NP~\cite{H12_2}. Aharonov and Eldar proved that approximating the ground state energy of commuting local Hamiltonians on good locally-expanding graphs within an additive error of $O(\epsilon)$ is in NP~\cite{AE13_1,AE13_2}. Finally, Gharibian, Landau, Shin, and Wang showed~\cite{GLSW14} that the commuting variant of the Stoquastic $k$-SAT problem is in NP for logarithmic $k$ and any constant $d$.

In terms of efficiently solvable variants of commuting LH, Yan and Bacon~\cite{YB12} showed that the special case in which all commuting terms are products of Pauli operators is in P. Aharonov, Arad, and Irani~\cite{AAI10} and Schuch and Cirac~\cite{SC10} showed that commuting LH in 1D can be solved efficiently by dynamic programming.\\

\noindent\textbf{Bosons and Fermions.} Approximating the ground state energy of Hamiltonians acting on indistinguishable particles is also QMA-hard, as shown by Liu, Christandl and Verstraete for fermions~\cite{LCV07} and Wei, Mosca, and Nayak for bosons~\cite{WMN10}. Schuch and Verstraete~\cite{SV09} showed QMA-hardness for the Hubbard model, whereas very recently, Childs, Gosset, and Webb showed QMA-hardness for the Bose-Hubbard model~\cite{CGW13}. See Chapter~\ref{scn:motivation} for more on indistinguishable particles.\\

\noindent\textbf{Approximation algorithms for LH.} Given the prevalence of heuristic algorithms for solving $\klhh$, a natural question is whether rigorous (classical) approximation algorithms for $\klhh$ can be derived. Here, Bansal, Bravyi and Terhal showed~\cite{BBT09} that $k$-LH on bounded degree planar graphs, as well as on the unbounded degree star graph, can be approximated within $\epsilon$ relative error for any $\epsilon\in\Theta(1)$ in polynomial time, i.e. they gave a Polynomial Time Approximation Scheme (PTAS). Gharibian and Kempe next gave~\cite{GK11} a PTAS for approximating the best \emph{product-state} solution for $\klhh$ on dense interaction graphs, and showed that product state solutions yield a $(d^{1-k})$-approximation to the optimal solution for \emph{arbitrary} (i.e. even non-dense) interaction graphs on $d$-dimensional systems. Based on numerical evidence, they conjectured\footnote{Private communication with F. Brand\~{a}o.} that for dense graphs, a quantum de Finetti theorem without symmetry holds, which can in turn be exploited to yield a PTAS for dense $\klhh$ (as suggested by mean-field theory folklore). Indeed, Brand\~{a}o and Harrow  proved~\cite{BH13} this conjecture, along with other results: A PTAS for planar graphs (improving on~\cite{BBT09}), and an efficient approximation algorithm for graphs of low threshold rank. Finally, Bravyi has recently given~\cite{B14} a Fully Polynomial Randomized Approximation Scheme (FPRAS) for approximating the partition function of the transverse field Ising model (which in turn implies an efficient approximation algorithm for determining the ground state energy of the model).\\

\noindent\textbf{Hardness of approximation for LH.} The PCP Theorem~\cite{AS98,ALMSS98} is one of the crowning achievements of modern complexity theory. As such, a major open question in quantum Hamiltonian complexity is whether a \emph{quantum} version of this theorem holds~\cite{AN02,A06}. Rigorously formulated in the work of Aharonov, Arad, Landau and Vazirani~\cite{AALV09}, the question has attracted much attention in the last decade. For example, Reference~\cite{AALV09} proved that a quantum analogue of Dinur's gap amplification step in her proof of the PCP theorem~\cite{D07} can be shown in the quantum setting. For further details on the quantum PCP conjecture, we refer the reader to the recent survey dedicated to the topic by Aharonov, Arad, and Vidick~\cite{AAV13}.

More generally, in terms of hardness of approximation for quantum complexity classes, Gharibian and Kempe~\cite{GK12} defined a quantum version of $\Sigma_2^p$ (the second level of the Polynomial Time Hierarchy), and showed hardness of approximation for various local Hamiltonian-related problems such as Quantum Succinct Set Cover. Reference~\cite{GK12} also showed a hardness of approximation and completeness result for QCMA, which is defined as QMA with a \emph{classical} prover~\cite{AN02}. (QCMA is also known by the name Merlin-Quantum-Arthur (MQA)~\cite{W09_2}.)

\section{The physics perspective}\label{sscn:physicshistory}

We now briefly describe the history of quantum Hamiltonian complexity from a physics perspective. This section is by no means comprehensive; the reader is referred to the surveys of Verstraete, Murg, and Cirac~\cite{VMC08} and Osborne~\cite{O11}, for example, or to their friendly neighborhood physicist for further details.\\

\noindent\textbf{Classical Hamiltonians.} Beginning with the case of classical Hamiltonians, an early and canonical example of computational hardness is Baharona's work~\cite{B82}, which showed that finding a ground state and computing the magnetic partition function of an Ising spin glass in a nonuniform magnetic field are NP-hard tasks. Jerrum and Sinclair, on the other hand, showed~\cite{JS93} \#P-completeness of computing the partition function of the ferromagnetic Ising model, and gave a Fully Polynomial Randomized Approximation Scheme (FPRAS) in the same paper. As implied by the title of this paragraph, note that the Ising model is \emph{classical} in that all variables are assigned values in the set $\set{+1,-1}$; thus, the ground state has an efficient classical description.\\

\noindent\textbf{Quantum Hamiltonians.} In contrast, for \emph{quantum} Hamiltonians, the last two decades have seen much effort towards classifying when a ground state can be described efficiently classically~\cite{VMC08}. One of the main instigators of this push was White's celebrated Density Matrix Renormalization Group\footnote{Note: The word \emph{Group} does not actually refer to a group in the usual mathematical sense here.} (DMRG) method~\cite{W92,W93}, which is a heuristic algorithm performing remarkably well in practice for finding ground states of 1D quantum systems.  It was later realized~\cite{OR95,RO97,VPC04, VMC08, WVSCD09} that DMRG can be viewed as a variational algorithm over the class of tensor network states known as Matrix Product States (MPS).

More generally, tensor network states have a long history in physics, appearing in works as early as 1941~\cite{KW41} (see~\cite{NHOH99} for a survey). As 1D tensor networks, MPS \cite{PVWC07} are arguably the most basic form of such states. They have been used in the last decade, for example, by Vidal~\cite{V03,V04} to efficiently classically simulate ``slightly entangled'' quantum computation (or quantum evolution). Generalizations of MPS (e.g., to higher dimensions) have also been proposed, such as the Projected Entangled Pair States (PEPS) of Verstraete and Cirac~\cite{VC04,VWPC06}, and the Multiscale Entanglement Renormalization Ansatz (MERA) of Vidal~\cite{V07,V08}. Note that while MPS and MERA networks can be efficiently contracted, Schuch, Wolf, Verstraete and Cirac have shown that contracting a PEPS network is in general $\#$P-complete~\cite{SWVC07}. More recently, Gharibian, Landau, Shin, and Wang showed~\cite{GLSW14} that even the basic task of determining whether an arbitrary tensor network represents a non-zero vector is not in the Polynomial-Time Hierarchy~\cite{MS72} unless the hierarchy collapses.

Continuing with applications of tensor networks in the study of ground spaces, Hastings showed~\cite{Ha07} in 2007 that the ground state of gapped 1D Hamiltonians can be well approximated by an MPS with polynomial bond dimension; this helped explain the effectiveness of DMRG. However, DMRG is a heuristic, and a rigorous proof that such a MPS can be found in polynomial time required further work. In this direction, Aharonov, Arad, and Irani~\cite{AAI10} and Schuch and Cirac~\cite{SC10} showed that given a fixed bond dimension as input, the optimal MPS of that bond dimension can be found efficiently via dynamic programming. Note that their algorithm does not require the 1D Hamiltonian to have a spectral gap. For gapped systems, Arad, Kitaev, Landau and Vazirani subsequently showed~\cite{AKLV13} that an MPS with sublinear bond dimension suffices to approximate the ground state; combined with the algorithms of~\cite{AAI10,SC10}, this yielded a subexponential time algorithm for gapped systems. Finally, Landau, Vazirani, and Vidick showed~\cite{LVV13} that the problem of finding an approximation to the ground state of a 1D gapped system is in BPP.\\

\noindent\textbf{Area laws.} A key problem in quantum many-body physics is understanding the entanglement structure of a ground state. Here, a specific question which has attracted much attention is the possible existence of area laws. Roughly, an \emph{area law} says that for any subset $S$ of particles chosen from an $n$-particle ground state, the amount of entanglement crossing the cut between $S$ and its complement scales not with the size of $S$, but rather with the size of the \emph{boundary} of $S$. In this direction, a breakthrough result was Hastings' proof~\cite{Ha07} of an area law for gapped $1$D systems. A combinatorial proof improving on Hastings' result for the frustration-free case was later given by Aharonov, Arad, Landau and Vazirani~\cite{AALV11, ALV12}, followed by a proof of Arad, Kitaev, Landau, and Vazirani's~\cite{AKLV13} which applies in frustrated settings as well. Whether a 2D area law holds remains a challenging open question. This subject will be treated in more detail in \S\ref{sscn:arealaw_overview} and \S\ref{sscn:arealaw}. 

\section{Selected recent developments}\label{sscn:recenthistory}

The field of Quantum Hamiltonian Complexity is rapidly evolving, with a number of developments having taken place since this survey first appeared. We briefly list a selected number of these developments below.

Beginning with computer science-oriented results, Cubitt, Perez-Garcia and Wolf showed~\cite{CPW15} that determining whether a translationally-invariant, nearest-neighbor Hamiltonian on a 2D square lattice (with constant local dimension) is gapped is undecidable. Bravyi and Hastings showed~\cite{BH14} that estimating the ground state energy of the Transverse field Ising Model on degree-3 graphs is StoqMA-complete, thus completing the complexity classification of Cubitt and Montanaro~\cite{CM13}. Gosset, Terhal, and Vershynina showed~\cite{GTV15} how to perform universal adiabatic quantum computation using the space-time circuit-to-Hamiltonian construction. Fitzsimons and Vidick gave~\cite{FV14} a multiprover interactive proof system for the Local Hamiltonian problem involving a constant number of entangled provers. Chubb and Flammia \cite{CF15} extended the works of Landau, Vazirani and Vidick~\cite{LVV13} and Huang~\cite{Hua14} to give a polynomial-time algorithm for approximating ground space projectors of gapped 1D Hamiltonians with degenerate ground spaces. Gharibian and Sikora showed~\cite{GS14} that the following problem, motivated by quantum memories, is QCMA-complete: Given a local Hamiltonian $H$ and two ground states $\ket{\psi_1}$ and $\ket{\psi_2}$ of $H$, is there a sequence of local unitaries mapping $\ket{\psi_1}$ to $\ket{\psi_2}$ ``through'' the ground space of $H$? Arad, Santha, Sundaram and Zhang~\cite{ASSZ15} and de Beaudrap and Gharibian~\cite{dBG15} independently gave linear time classical algorithms for quantum $2$-SAT, improving on Bravyi's quartic time algorithm~\cite{B06}.

In the physics direction, Ge and Eisert showed~\cite{GE14} that in two and higher dimensions, it is not in general true that an area law for the Renyi entanglement entropy implies the ability to faithfully describe a quantum many-body state by an efficient tensor network. Aharonov \emph{et al.} showed~\cite{AHLNSV14} that a \emph{generalized} version of the area law fails to hold. Movassagh and Shor showed~\cite{MS14} that a generalization of Bravyi \emph{et al.}'s~\cite{BCMNS12} spin-$1$ model to integer spin-$s$ chains ($s>1$) yields a power law violation of the area law (the energy gap in these models is an inverse polynomial in the system size). Brand\~{a}o and Cramer showed~\cite{BC14} that exponential decay in the specific heat capacity at low temperatures yields an area law (up to a logarithmic correction) for low-energy states (i.e., not just the ground state). Mari\"{e}n \emph{et al.} proved~\cite{MAAV14} the stability of the area law for entanglement entropy in quantum spin systems in the setting of adiabatic and quasi-adiabatic evolutions.

\chapter{Motivations From Physics}\label{scn:motivation}

We have thus far defined $k$-local Hamiltonians and discussed a brief history of their study from computer science and physics perspectives. However, we have not yet asked perhaps the most fundamental question: \emph{Why?} Namely, where do Hamiltonians come from, and why do we care about their study? Where does the picture sketched thus far fit into a condensed matter physicist's views of what is important and what is not? In this section, we attempt to answer these questions.

We begin in \S\ref{sscn:setstage} and \S\ref{sscn:timeandthermal} by discussing the high-level questions condensed matter physics is typically interested in. In \S\ref{sscn:wherefrom}, we briefly pause to consider where Hamiltonians come from, and discuss challenges faced by physicists regarding their use. Finally, \S\ref{sscn:techniques} gives a brief introduction to classical and quantum techniques for overcoming the challenges mentioned above.

\section{Setting the stage}\label{sscn:setstage}
To set the stage, we first clarify the high-level settings in which the study of local Hamiltonians is relevant to condensed matter physicists:\\

\begin{enumerate}
    \item \textbf{The question.} Typically, we are interested in this question: Given a quantum state $\rho$, compute some \emph{local} property of $\rho$ (see, e.g., Equation~\ref{eqn:sim}). In other words, determining an accurate description of the \emph{entire} wave function is {not} always necessary. (There are exceptions to this statement, such as when studying topological properties, which are intrinsically \emph{non}-local.) For example, given a state $\rho$ (either in the lab or via some succinct classical representation), we might wish to estimate the $2$-point correlation function
\[
    \langle\sigma^z_m\sigma^z_n\rangle := \trace\left[\rho (\sigma^z_m\otimes \sigma^z_n)\right],
\]
where $\sigma^z_m$ denotes the Pauli $Z$ operator applied to qubit $m$. Such local properties are dubbed \emph{intensive}, in that they do not scale with the system's size. In contrast, the ground state energy of a Hamiltonian may scale with the system size, and is thus \emph{extensive}.

    \item \textbf{The setting.} When it comes to interaction (i.e. constraint) graphs, it is natural for theoretical computer scientists to study constraint satisfaction problems on a variety of graphs: Sparse graphs, dense graphs, expanders, and so forth. Condensed matter physicists, on the other hand, often simplify their models approximating the physical world by assuming that the collection of quantum particles to be studied lives on a $d$-dimensional lattice of $N$ sites, and that all interactions are nearest-neighbor. Indeed, typically in real materials particles \emph{are} arranged in a regular lattice and the interaction range is short. Moreover, we are interested in the behavior of local properties in the thermodynamic limit $N\rightarrow+\infty$, since $N$ is very large in a macroscopic piece of material.
\end{enumerate}

\noindent We thus henceforth focus mainly on local properties of systems arranged on a lattice in the remainder of this section.

\section{Time evolution versus thermal equilibrium}\label{sscn:timeandthermal}

Given that we wish to study local properties of systems arranged on a lattice, we now ask: In which contexts can we study such local properties? Two such contexts are: (a) With respect to the \emph{time evolution} of the quantum system, and (b) at \emph{thermal equilibrium}. We now discuss these in further depth.

\paragraph{Time evolution.} In this setting, we are interested in how local properties of the system behave as the system evolves in time. Here, recall that the time evolution of a quantum state at time $t$, denoted $\rho(t)$, is given by the Schr\"{o}dinger equation, which is a postulate of quantum mechanics:
\[
    \rho(t) = e^{-iHt}\rho(0)e^{iHt},
\]
where $H$ is a Hermitian operator known as the \emph{Hamiltonian} of the system. In particular, in this setting, we imagine that the system $S$ being studied is \emph{isolated} from the outside world.

\paragraph{Thermal Equilibrium.} In contrast to the setting of time evolution above, when studying local properties of our system $S$ at \emph{thermal equilibrium}, we imagine that $S$ is allowed to \emph{interact} with its environment for a ``long time''. Eventually, the system will reach a stationary state known as the \emph{thermal state}, $\rhoeq$, which no longer depends on time. Here, by ``stationary state'', we mean that the {local} properties of the system have ceased to fluctuate. The density matrix of our thermal state is known as the \emph{Gibbs state}, which can be derived from the equal \emph{a priori} probability postulate of statistical mechanics:
\begin{equation}\label{eqn:gibbs}
    \rhoeq = \frac{e^{-\beta H}}{Z}\quad\quad\text{for}\quad\quad Z:=\trace\left(e^{-\beta H}\right).
\end{equation}
Here, $Z$ is called the \emph{partition function}, and the parameter $\beta$ scales inversely with temperature $T$. Two remarks are in order. First, as stated earlier, $\rhoeq$ does not depend on time, $t$. Second, $\rhoeq$ is not \emph{provably} a description of $S$ at equilibrium; rather, it is an educated guess from statistical mechanics which works well for many systems in practice, \emph{assuming} one is interested in local properties of $S$. For further reading on notions of equilibration, the reader is referred to (e.g.)~\cite{LPSW09}.

\section{Where do Hamiltonians come from?}\label{sscn:wherefrom}

We have now stated that we are interested in studying local properties of Hamiltonians on the lattice, and that we can conduct such studies either with respect to time evolution or at thermal equilibrium. But we still have not answered a key question: Where does our Hamiltonian \emph{come from} to begin with? There are two answers to this question: The easy answer is that Hamiltonians naturally arise from the Schr\"{o}dinger equation. However, this is somewhat unsatisfying --- it gives no intuition as to which classes of Hamiltonians describe physically relevant systems. Thus, the more interesting (and difficult) question is: How do we identify physically motivated classes of Hamiltonians?

To answer this, suppose for a moment that you are a physicist, and you are sitting in your lab observing your experimental apparatus induce some quantum evolution. How would you determine the Hamiltonian $H$ governing this evolution? Does the Schr\"{o}dinger equation give you an efficient approach for deducing $H$? {No!} In practice, one has no idea what $H$ actually is. Thus, we must resort to the grade-school Scientific Method --- make an educated guess $H'$ (a Hamiltonian which might reproduce the local properties of our system), and then corroborate this guess via experiments. In particular, note that our aim is not to guess $H$ itself, as this is often intractable due to the large number of degrees of freedom involved. Instead, we desire a \emph{simplified} model $H'$ which ignores certain degrees of freedom, but nevertheless captures local properties of interest. For example, one might choose to model only the {spin} of an electron accurately, and ignore the electron's other properties. For this reason, Hamiltonian models appearing in the condensed matter literature should {not} be thought of as ``accurately characterizing a system'', but rather as \emph{phenomenological} objects which, up to minor corrections, appear to model certain local properties of the system well.

Summarizing, to determine the Hamiltonian $H$ modeling our system, we (1) guess a candidate simplified Hamiltonian $H'$, (2) check what a theoretical (i.e. pen-and-paper) calculation on $H'$ predicts about the behavior of the system, and (3) confirm this prediction via experimental tests. What kind of theoretical calculations might we consider in step (2) above? This is where key issues faced by physicists arise. For example, in step (2) we might wish to estimate a two-point correlation function of $H'$'s ground state. However, we immediately hit a brick wall, as the ability to estimate such functions in general implies the ability to estimate a Hamiltonian's ground state energy (which is QMA-hard)! So it seems we are in a bind; to better understand physical systems, we model them using Hamiltonians, but studying Hamiltonians themselves is not necessarily easy.

Indeed, this is precisely the focus of much of condensed matter physics. Here are two key questions faced by the community:
\begin{itemize}
    \item Given a phenomenological Hamiltonian $H'$, how well does it actually model the desired properties of the underlying system $S$?
    \item How well does $H'$ allow us to make predictions about future properties of $S$ which we might be interested in?
\end{itemize}
In the next section, we discuss various approaches for answering questions related to these.

\section{The study of Hamiltonians: Techniques and tools}\label{sscn:techniques}

In the previous section, we identified key questions faced by the physics community regarding the study of local Hamiltonians. In this section, we give a flavor of both classical and quantum techniques for circumventing the challenges arising in these scenarios. In particular, the basic goal here is: Given a phenomenological Hamiltonian $H$, calculate certain properties of $H$. We begin with classical simulations in \S\ref{ssscn:classicalsim}, and move to quantum simulations in \S\ref{ssscn:quantsim}.

\subsection{Classical simulations}\label{ssscn:classicalsim}

Given a Hamiltonian $H$, in this section, we use classical computers to study the following:
\begin{itemize}
    \item Are there efficient \emph{algorithms} for approximating local properties of $H$?
    \item Can objects of interest, such as the ground state of $H$, be represented by a space-efficient \emph{data structure}? Do these structures allow us to perform useful computations?
    \item Finally, given such a space-efficient data structure, can it be used to make \emph{predictions} about the system for future properties we might be interested in?
\end{itemize}
Two remarks are in order here. First, an \emph{efficient} algorithm to a theoretical computer scientist typically means a rigorous algorithm which runs in polynomial time on a {worst case} instance. In contrast, to a physicist, an efficient algorithm often means a fast heuristic algorithm which works well in practice. Second, the third question on the ability to make predictions about a system is vital --- it allows us to mathematically predict how the properties of a material will change or behave under certain conditions, which in turn allows us to design advanced materials with desirable properties.

\paragraph{The variational principle.} To answer the questions above, a primary line of attack is via the heuristic of the \emph{variational principle}. Specifically, suppose one is interested in optimizing a function $f(S)$ over a set $S$. Then, the variational principle simply means carrying out the optimization over some restricted set $S'\subseteq S$, which we hope will allow an easier computation.

For example, suppose we wish to compute the ground state energy of Hamiltonian $H$ acting on $n$ qubits, i.e.
\begin{equation}\label{eqn:min}
    \min_{\rho\succeq 0, \trace(\rho)=1} \trace(H\rho).
\end{equation}
Since $\rho$ can be highly entangled, it may require exponentially many bits to describe classically; this makes designing a heuristic algorithm for solving Equation~\ref{eqn:min} difficult. Via the variational principle, however, one can instead optimize over a simpler set $S'$ of quantum states in order to approximate the ground state energy. Which set of states $S'$ should we choose? Below, we discuss three popular candidates.\\

\noindent\emph{Product states.} The most basic set $S'$ one can choose is the set of all tensor product states $\rho=\rho_1\otimes\rho_2\otimes\cdots\otimes \rho_n$. Do not be fooled by the simplicity of this ansatz, however; it has proven quite effective in many scenarios, and there is an entire area of study devoted to it known as \emph{mean field theory}. (See \S\ref{sscn:meanfield} for more on mean field theory.) In addition, working with product states is by no means easy --- optimizing Equation~\ref{eqn:min} over product states, for example, is NP-hard. (To see this, note that if $H$ encodes a 3-SAT instance, then the ground state of $H$ is without loss of generality a product state.) \\

\noindent\emph{Gaussian states.} Another common choice of $S'$ is the set of \emph{Gaussian} states, i.e.
\[
    \rho = e^{-Q(a_1^\dagger,a_2^\dagger,\ldots,a_1,a_2,\ldots)},
\]
where $Q$ is some quadratic function, and the operators $\{a_i^\dagger\}$ and $\set{a_i}$ are known as \emph{creation} and \emph{annihilation} operators, respectively. (See \S\ref{ssscn:quantsim} for more on creation and annihilation operators.) Local predictions on Gaussian states can be made efficiently, i.e. in time scaling polynomially with the system size. For further reading on Gaussian states, the reader might consider (e.g.)~\cite{WHTH07, ARL14}.\\

\noindent\emph{Tensor network states.} The final choice of $S'$ we discuss is the set of \emph{tensor network states}. (See \S\ref{sscn:tensor} for a definition of tensor networks.) Unlike product states, tensor network states can represent a variety of quantum states with exotic quantum correlations, such as (chiral \cite{DR13, WTSC13}) topological states. However, as mentioned in \S\ref{sscn:physicshistory}, this expressive power has a downside: Contracting tensor networks is $\#{\rm P}$-hard in general. To circumvent this, the strategy is to choose a clever subset of tensor network states $S''$ which not only allows us to model the system we are interested in, but which also allows efficient calculation of local properties. Arguably the most successful application of this idea in conjunction with the variational principle has been White's DMRG algorithm~\cite{W92,W93} on MPS, which is discussed in further depth in \S\ref{sscn:DMRG}. Note, however, that DMRG is a heuristic; in contrast, rigorously finding a good MPS approximation to the ground state of a $1$D (gapless) Hamiltonian is NP-hard~\cite{SCV08}.

Let us delve further into tensor networks, as they have proven a particularly useful ansatz over the last two decades. In particular, a key question regarding them is:
\begin{quote}
    \emph{Which quantum systems have ground states which can be well-approximated by a tensor network of small bond dimension $D$?}
\end{quote}
Here, the \emph{bond dimension} is an important parameter characterizing ``how much'' entanglement the network can faithfully represent; the larger $D$ is, the more quantum states we can represent, at the cost of requiring more storage space. (See \S\ref{sscn:tensor} for further details.) Since we are restricted to polynomial-sized descriptions of states in practice, it is thus imperative to keep the bond dimension as small as possible (representing an arbitrary quantum state requires exponential $D$).

We now briefly overview recent progress regarding the question above. In $1$D gapped systems, 
MPS with subpolynomial $D$ approximate the ground state well~\cite{Ha07,AKLV13}, and moreover, there is a polynomial-time algorithm for finding a good MPS approximation to the ground state~\cite{LVV13}. In $1$D gapless (critical) systems, MPS with polynomial $D$ approximate the ground state well if certain R\'{e}nyi entropies diverge at most logarithmically~\cite{VC06} (see also~\cite{SWVC08}). On the other hand, at \emph{finite} temperature $T>0$ in any spatial dimension $d$, Hastings has shown~\cite{H06} that a tensor network with bond dimension $D=\exp (O(\beta\log(N/\epsilon)^d))$ approximates the Gibbs state of the system within trace distance $\epsilon$, where $N$ is the size of the lattice, and $\beta:=1/\kappa_B T$ for temperature $T$ and Boltzmann constant $\kappa_B$. If $H$ is gapped and the density of low energy states grows polynomially in $N$ for a fixed energy, then this result also implies a tensor network approximation of the \emph{ground state} with $D=\exp(O(\log^{d+1}(N/\epsilon)))$.  Molnar \emph{et al.} improved~\cite{MSVC14} these results to $D=(N/\epsilon)^{O(\beta)}$ and $D=\exp (O(\log^2 (N/\epsilon)))$ in the finite temperature and ground state cases, respectively. 
Intuitively, the ``density of states'' assumption here says that the low energy spectrum is ``sparsely populated'' by eigenvectors; this is, in fact, typical of quantum systems in nature.

For more information on the use of the variational principle in conjunction with tensor networks, the reader is referred to~\cite{VMC08}. An introductory discussion on the variational principle aimed at physicists is given in Chapter $7$ of Griffiths~\cite{G04}.

\subsection{Quantum simulations}\label{ssscn:quantsim}

In \S\ref{ssscn:classicalsim}, we discussed the use of \emph{classical} simulations to study phenomenological Hamiltonians $H$. One particularly difficult setting which, however, does not seem amenable to these techniques is that of \emph{strongly correlated} materials. For example, high-temperature superconductors, despite having seen over two decades of research, remain arguably poorly understood. Recent research in \emph{quantum} simulations has, however, led to exciting advances in this direction~\cite{QH09}. Key to this progress are two ingredients: One old, and one new. The ``old'' ingredient is the famous Hubbard model, formulated by John Hubbard in the 1960's. The ``new'' ingredient is Feynman's remarkable idea~\cite{F82} of building \emph{quantum} devices to efficiently simulate quantum systems. In particular, the idea explored in this section is the use of quantum simulations to uncover properties of the Hubbard model. Further details on this topic can be found, for example, in the brief survey of Quintanilla and Hooley~\cite{QH09}.

Before we begin, let us stress that to study (say) the Hubbard model, one need not necessarily construct a \emph{general-purpose} quantum computer; rather, a device tailored to simulating this one model would suffice. Thus, the goal of quantum simulations may be easier to achieve than a large scale quantum computer~\cite{CZ12}. Indeed, this is the approach taken in this section.

In order to define the Hubbard model, we must first make a brief detour to introduce the notions of fermions and bosons, which are the types of particles on which the model acts.

\paragraph{Detour: Indistinguishable particles.} Before defining fermions and bosons formally, let us set the stage with a high-level overview of indistinguishable particles. The Standard Model of particle physics, developed in the latter half of the $20^{\rm th}$ century, characterizes various types of interactions which occur in nature, such as the electromagnetic interaction. These interactions, in turn, govern the dynamics of subatomic particles such as photons, electrons, and protons. In particular, the model treats these particles  as \emph{indistinguishable}. In other words, if for example we are given a pair of photons at positions $X$ and $Y$, then the physical state of the photons (described by the wave function, up to global phase) remains \emph{invariant} if we swap the positions of the photons, i.e. the physical state only cares about the \emph{number} of particles at each site. However, there is a twist: The wave function of some particles, such as electrons, obtains a global phase of $-1$ upon performing a swap. This may \emph{a priori} appear irrelevant, as global phases cannot be physically observed. However, it is a remarkable conclusion of the Standard Model that this seemingly innocuous global phase turns out to dictate the very structure of matter around us.

To understand this phenomenon, the Standard Model differentiates between two types of particles: \emph{Bosons} and \emph{fermions}. Mathematically, these types are characterized by the Spin-Statistics theorem, which states that given a system of identical particles, one of the two cases must hold:
\begin{itemize}
    \item If the particles have integer spin (i.e. $s=0,1,2,\ldots$), then exchanging any pair of them leaves the wave function invariant. Such particles are called \emph{bosons}.
    \item If the particles have half-integer spin (i.e. $s=1/2,3/2,5/2, \ldots$), then exchanging any pair of them induces a global phase of $-1$ on the wave function. Such particles are called \emph{fermions}.
\end{itemize}
Note that the latter statement above captures the curious phase of $-1$ discussed earlier, and is essentially the Pauli exclusion principle. 

What does the Spin-Statistics theorem actually mean in terms of how large numbers of bosons or fermions behave? The distinction between bosons and fermions manifests itself at thermal equilibrium by leading to two possible statistical distributions governing how a system of indistinguishable and non-interacting particles populate a set of discrete energy states: Bose-Einstein statistics for bosons, and Fermi-Dirac statistics for fermions. (This also explains the name Spin-\emph{Statistics} Theorem above.) We begin by discussing the former. Note that both statistics hold for sufficiently concentrated systems of particles at low temperature. At high temperature, the same statistics hold; however, these are now well-approximated by classical Maxwell-Boltzmann statistics.\\

\noindent \emph{Bose-Einstein statistics.} Developed by Satyendra Nath Bose and Albert Einstein in the mid-1920's, Bose-Einstein statistics predicts that at low temperature, particles tend to aggregate in the same quantum state, namely the ground state. As a result, bosons typically play the role of ``force carrier'' particles in nature, i.e. they transmit interactions. A good example of this is a laser, which consists of many bosons in the same state. In addition, the tendency for a system of bosons to occupy the ground state at very low temperatures can lead to a special state of matter known as a \emph{Bose-Einstein condensate}, which can exhibit quantum effects at the macroscopic scale. For example, when Helium-4, which is a gas of bosons, is cooled to temperatures near absolute zero, it becomes a \emph{superfluid}, i.e. it behaves like a fluid with zero viscosity. *Thus, the ability for many bosons to aggregate in the same mode is crucial to the role they play in nature. Examples of bosons include elementary particles, such as photons, gluons, and the Higgs boson, as well as composite particles, such as the Helium-4.\\

\noindent\emph{Fermi-Dirac statistics.} Named after Enrico Fermi and Paul Dirac, Fermi-Dirac statistics applies to fermions, and in stark contrast to bosons, follows the principle that two fermions cannot occupy the same quantum state. For example, if two electrons (electrons are fermions with spin $1/2$) are at the same site, then they must differ in at least one property, such as having anti-aligned spins. It is precisely for this reason that electrons arrange themselves in orbits with higher and higher energy about the nucleus of an atom, giving matter a ``rigid'' structure and non-trivial volume. In this sense one may intuitively think of fermions as the ``building blocks of matter''. Examples of fermions include elementary particles such as electrons, quarks, and leptons, and composite particles such as protons and neutrons. In fact, any composite particle composed of an odd number of fermions is also a fermion; for example, protons and neutrons consist of three {quarks}. Any composite particle composed of an even number of fermions, however, is a boson.

\paragraph{Back to quantum simulation: The Hubbard model.} We now introduce the Hubbard model, which since its inception in 1963~\cite{H63}, has become a ``standard model'' in condensed matter physics. Hubbard's aim was to propose ``the simplest possible model'' which would explain the behavior of strongly correlated materials~\cite{QH09}. In particular, his model aims to describe the behavior of electrons in solids, and can capture the transition of a system between being a conductor and an insulator. It is thus of particular interest in the study of high-temperature superconductivity. Although the original model was proposed using fermions, a bosonic version is also known and referred to as the \emph{Bose-Hubbard} model, which we also introduce here.

To help understand the technical definitions we give below, the intuition behind the Hubbard model is as follows. Previous to Hubbard's proposal, the \emph{tight-binding} approximation from solid-state physics explained conduction by viewing electrons as ``hopping'' from the electron orbitals of one atom to another. However, when certain materials are heated, thus increasing the spacing between atoms, they can transition from being a conductor to an insulator; the tight-binding approximation fails to account for this. To resolve this, Hubbard introduced a second term to the tight-binding approximation's Hamiltonian which is meant to model the ``on-site repulsion'' resulting from the Coulomb repulsion between electrons.

For simplicity, we begin with the Bose-Hubbard model (i.e. bosons).\\

\noindent\emph{The Bose-Hubbard model.} The Bose-Hubbard model involves interacting bosons, and its Hamiltonian on a 2D lattice is given by
\begin{equation}\label{eqn:bosehubb}
    H=-t\sum_{\langle n,m \rangle} a_n^\dagger a_m + u\sum_n a_n^\dagger a_n^\dagger a_n a_n.
\end{equation}
Here, $\langle n,m\rangle$ denotes that the summation is taken over nearest neighbors, $t$ is the {hopping amplitude}, $u$ is the interaction strength, and $\{a_n^\dagger\}$ and $\set{a_n}$ are the bosonic creation and annihilation operators, respectively. Let us define the latter.

Let $\{\ket{k}_n\}_{k=0}^{+\infty}$ be a complete set of local orthogonal basis vectors, where $\ket{k}_n$ means that there are $k$ (identical) bosons in location or \emph{mode} $n$ (the subscript is omitted if the location is unspecified). Then, the state $|0\rangle$ has a special name: It is the \emph{vacuum} state. The bosonic creation and annihilation operators are given by
\begin{eqnarray*}
    a^\dagger\ket{k}=\sqrt{k+1}\ket{k+1}\quad\quad\text{and}\quad\quad a\ket{k}=\sqrt{k}\ket{k-1},
\end{eqnarray*}
and they satisfy the canonical commutation relations
\[
    [a_m,a_n] = [a^\dagger_m,a^\dagger_n]=0\quad\quad\text{and}\quad\quad[a_m,a^\dagger_n]=\delta_{mn},
\]
where $[x,y]=xy-yx$, and $\delta_{mn}$ is the Kronecker delta function ($\delta_{mn}=1$ if $m=n$ and $\delta_{mn}=0$ otherwise). Thus, for example, the state of $k$ (identical) bosons can be expressed as
\[
    \ket{k}=\frac{(a^\dagger)^k}{\sqrt{k!}}\ket{0}.
\]

\noindent\emph{The Hubbard model.} In the Bose-Hubbard model, one could have arbitrarily many bosons in a single mode. Recall, however, that by the Pauli Exclusion Principle, any pair of fermions occupying the same site must differ in some property, e.g. their spin. The (fermionic) spin-$1/2$ Hubbard Hamiltonian is
\begin{equation}\label{eqn:fermhubb}
    H=-t\sum_{\substack{\langle n,m \rangle\\\sigma\in\set{\uparrow,\downarrow}} } a_{n,\sigma}^\dagger a_{m,\sigma} + u\sum_n a_{n,\uparrow}^\dagger a_{n,\uparrow} a_{n,\downarrow}^\dagger a_{n,\downarrow}.
\end{equation}
In this case, the fermionic creation and annihilation operators are defined on basis $\set{\ket{0},\ket{1}}$ as
\begin{eqnarray*}
    a^\dagger\ket{0}=\ket{1},\quad a^\dagger\ket{1}=0,\quad a\ket{0}=0,\quad a\ket{1}=\ket{0},
\end{eqnarray*}
with the canonical \emph{anti}-commutation relations
\[
    \{a_{m,\sigma},a_{n,\sigma'}\} = \{a^\dagger_{m,\sigma},a^\dagger_{n,\sigma'}\}=0\quad\quad\text{and}\quad\quad\{a_{m,\sigma},a^\dagger_{n,\sigma'}\}=\delta_{mn}\delta_{\sigma\sigma'},
\]
where $\{x,y\}=xy+yx$. For clarity, in matrix form, we have
\[
    a^\dagger=\left(
                \begin{array}{cc}
                  0 & 0 \\
                  1 & 0 \\
                \end{array}
              \right)
    ,\quad\quad a=\left(
                \begin{array}{cc}
                  0 & 1 \\
                  0 & 0 \\
                \end{array}
              \right).
\]
Note that, as expected for fermions, $a^\dagger$ is nilpotent, meaning $(a^\dagger)^2=0$ --- this implies precisely that a second fermion cannot be created in an occupied mode.

As an aside, for the computer science reader interested in the Bose-Hubbard and Hubbard models, suggested reading might include Schuch and Verstraete's~\cite{SV09} QMA-hardness result for the Hubbard model, and Childs, Gosset, and Webb's QMA-hardness result for the Bose-Hubbard model~\cite{CGW13}.\\

\noindent\emph{Studying the Hubbard model.} Shortly after the Hubbard model was proposed, Lieb and Wu solved it exactly in $1$D in 1968~\cite{LW68}, obtaining its ``phase diagram''. A \emph{phase diagram} is a diagram depicting the quantum phases (e.g. metal, insulator, etc\ldots) of a system as a function of tuning parameters. In the Hubbard model, the tuning parameters are the electron density and the repulsion strength. In particular, Lieb and Wu found that the system is a \emph{Mott insulator} when there is one electron per site (i.e. at ``half filling''), and that it becomes metallic when we add or remove a small number of electrons (i.e. slightly away from half filling). Such transitions between different phases of matter are called \emph{(quantum) phase transitions}.

Beyond $1$D, unfortunately, the Hubbard model has proven remarkably difficult to solve analytically. One approach to cracking the higher dimensional case over the last decade or so has involved \emph{quantum} simulation via ultracold atoms~\cite{QH09}. The idea here is as follows: When ultracold atoms are trapped in crossed laser beams (i.e. in an optical lattice), then under certain circumstances, the behavior of the \emph{atoms themselves} is described by the Hubbard model.
Thus, to uncover properties of the Hubbard model, we can instead probe such atomic systems.

Of course, this approach also faces difficulties. For example, when we say ``ultracold'' atoms, we do mean \emph{ultra}cold: Such systems behave quantumly only when cooled to near absolute zero (for example, a few billionths of a degree above absolute zero). Nevertheless, recent experiments have managed to resolve some of these issues~\cite{MJ99,JSGME08,SHWBBCHRR08}, obtaining a good simulation of the 3D Hubbard model. Such experiments have offered evidence of a similar metal-insulator phase transition, just as was discovered analytically for the 1D Hubbard model by Lieb and Wu~\cite{LW68}. This offers hope for a huge step forward, as obtaining the phase diagrams of the 2D and 3D Hubbard models has been an open question for decades, whose resolution may hopefully lead to further breakthroughs in our understanding of strongly correlated materials.

\chapter{Physics Concepts in Greater Depth}\label{scn:physicsforCS}

In Chapter~\ref{scn:motivation}, we gave a high-level overview of certain central questions in condensed matter physics. We now discuss various concepts introduced therein in greater depth. We begin in \S\ref{sscn:terms} by providing a glossary of selected common terms which appear in the many-body physics literature. \S\ref{sscn:meanfield} discusses mean field theory. In \S\ref{sscn:tensor},~\S\ref{sscn:DMRG},~\S\ref{sscn:MPS}, and~\S\ref{sscn:MERA}, we review tensor networks and their special cases (e.g., MPS and MERA). \S\ref{IDMRG} sketches the implementation of the DMRG as a variational algorithm over MPS. Lastly, \S\ref{sscn:arealaw_overview} provides a brief survey on the subject of the area laws.

\section{A glossary of physics terms}\label{sscn:terms}
The following is a glossary of common terms in the physics literature.\\

\noindent\textbf{Basic terminology.}
\begin{itemize}
        \item \emph{Singlet}: A singlet is a quantum state with zero total spin. In a system of two qubits, it is given by the Bell state:
            \[
                \ket{\psi^-}:=\frac{1}{\sqrt{2}}(\ket{01}-\ket{10})=\frac{1}{\sqrt{2}}(\ket{\uparrow\downarrow}-\ket{\downarrow\uparrow}).
            \]
            Note that for two-qubit systems, the singlet projects onto the antisymmetric space, and $I-\ketbra{\psi^-}{\psi^-}$ projects onto the symmetric space.
        \item \emph{Spin}: Spin is an ``intrinsic'' type of angular momentum carried by quantum particles. It can be measured along certain directions in real space. Let us fix a set of three orthogonal axes (our world has three spatial dimensions) with respect to which we wish to perform a measurement. This induces a set of observables $\set{S_x,S_y,S_z}$. The measurement outcomes (eigenvalues) of each observable $S_x,S_y$ or $S_z$ for a \emph{spin-(k/2)} particle ($k$ a positive integer) can be $-k/2,-k/2+1,\ldots,k/2$. For example, a spin-$(3/2)$ particle can yield $\pm 1/2$ and $\pm 3/2$ upon measurement. As such, a spin-$(k/2)$ particle corresponds to a qudit with dimension $d=k+1$.

    \item \emph{Many-body system}: A many-body system consists of a large number of particles (usually arranged in a regular lattice). It is governed by a many-body Hamiltonian. The specification of a particular such Hamiltonian is analogous to choosing a family of constraints in a classical CSP. Indeed, one might view $3$-SAT as a many-body system in which all constraints are $3$-local Boolean formulae in Conjunctive Normal Form (CNF). Note that many-body systems can be either classical or quantum, such as the classical and quantum Ising models.
    \item \emph{Interaction (hyper)graph}: The interaction (hyper)graph illustrates which sets of particles are constrained by a common local constraint in the Hamiltonian. For $2$-local Hamiltonians, this is just an undirected graph in which $(i,j)$ is an edge if and only if there exists a term $H_{ij}$ in the Hamiltonian acting jointly on particles $i$ and $j$.
    \item \emph{Solving a model}: The meaning of this phrase depends largely on context. Typically, we are given a classical description of a local Hamiltonian $H$, and we wish to find some property of the system governed by $H$. For example, we may wish to solve for the ground state energy of $H$, or to determine whether $H$ is gapped. \emph{Solving the model} means calculating the value of whatever property of $H$ we are interested in.
    \item \emph{Two-point correlation function}: Two-point correlation functions are defined as
\[
    \langle O_mO'_n\rangle := \trace\left[\rho (O_m\otimes O'_n)\right]
\]
for some local operators $O_m,O'_n$ acting on sites $m,n$, respectively.
\end{itemize}

\noindent\textbf{Commonly studied models in condensed matter physics.} All models listed below are $2$-local.

\begin{itemize}
    \item \emph{Classical Ising model}: We are given a lattice, and at each lattice site $i$ there is a binary variable $x_i\in\set{+1,-1}$. Let $\vec x:=(x_1,x_2,\ldots,x_n)$ with $n$ the total number of sites. An assignment $\vec{x}\in\{+1,-1\}^n$ is called a \emph{configuration}, and the \emph{energy} of a configuration $\vec{x}$ is given by the Hamiltonian
    \begin{equation}\label{eqn:cIsing}
        H(\vec{x})=\sum_{\langle i,j\rangle}J_{ij}x_ix_j,
    \end{equation}
    where $J_{ij}$'s are interaction strengths (real numbers), and $\langle i,j\rangle$ denotes summation over nearest-neighbor pairs $(i,j)$ (according to whatever underlying lattice we are given, such as a 1D chain or a 2D square lattice). It is common to write $H$ instead of $H(\vec{x})$ for notational simplicity.

    To illustrate the connection to computer science, note that setting all $J_{ij}=1$ and searching for the ground state configuration of the Hamiltonian $H(\vec{x})-\abs{E}$ is equivalent to the MAX CUT problem, where $\abs{E}$ is the number of edges in the interaction graph.

    One may also consider adding a linear term to the Ising Hamiltonian, i.e.,
    \[
        H(\vec{x})=\sum_{\langle i,j\rangle}J_{ij}x_ix_j + \mu \sum_i m_ix_i,
    \]
    where $m_i$ models the effect of an external magnetic field on site $i$, and $\mu$ is the magnetic moment. We usually set $\mu=1$ without loss of generality.
            \item \emph{Quantum Ising model}: The quantum Ising model is also known as the transverse field Ising model \cite{Pfe70, Sac11}. Its Hamiltonian is
                \[
                    H = -J\sum_{\langle i,j\rangle} \sigma^z_i\sigma^z_j - g\sum_i\sigma^x_i,
                \]
                where $J$ is the coupling constant, and $g$ is the magnitude of the transverse magnetic field. Recall that $\sigma^z$ and $\sigma^x$ are the Pauli $Z$ and $X$ matrices, respectively. This model was recently found to characterize a new complexity class called TIM~\cite{CM13}.
        \item \emph{Quantum Heisenberg model}: The quantum Heisenberg model's Hamiltonian is given by
                \[
                    H = -\sum_{\langle i,j\rangle} (J_x\sigma^x_i\sigma^x_j +J_y\sigma^y_i\sigma^y_j +J_z\sigma^z_i\sigma^z_j ) + h\sum_i \sigma^z_i,
                \]
                where $J_x,J_y,J_z$ are coupling constants, and $h$ is the external magnetic field. Various special cases of this model are of particular interest, such as the \emph{XX model}~\cite{LSM61} (also known as the XY model)
                \[
                    H = -\sum_{\langle i,j\rangle} \sigma^x_i\sigma^x_j +\sigma^y_i\sigma^y_j,
                \]
                and the \emph{XXZ model} \cite{YY661, YY662, Sac11}
                \[
                    H = -\sum_{\langle i,j\rangle} \sigma^x_i\sigma^x_j +\sigma^y_i\sigma^y_j+\Delta\sigma^z_i\sigma^z_j.
                \]
        \item \emph{Anti-ferromagnetic Heisenberg model}: This is another special case of the Heisenberg model. In physics notation the Hamiltonian can be written compactly as (see Equation~\ref{eqn:sigma} for the definition of $\overrightarrow{S}$)
            \[
                H = \sum_{\langle i,j\rangle} \overrightarrow{S}_i\cdot\overrightarrow{S}_j=\sum_{\langle i,j\rangle}\sigma^x_i\sigma^x_j+\sigma^y_i\sigma^y_j+\sigma^z_i\sigma^z_j,
            \]
where we have assumed nearest neighbor interactions (of course, in principle the underlying interaction graph can be any simple graph). 
        To intuitively understand what this $2$-local constraint enforces for the spin-$1/2$ particle case, note that
            \[
                I-(\sigma^x\sigma^x+\sigma^y\sigma^y+\sigma^z\sigma^z) \propto \ketbra{\psi^-}{\psi^-},
            \]
            where $\ket{\psi^-}$ is the singlet. Hence, the ground state of this $2$-local constraint is the $1$-dimensional anti-symmetric subspace, in which spins are aligned in a conflicting fashion, i.e. up-down or down-up.

            We now discuss the properties of the 1D spin-$S$ anti-ferromagnetic Heisenberg chain. For $S=1/2$, the model is exactly solvable by the Bethe ansatz~\cite{B31}. The energy spectrum is gapless, and the spin correlation functions (e.g., $\langle\sigma_i^z\sigma_j^z\rangle$) decay as a power law in $|i-j|$. For $S=1$ (see Equation~\ref{eqn:spin1pauli} for the spin-$1$ definition of $\overrightarrow{S}$), the model is difficult to solve, and was historically expected to be gapless and exhibit power-law decay. Surprisingly, in 1983, Haldane \cite{Hal83a, Hal83b} argued that whether $S$ is a half-odd-integer or an integer is essential: For half-odd-integer $S$ the model was predicted to be gapless and exhibit power-law decay, while for integer $S$, it was predicted to be gapped and exhibit exponential decay. The former, i.e., the absence of an energy gap for half-odd-integer $S$, was proven rigorously \cite{AL86} by using an extension of the Lieb-Shultz-Mattis theorem \cite{LSM61}. As for the latter, evidence for the existence of the ``Haldane gap'' for $S=1$ has been found both numerically and experimentally \cite{BMA+86, MBAH88, RVR+87}.

        \item \emph{AKLT model}: As the 1D spin-1 Heisenberg model has proven difficult to solve, Affleck, Kennedy, Lieb, and Tasaki proposed and studied the similar AKLT model in 1987~\cite{AKLT87, AKLT88}. The AKLT model is artificial and not believed to be experimentally realizable. However, it has the following desirable properties: (i) It looks superficially similar to the spin-1 anti-ferromagnetic Heisenberg chain; (ii) it can be solved exactly; (iii) Haldane's argument (see \emph{anti-ferromagnetic Heisenberg model}) can be rigorously verified for this model. The AKLT model is also useful for understanding MPS \cite{PVWC07, Sch11}, symmetry protected topological (SPT) order \cite{GW09, PBTO12}, etc.

            The 1D spin-1 AKLT Hamiltonian is defined as (see Equation~\ref{eqn:spin1pauli} for definitions of $\sigma_x$, $\sigma_y$, and $\sigma_z$ in the spin-1 case)
            \[
                H = \sum_i \overrightarrow{S}_i\cdot\overrightarrow{S}_{i+1} + \frac{1}{3}\left(\overrightarrow{S}_i\cdot\overrightarrow{S}_{i+1}\right)^2.
            \]
            This model is best understood via the correspondence between a qutrit (spin-$1$) and the symmetric subspace of two qubits (spin-$1/2$):
            \[
                |+1\rangle\leftrightarrow|11\rangle,\qquad|0\rangle\leftrightarrow(|01\rangle+|10\rangle)/\sqrt2,\qquad|-1\rangle\leftrightarrow|00\rangle.
            \]
            The ground state of the AKLT model can be constructed in three steps: (1) For each site $i$, assign two qubits labeled $i_L$ and $i_R$; (2) $i_R$ and $(i+1)_L$ form a singlet; (3) project $i_L$ and $i_R$ onto the symmetric subspace. The ground state is called a ``valence-bond state'' as it is a tensor product of singlets in the qubit representation. We emphasize that the second term in the AKLT Hamiltonian is intentionally added to ensure exact solvability, i.e., that step (2) is the optimal strategy to minimize energy. (Omitting this second term yields the anti-ferromagnetic Heisenberg chain.) Heuristically, ignoring step (3) does not prevent us from capturing the physics of the model. As such, it is straightforward to understand the following properties of the AKLT model intuitively (see \cite{AKLT87, AKLT88} for rigorous proofs):
            \begin{enumerate}
                \item The model is gapped.
                \item The ground state degeneracy depends on boundary conditions: The dimension of the ground state space is $1$ for periodic boundary conditions, and $4$ for open boundary conditions. This is because in the latter case we have two free qubits $1_L$ and $n_R$ ($n$ is the number of sites), while in the former case these two qubits form a singlet.
                \item Correlation functions decay exponentially as the ground state is short ranged.
                \item The ground state can be written exactly as an MPS of bond dimension $2$ \cite{PVWC07, Sch11}, as its Schmidt rank across any bipartite cut is 2.
            \end{enumerate}

\end{itemize}

\noindent\textbf{Physical phenomena.}

\begin{itemize}
    \item \emph{Ferromagnetic and anti-ferromagnetic order}: In a (spin) system with ferromagnetic order, all spins are aligned in the same direction, e.g., all spins up. In contrast, an anti-ferromagnetic system has alternating up and down spins (i.e., neighboring spins point in opposite directions). Note that for this to make sense, the underlying lattice must be bipartite. (We remark that not all 2D lattices are bipartite, such as the triangular or Kagome lattices.) For example, at zero temperature $H=-\sum_i \sigma_i^z\sigma_{i+1}^z$ has ferromagnetic order (its ground states are the all-spin-up and the all-spin-down states), whereas $-H$ has anti-ferromagnetic order. To give an intuitive example, a permanent magnet in ordinary life exhibits ferromagnetic order with macroscopic magnetization.
    \item \emph{Thermal equilibrium, Boltzmann distribution \& Gibbs state}: When a classical system is in \emph{thermal equilibrium}, the probability of observing each configuration $\vec{x}$ of the system is proportional to $\exp\inc{-\beta H(\vec{x})}$, where $H$ is the Hamiltonian (energy), and $\beta=1/T$ is the inverse temperature\footnote{In real physical systems (e.g., gases made of atoms and/or molecules), we usually work with the International System of Units (SI). The inverse temperature is given by $\beta=1/(k_BT)$ in SI, where the unit for temperature $T$ is Kelvin and $k_B\approx1.38\times10^{-23}J/K$ is the Boltzmann constant. In the main text, we follow the convention of theoretical physicists and have rescaled the units such that all quantities are dimensionless (numbers) and in particular $k_B=1$. Thus, note that it does not make sense to define ``room temperature'' as having any particular value (such as, say, $T=20$).}. This distribution of configurations is called the \emph{Boltzmann distribution} or the canonical ensemble. The \emph{partition function} $Z$ is just the normalizing constant of the distribution:
\begin{equation}
  \label{eq:1}
  Z(\beta) \defeq \sum_{\vec{x}}\exp(-\beta H(\vec{x})).
\end{equation}
The \emph{free energy} is defined as
\begin{equation}
  \label{eq:2}
  F(\beta) \defeq -\log Z(\beta).
\end{equation}
Let $O(\vec{x})$ be a physical quantity as a function of $\vec{x}$. Its expectation is given by
\begin{equation}
\langle O\rangle=\frac{1}{Z(\beta)}\sum_{\vec{x}}O(\vec{x})\exp(-\beta H(\vec{x})).
\end{equation}

When a quantum system is in \emph{thermal equilibrium}, the system is in a mixed quantum state described by the density matrix
\begin{equation}\label{eqn:gibbsformal}
    \rho\propto\exp(-\beta H)=\exp(-H/T)
\end{equation}
for $H$ the Hamiltonian and $\beta$ the inverse temperature. This state is known as the \emph{Gibbs} or \emph{thermal} state. The quantum partition function and free energy can be defined analogously to the classical setting:
\begin{equation}
Z(\beta):=\tr\exp(-\beta H).
\end{equation}
Similarly, the expectation of an operator $\hat O$ is given by
\begin{equation}
\langle\hat O\rangle=\tr\left(\hat O\exp(-\beta H)\right)/Z(\beta).
\end{equation}
Note that ``zero temperature'' should be understood as taking the limit $\beta\rightarrow+\infty$.

    \item \emph{Quantum phase transitions and criticality}: Given a family of Hamiltonians $H(\lambda)$ as a smooth (i.e. infinitely differentiable) function of some tuning parameter $\lambda$, let $|\psi(\lambda)\rangle$ be the ground state of $H(\lambda)$. Note that $|\psi(\lambda)\rangle$ may not be smooth (or continuous) in $\lambda$ (even if $H(\lambda)$ is finite dimensional); at such singularities, a \emph{quantum phase transition} occurs. More specifically, at a first-order quantum phase transition $|\psi(\lambda)\rangle$ is not continuous in $\lambda$, whereas at a continuous (second-order) phase transition, $|\psi(\lambda)\rangle$ is continuous but not smooth in $\lambda$, e.g., $\mathrm{d}|\psi(\lambda)\rangle/\mathrm{d}\lambda$ may not be continuous at $\lambda=\lambda_c$. We call $\lambda=\lambda_c$ a \emph{critical point} and $H(\lambda_c)$ a \emph{critical system}; moreover, the physics in the neighborhood of a critical point is called \emph{critical phenomena}. Heuristically, a critical system may be scale-invariant so that its low-energy effective theory is a conformal field theory.

A prototypical example of continuous quantum phase transitions is in the context of the transverse field Ising chain
                    \[
                        H=-\sum_i\sigma_i^z\sigma_{i+1}^z + \lambda\sigma_i^x,
                    \]
                    where we take the \emph{thermodynamic limit} in the sense that the number of spins goes to infinity (i.e., the index $i$ ranges over all integers). For $\lambda\rightarrow+\infty$, the ground state is $\ket{\rightarrow \rightarrow \rightarrow \cdots}$ (unique), where $\ket{\rightarrow}$ is the ground state of $-\sigma^x$. For $\lambda=0$, the ground states are the all-spin-up and the all-spin-down states ($2$-fold degenerate). Indeed, the dimension of the ground state space is one for $\lambda>1$ and two for $\lambda<1$, and a second-order phase transition occurs at $\lambda_c=1$.

\end{itemize}

\section{Mean-field theory}\label{sscn:meanfield}

In \S\ref{ssscn:classicalsim}, we alluded to \emph{mean-field theory} as a (heuristic) variational method over the set of product states to qualitatively extract properties of a many-body system. In this section, we develop mean-field theory from a ``mean-field'' perspective via a textbook example, and briefly comment on its reformulation as a variational method.

A many-body Hamiltonian $H$ is in general difficult to solve due to coupling between particles. To circumvent this, mean-field theory constructs a decoupled and exactly solvable Hamiltonian $H_{mf}$. This process requires good physical intuition and insight into our system of interest, and there is \emph{a priori} no guarantee as to how well $H_{mf}$ approximates $H$. Nevertheless, mean-field theory has proven very successful in some physically important contexts, and thus is among the most widely used approximation methods in condensed matter physics.

The example we study here is the classical Ising model
\[
    H(\vec{x})=-J\sum_{\langle i,j\rangle}x_ix_j
\]
on a $D$-dimensional hypercubic lattice (see Equation~\ref{eqn:cIsing}) at thermal equilibrium, i.e. we study the system's Boltzmann distribution (see the glossary entry for thermal equilibrium in \S\ref{sscn:terms})), and the tuning parameter for us is the temperature $T$. Let us now describe the property of $H$ we wish to calculate in terms of $T$. Note that mapping $x_i\rightarrow-x_i$ for all $i$ leaves the Hamiltonian invariant, i.e. the mapping is a symmetry of the Hamiltonian. However, the ground states (i.e. the all-spin-up and the all-spin-down states) do not respect this symmetry, as they are not left invariant under the action of this mapping (i.e. all-spin-up is mapped to all-spin-down and vice versa). This is called \emph{spontaneous symmetry breaking}, a phenomenon in which the physical states are less symmetric than the Hamiltonian itself. Note that spontaneous symmetry breaking implies a degenerate ground space.

We can try to break this degeneracy by introducing a symmetry breaking perturbation, e.g., an infinitesimal strength magnetic field in the positive direction. (We use infinitesimal strength here as we do not wish to change our system.) This is modeled via a linear term as follows:
\begin{equation}\label{eqn:perturbed}
    H(\vec{x})=-J\sum_{\langle i,j\rangle}x_ix_j-h\sum_ix_i
\end{equation}
with $h\rightarrow0^+$. Then, we claim that $\langle x_i\rangle=1$ for all $i$ at zero temperature $T=0$ and $\langle x_i\rangle=0$ at $T=+\infty$. To see this, first recall that at thermal equilibrium, the probability of observing each configuration $\vec{x}$ of the system is proportional to $\exp\inc{-H(\vec{x})/T}$. Then, for $T\rightarrow 0$, our system is in the all-spin-up state, since the magnetic field we added gives this configuration slightly less energy than the all-spin-down state; thus, $\langle x_i\rangle=1$. For $T\rightarrow +\infty$, on the other hand, all configurations $\vec{x}$ are equally likely; thus, $\langle x_i\rangle=0$. In physics terminology, when $\langle x_i\rangle$ is non-vanishing, we say we have \emph{spontaneous magnetization}. Thus, in our case we have spontaneous magnetization at $T=0$, and there must exist a critical temperature $T_c$ between $T=0$ and $T=+\infty$ at which point the spontaneous magnetization vanishes. Indeed, the system undergoes a phase transition at $T_c$ in the sense that other physical quantities also become singular, e.g., the specific heat and the correlation length diverge. Hence, in ``solving this model'', our goal is to estimate the critical temperature $T_c$.

Before doing so, let us briefly review known exact results regarding this model. For $D=1$, the model is easily solved by dynamic programming (known as the \emph{transfer-matrix method} in physics), and $T_c=0$. For $D=2$, the model was first solved by Onsager \cite{Ons44}, and $T_c=2J/\ln(1+\sqrt2)\approx2.27J$. 
Onsager's method is notorious for being mathematically involved; however, it can be reformulated and simplified by fermionization \cite{SML64}. For $D\ge3$, no exact solution is known.

Returning to mean-field theory, we now define our decoupled mean field Hamiltonian $H_{mf}$. For this, consider an unknown parameter (to be determined later) $m=\langle x_i\rangle$ for all $i$, which is physically motivated since the magnetization is expected to be uniform. Then, by substituting the straightforward identity
\[
    x_i x_j=x_i \langle x_j\rangle + \langle x_i\rangle x_j - \langle x_i \rangle\langle x_j\rangle + (x_i - \langle x_i\rangle)(x_j - \langle x_j\rangle),
\]
into Equation~\ref{eqn:perturbed}, we can write
\begin{eqnarray}
H_{mf}(m,\vec{x})&=&-J\sum_{\langle i,j\rangle}\left(x_i \langle x_j\rangle + \langle x_i\rangle x_j - \langle x_i \rangle\langle x_j\rangle +\right.\nonumber\\ &&\left.(x_i - \langle x_i\rangle)(x_j - \langle x_j\rangle)\right)-h\sum_ix_i\nonumber\\
&\approx&-J\sum_{\langle i,j\rangle}\langle x_i\rangle x_j+x_i\langle x_j\rangle-\langle x_i\rangle\langle x_j\rangle\nonumber\\
&=&DJm^2-2DJm\sum_ix_i,\label{eqn:mf}
\end{eqnarray}
where in the second step, we drop the linear magnetization term since $h\rightarrow 0^+$ and we use the approximation $x_i\approx \langle x_i \rangle$ for all $i$. $H_{mf}(m)$ is now decoupled (non-interacting) and can be easily solved. Let $\langle x_i\rangle_{mf}$ denote the expectation of $x_i$ evaluated using the Boltzmann distribution of the mean-field Hamiltonian $H_{mf}$ at temperature $T$ (we omit parameters $m$ and $T$ in the notation $\langle x_i\rangle_{mf}$ for simplicity), i.e. $\langle x_i\rangle_{mf}$ is the mean-field magnetization. Then:
\begin{eqnarray*}
\langle x_i\rangle_{mf}&=&\frac{\sum_{\vec{x}}x_i\exp(-\beta H_{mf}(\vec{x}))}{\sum_{\vec{x}}\exp(-\beta H_{mf}(\vec{x}))}\\&=&\frac{\sum_{x_i=\pm1}x_i\exp(2\beta DJmx_i)}{\sum_{x_i=\pm1}\exp(2\beta DJmx_i)}\\&=&\tanh(2\beta DJm),
\end{eqnarray*}
where the second equality follows by Equation~\ref{eqn:mf}, and where $H_{mf}(\vec{x})$ also implicitly takes parameter $m$). The mean-field free energy is:
\begin{eqnarray*}
F_{mf}&=&-\ln Z_{mf}\\&=&-\ln \left(\sum_{\vec{x}}\exp(-\beta H_{mf}(\vec{x}))\right)\\&=&n\beta DJm^2-n\ln(2\cosh(2\beta DJm)),
\end{eqnarray*}
where $Z_{mf}$ is the mean-field partition function, and $n$ is the total number of spins. The self-consistency $m=\langle x_i\rangle_{mf}$ of mean-field theory implies the mean-field equation
\begin{equation} \label{mf}
m=\tanh(2\beta DJm).
\end{equation}
This equation has only one solution $m=0$ when $2\beta DJ\le1$ and has three solutions $m=0,\pm m_0$ when $2\beta DJ>1$. In the latter case, one easily verifies $F_{mf}(\pm m_0)<F_{mf}(0)$. Since in classical statistical mechanics the physical solution is the one with the lowest free energy, mean-field theory predicts spontaneous magnetization when $2\beta DJ>1$, i.e., $T_c=2DJ$. In comparison, recall Onsager's result~\cite{Ons44} for $D=2$ of $T_c\approx2.27J$. The mean-field result thus yields a reasonably good approximation for $D=2$.

As an aside, note that instead of a mean-field approach, one can carry out a renormalization group (RG) analysis \cite{Wil71a, Wil71b, KF72, Wil72}, which is much more sophisticated. For $D\le 3$, RG outperforms mean-field theory. For $D\ge 4$, however, both methods agree with each other (on critical exponents).


We now briefly comment on how to reformulate the technique used here as a variational method. Instead of imposing the self-consistency condition by hand, we can view $H_{mf}(m)$ as a family of variational Hamiltonians and do not interpret $m$ as $\langle x_i\rangle$. To obtain the most ``physical'' variational Hamiltonian, we minimize the mean field free energy $F_{mf}(m)$ with respect to $m$ and find that $\mathrm{d}F_{mf}(m)/\mathrm{d}m=0$ is equivalent to the mean-field equation (\ref{mf}). This is not a surprise, but rather a general feature of mean-field theory as a consequence of the Hellmann-Feynman theorem \cite{Fey39}.

Finally, a remark about classical versus quantum phase transitions. In this section, we studied a classical phase transition at finite temperature. Mean-field theory also applies to quantum phase transitions at zero temperature. Briefly speaking, we construct a family of variational (mean-field) Hamiltonians and minimize the mean-field ground state energy. Note that unlike in the finite temperature case discussed above, here we are at zero temperature; thus, our physical state is the ground state and we minimize the ground state energy as opposed to the free energy. (Roughly, free energy takes into account both energy and thermal (entropic) effects at finite temperature.) The ground states of the mean-field Hamiltonians are product states since the mean-field Hamiltonians are decoupled. This roughly explains folklore intuition as to why the variational method over product states is called mean-field theory. We point motivated readers to standard physics textbooks (e.g., \cite{Sac11}) for more examples of mean-field theory.

\section{Tensor networks}\label{sscn:tensor}

Tensor networks were first discussed in \S\ref{ssscn:classicalsim}. Due to their success and popularity in the field, we now give an introduction geared towards computer scientists. Specifically, since arbitrary quantum states $\ket{\psi}\in(\B)^{\otimes n}$ may require exponentially many bits to represent classically, physicists have derived clever ways of encoding certain classes of entangled quantum states in succinct forms. One such approach is via \emph{tensor networks}. Such networks include as special cases Matrix Product States (MPS)~\cite{FNW92,V03,PVWC07} and Projected Entangled Pair States (PEPS)\cite{VC04,VWPC06}.

Informally, to a computer programmer, a \emph{tensor} $M(i_1,i_2,\ldots,i_k)$ is simply a $k$-dimensional array; one plugs in $k$ indices, and out pops a complex number. Hence, in terms of linear algebra, a $1$-dimensional tensor is a vector, and a $2$-dimensional tensor is a matrix. Physicists often like to make this more confusing than it is by simplifying the notation and placing indices as super- or sub-scripts --- for example, they might denote a 3D array $M(i_1,i_2,i_3)$ by $M^{i_1,i_2}_{i_3}$. (There is, of course, a physically motivated reason to write the indices as super- or sub-scripts; however, for those seeing tensor networks for the first time, it is likely simpler to avoid this notation for now.) To make it easier to work with tensors, there is a simple but extremely useful graphical representation. Figure~\ref{fig:disp}(a) shows $M$, for example. Here, the vertex corresponds to the tensor $M$. Each edge corresponds to one of the input parameters to $M$.
\begin{figure}\centering
  \includegraphics[height=3cm]{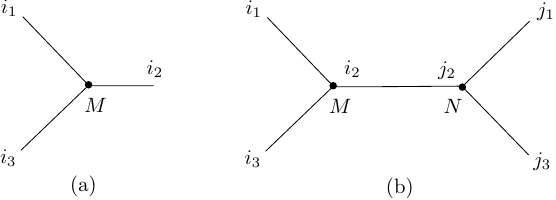}
  \caption{(a) A single tensor $M(i_1,i_2,i_3)$. (b) Two tensors $M(i_1,i_2,i_3)$ and $N(j_1,j_2,j_3)$ contracted on the edge $(M,N)$.}\label{fig:disp}
\end{figure}

Continuing our informal discussion, in Figure~\ref{fig:disp}(b), an edge with two vertices as endpoints corresponds to the operation of \emph{contracting} tensors on edges. Specifically, Figure~\ref{fig:disp}(b) takes two 3D tensors $M$ and $N$, and contracts them on their $i_2$ and $j_2$ inputs, respectively. The resulting mathematical object is a $4$-dimensional tensor $P$ defined as
\[
    P(i_1,i_3,j_1,j_3)=\sum_{k}M(i_1,k,i_3)N(j_1,k,j_3).
\]
Note the resulting tensor in Figure~\ref{fig:disp}(b) has four ``legs'' (i.e. edges with only one endpoint); this is because $P$ takes in four inputs.

More formally, a $k$-dimensional tensor $M$ (as defined above) is a map $M:[d_1]\times\cdots\times[d_k]\mapsto\complex$, where each $d_i$ is a natural number. (We remark that sometimes the dimension $k$ of a tensor is referred to as its \emph{rank}. Note that this notion of ``rank'' is \emph{not} the same as the usual linear algebraic notion of rank for matrices.) For this reason, we can observe the following straightforward way to connect $n$-qubit states $\ket{\psi}$ and $n$-tensors. Let $\ket{\psi}=\sum_{i=1}^{2^n}\alpha_i\ket{i}$ for $\set{\ket{i}}_{i=1}^{2^n}\subseteq(\B)^{\otimes n}$ the computational basis. Then, for any $i\in[2^n]$, letting $i_1\cdots i_n$ denote the binary expansion of $i$, we can define an $n$-tensor $M(i_1,\ldots,i_n)$ for $i_k\in\set{0,1}$ which simply stores all $2^n$ amplitudes of $\ket{\psi}$, i.e. $M(i_1,\ldots, i_n):=\alpha_{i_1\ldots i_n}$. In other words, we can write
\begin{equation}\label{eqn:tensvec}
    \ket{\psi}=\sum_{i=1}^{2^n}M(i_1,\ldots, i_n)\ket{i}.
\end{equation}
More generally, one can generalize this correspondence to represent $n$-qudit systems with local dimension $d$. Then, each index to $M$ would take a value in $[d]$, and $d$ is called the \emph{bond dimension}.

\begin{Question}\label{q:tens1}
In Figure~\ref{fig:disp2}, we depict five different tensor networks. For simplicity, we assume here that all input parameters to a tensor are from the set $[d]$.
\begin{enumerate}
    \item For (a), what type of linear algebraic object does the figure correspond to?
    \item Which operations on objects of the type in (a) do images (b) and (c) depict? What is the output of the tensors in (b) and (c)?
    \item How many tensors is the network in (d) composed of? How many input parameters does each of these constituent tensor networks have (before contraction)? How many input parameters does the final, contracted tensor network have?
    \item Image (d) corresponds to an $m$-dimensional tensor, which using the tensor-vector correspondence in Equation~\ref{eqn:tensvec}, can be thought of as representing an $m$-qubit vector $\ket{\psi}$ (whose amplitudes are computed using the specific contractions between tensors indicated by the network). With this picture in mind, what does $(e)$ correspond to?
    \item Image (e) combines $2m$ tensors into a network which takes no inputs and outputs a complex number $\alpha$. Assuming the bond dimension $d$ is a constant, given these $2m$ tensors, how can we compute $\alpha$ in time polynomial in $m$? Hint: Consider an iterative algorithm which in step $i\in[m]$ considers all tensors up to $N_i$ and $M_i$.
\end{enumerate}
\begin{figure}\centering
  \includegraphics[height=6.2cm]{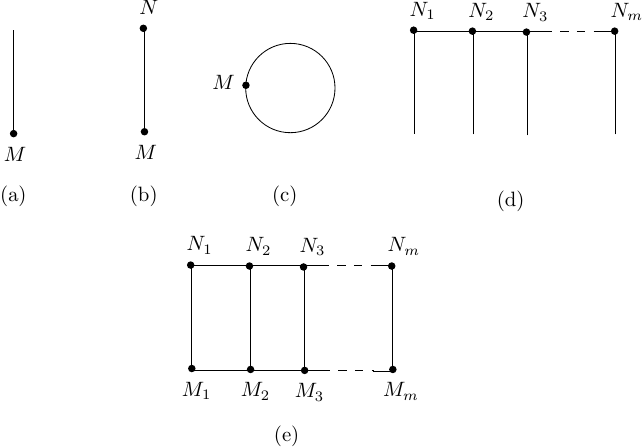}
  \caption{Five tensor networks, studied in Question~\ref{q:tens1}.}
  \label{fig:disp2}
\end{figure}
\end{Question}

There is another view of tensors which also proves useful, in which a tensor is seen as a linear map. Let $S$ denote the set of legs of a tensor $M$, and partition $S$ into subsets $S_1$ and $S_2$. Then, by fixing inputs to all legs in $S_1$, we ``collapse'' $M$ into a new tensor $M'$ corresponding to some vector in $\Bdi{\abs{S_2}}$. For example, consider again our tensor $M$ in Figure~\ref{fig:disp}(a), and let $S_1=\set{i_1}$ and $S_2=\set{i_2,i_3}$. Then, denote by $M_{k}$ the tensor obtained by hardcoding $i_1=k$, i.e. $M_k(i_2,i_3):=M(k,i_2,i_3)$. By Equation~\ref{eqn:tensvec}, $M_k$ corresponds to some vector $\ket{\psi_k}\in\Bdi{\abs{S_2}}$. In other words, we have just demonstrated a mapping which, given any computational basis state $\ket{k}\in\Bd$, outputs a vector $\ket{\psi_k}\in\Bdi{2}$, i.e. we have a linear map $\Phi:\Bd\mapsto\Bdi{2}$. Returning to our more general example with $M$, $S$, $S_1$ and $S_2$, this approach allows us to view a tensor $M$ with legs in $S$ as a linear map $M:\Bdi{\abs{S_1}}\mapsto \Bdi{\abs{S_2}}$.

To demonstrate the effectiveness of this linear map view of tensors, we revisit Question~\ref{q:tens1}(5). This time, we break up this tensor into tensors $T_i$, as depicted in Figure~\ref{fig:disp3}. Then, we can think of $T_1$ as representing the conjugate transpose of a vector $\ket{\psi}\in\Bdi{2}$. Next, $T_i$ for $i\in\set{2,\ldots, m-1}$ can be thought of as linear maps from $\Bdi{2}$ to $\Bdi{2}$, where the two left legs are the inputs, and the two right legs are outputs. It follows that after contracting $T_1$ through $T_{m-1}$, the result is the conjugate transpose of some vector $\ket{\psi}\in\Bdi{2}$. Since the last tensor $T_m$ represents some $\ket{\phi}\in\Bdi{2}$, performing the final contraction computes the inner product $\braket{\psi}{\phi}$ outputting a scalar, as claimed in Question~$\ref{q:tens1}(5)$. Note that since the bond dimension $d$ is considered a constant, this linear map view implies that the contraction of the entire network can clearly be performed in time polynomial in $m$.

\begin{figure}\centering
  \includegraphics[height=3.3cm]{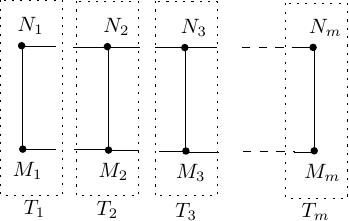}
  \caption{Demonstrating the linear map view of tensor networks.}
  \label{fig:disp3}
\end{figure}

\section{Density Matrix Renormalization Group}\label{sscn:DMRG}

In \S\ref{ssscn:classicalsim}, we discussed the variational principle, and in \S\ref{sscn:tensor}, we introduced tensor networks. We now combine the two to discuss the Density Matrix Renormalization Group (DMRG) algorithm, which is nowadays generally considered the most powerful numerical method for studying one-dimensional quantum many-body systems. In many applications of DMRG, we are able to obtain the low-energy physics (i.e. physical properties at low energy, such as the ground state energy, ground state correlation functions, etc\ldots) of a $1$D quantum lattice model with extraordinary precision and moderate computational resources. Historically, White's invention of DMRG~\cite{W92, W93} in the early 1990's was stimulated by the failure of Wilson's numerical renormalization group~\cite{Wil75} for homogeneous systems. A subsequent milestone was achieved when it was realized~\cite{OR95,RO97,VPC04, VMC08, WVSCD09} that DMRG is in fact a variational algorithm over a specific class of tensor networks known as Matrix Product States (MPS) (see \S\ref{sscn:MPS} below for a definition of MPS).

The purpose of this section is to outline at a high level how DMRG works from the MPS point of view. For further details, we refer the reader to the following review papers on the topic. Schollw\"{o}ck~\cite{Sch11} is a very detailed account of coding with MPS. The earlier paper of Schollw\"{o}ck~\cite{Sch05} discusses DMRG mostly in its original formulation without explicit mention of MPS. Finally, Verstraete, Murg and Cirac~\cite{VMC08} and Cirac and Verstraete~\cite{CV09} focus on the role MPS plays in DMRG, as well as other variational classes of states, such as Tree Tensor States, Multiscale Entanglement Renormalization Ansatz (MERA) and Projected Entangled Pair States (PEPS).

\subsection{Matrix Product States}\label{sscn:MPS}

Matrix Product States (MPS) are the simplest class of tensor network states, and as such, have received much attention. Consider a $1$D quantum spin system of local dimension $d$. We associate each site $i$ with $d$ matrices $A_i^{j=1,2,\ldots,d}$ of dimension $D\times D$, except at the boundaries, where $A_1^{j=1,2,\ldots,d}$ is of dimension $1\times D$ and $A_n^{j=1,2,\ldots,d}$ is of dimension $D\times1$ (where $n$ is the total number of sites). Then an MPS is given by
\[
    |\psi\rangle=\sum_{j_1,\ldots,j_n=1}^dA_1^{j_1}A_2^{j_2}\ldots A_n^{j_n}|j_1,\ldots,j_n\rangle.
\]
How can one interpret this expression for $\ket{\psi}$? First, note that for any $i\in\set{2,\ldots,n-1}$, we have $A_i^{j_i}:[d]\mapsto\lin{\complex^D}$. In other words, fixing an index $j_i\in[d]$ pops out a $D\times D$ complex matrix $A_i^{j_i}$. Similarly, $A_n^{j_n}$ ($A_1^{j_1}$) outputs a complex vector (conjugate transpose of a complex vector). It follows that for any string $j_1\cdots j_n\in [d]^n$, the expression $A_1^{j_1}A_2^{j_2}\ldots A_n^{j_n}$ yields a complex number (since it is of the form $\bra{v_1}V_2\cdots V_{n-1}\ket{v_n}$), i.e. it yields the \emph{amplitude} corresponding to $\vec{j}$. Thus, the amplitudes are encoded as products of matrices, justifying the name \emph{Matrix Product State}. Some additional terminology: The indices $j_i$ are referred to as \emph{physical} indices, as their dimension $d$ is fixed by the physical system. The value $D$ is called the \emph{bond dimension}, which we discuss in more depth shortly. Graphically, an MPS is given by Figure~\ref{fig:disp2}(d), where the vertical lines denote physical indices, and the horizontal lines represent tensor contractions or matrix products.

With a bit of thought, one can see that \emph{any} state $|\psi\rangle\in(\complex^d)^{\otimes n}$ can be written as an MPS exactly if the bond dimension $D$ is chosen large enough. Indeed, this can be achieved by setting $D$ to be at least the maximum Schmidt rank of $|\psi\rangle$ across any bipartite cut. In general, however, such a value of $D$ unfortunately grows exponentially with $n$, and thus large values of $D$ are not computationally feasible. The strength of MPS is hence as follows: Any $n$-particle quantum states whose entanglement across bipartite cuts is ``small'' (i.e. of polynomial Schmidt rank in $n$) can be represented succinctly by an MPS.

Moreover, this niche filled by MPS turns out to be quite interesting, as condensed matter physicists are mainly interested in ground states which are highly \emph{non}-generic. For example, recall that in 1D gapped systems, we have an area law~\cite{Ha07, ECP10, BH12, AKLV13, Hua14}, implying that in the 1D setting the entanglement entropy across any bipartite cut is bounded by a constant independent of $n$. In particular, Reference~\cite{AKLV13} shows that ground states of $1$D gapped systems with energy gap $\epsilon$ can be well approximated by an MPS with sublinear bond dimension $D=\exp\left(\tilde{O}\left(\frac{\log^{3/4}n}{\epsilon^{1/4}}\right)\right)$. In conformally invariant critical (gapless) systems, the area law is slightly violated with a logarithmic factor $\sim\log n$~\cite{CC04, CC09}, suggesting that MPS is still a fairly efficient parametrization.

Finally, a key property of MPS is that, given an MPS description of a quantum state, we can \emph{efficiently} compute its physical properties, such as energy, expectation of order parameters, correlation functions, and entanglement entropy~\cite{Sch11}. This is in sharp contrast to more complicated tensor networks such as Projected Entangled Pair States (PEPS), which are known to be $\#$P-complete to contract~\cite{SWVC07}.

\subsection{Implementation of DMRG}\label{IDMRG}

Having introduced MPS, we now briefly review the idea behind DMRG from an MPS perspective. Specifically, given an input Hamiltonian $H$, we compute the MPS of some bond dimension $D$ that best approximates the ground state by minimizing the energy $\langle\psi|H|\psi\rangle$ with respect to all such MPS $|\psi\rangle$, i.e., with respect to $O(ndD^2)$ parameters. Note that in general the bond dimension $D$ must to grow with the system size $n$ (especially in critical/gapless systems). Unfortunately, if $D=\poly(n)$ the aforementioned minimization problem can be NP-hard even for frustration-free Hamiltonians~\cite{SCV08}. To cope with this, DMRG is thus a \emph{heuristic} algorithm for finding local minima: There is no guarantee that the local minima we find are global minima, nor that the algorithm converges rapidly. However, perhaps surprisingly, in practice DMRG works fairly well even in critical/gapless systems.

At a high level, the DMRG algorithm proceeds as follows. We start with an MPS denoted by $\{A_{i=1,2,\ldots,n}^{j=1,2,\ldots,d}\}$, and subsequently perform a sequence of local optimizations. A local optimization at site $i_0$ means minimizing $\langle\psi|H|\psi\rangle$ with respect to $A_{i_0}^{j=1,2,\ldots,d}$, while keeping all other matrices $A_{i\neq i_0}^{j=1,2,\ldots,d}$ fixed. Such local optimizations are performed in a number of ``sweeps'' until our solution $\{A_{i=1,2,\ldots,n}^{j=1,2,\ldots,d}\}$ converges. Here, a \emph{sweep} consists of applying local optimizations in sequence starting at site $1$ up to site $n$, and then backwards back to site $1$. In other words, we apply the optimization locally in the following order of sites: $1,2,\ldots,n-1,n,n-1,\ldots,2,1$.

\section{Multi-Scale Entanglement Renormalization Ansatz}\label{sscn:MERA}

We now discuss a specific type of tensor network known as the Multi-Scale Entanglement Renormalization Ansatz (MERA)~\cite{V07,V08} (see \S\ref{sscn:tensor} for a definition of tensor networks), which falls somewhere between MPS and PEPS. Like MPS and unlike PEPS, the expectation value of local observables for MERA states can be computed exactly efficiently. Like PEPS and unlike MPS, MERA can be used to well-approximate (certain) states in $D$-dimensional lattices for $D\geq 1$. It should be noted that, as with MPS and PEPS, there is not necessarily any guarantee as to how well MERA can approximate a particular state; rather, as with many ideas in physics, MERA is an intuitive idea which appears to work well for certain Hamiltonian models, such as the 1D quantum Ising model with transverse magnetic field on an infinite lattice~\cite{V07}.

There are two equivalent ways to think about MERA. The first is to give an efficient (log-depth) quantum circuit which, given a MERA description of a state $\ket{\psi}$, prepares $\ket{\psi}$ from the state $\ket{0}^{\otimes n}$. The disadvantage of this view, however, is that it does not yield much intuition as to why MERA is structured the way it is. The second way to think about MERA is through a physics-motivated view in terms of DMRG; as this view provides the beautifully simple rationale behind MERA, we present it first.

\paragraph{The DMRG-motivated view.} This viewpoint is presented in~\cite{V07}, and proceeds as follows. To begin, the general idea behind Wilson's real-space renormalization group (RG) methods (see~\cite{WN92}) is to partition the sites of a given quantum system into \emph{blocks}. One then simplifies the description of this space by truncating part of the Hilbert space corresponding to each block; this process is known as \emph{coarse-graining}. The entire coarse-graining procedure is repeated iteratively on the new lower dimensional systems produced, until one obtains a polynomial-size (approximate) description (in the number of sites, $n$) of the desired system.

The key insight of White~\cite{W92,W93} was to realize that for 1D systems, the ``correct'' choice of truncation procedure on a block $B$ of sites is to simply discard the Hilbert space corresponding to the ``small'' eigenvalues of $\rho_B$, where $\rho_B$ is the reduced state on $B$ of the initial $n$-site state $\ket{\psi}$. Here, the value of ``small'' depends on the approximation precision desired in the resulting tensor network representation. Intuitively, such an approximation works well if $B$ in $\ket{\psi}$ is not highly entangled with the remaining sites. When this condition does not hold, however, DMRG seems to be in a bind. The idea of MERA is hence to precede each truncation step by a \emph{disentangling} step, i.e. by a local unitary which attempts to reduce the amount of entanglement along the boundary of $B$ between $B$ and the remaining sites before the truncation is carried out.

More formally, MERA is defined on a $D$-dimensional lattice $L$ as follows~\cite{V07}. For simplicity, we restrict our attention to the case of $D=1$ on spin-$(1/2)$ systems with periodic boundary conditions, but the ideas here extend to $D\geq 1$ on higher dimensional systems. Let $L$ correspond to  Hilbert space $V_L\equiv \bigotimes_{s\in L}V_s$, where $s\in L$ denote the lattice sites with respective finite-dimensional Hilbert spaces $V_s$. Consider now a block $B\subset L$ of neighboring sites, whose Hilbert space we denote as $V_B\equiv \bigotimes_{s\in B}V_s$. For simplicity, let us assume $B$ consists of two sites $s_1$ and $s_2$, with neighboring sites $s_0$ and $s_3$ immediately to the left and right, respectively. The disentangling step is performed by carefully choosing unitaries $U_{01},U_{23}\in\unitary{\complex^4}$ (the specific choice of $U_{01},U_{23}$ depends on the input state $\ket{\psi}$), and applying $U_{ij}$ to consecutive sites $i$ and $j$. The \emph{truncation} step follows next by applying isometry $V_{12}:\lin{\complex^4}\mapsto\lin{\complex^2}$ to sites $1$ and $2$, where  $\complex^2$ is the truncated space we wish to keep and where $V_{12}V_{12}^\dagger=I$. By applying this procedure to neighboring disjoint pairs of spins, we obtain a new spin chain with $n/2$ sites (assuming $n$ is even in this example). The entire procedure is now repeated on these $n/2$ coarse-grained sites. After $O(\log n)$ iterations, we end up with a single site. The tensor network is then obtained by writing down tensors corresponding to the linear maps of each $U_{ij}$ and $V_{ij}$, and connecting these tensors according to the geometry underlying the process outlined above. The resulting tensor network has a tree-like structure, with a single vertex at the top, and $n$ legs at the bottom corresponding to each of the $n$ original sites. It is depicted in Figure~\ref{fig:MERA1}(a).

\begin{figure}\centering
  \includegraphics[height=3.5cm]{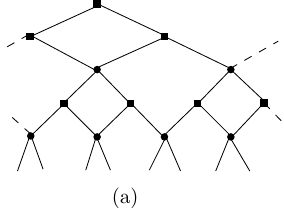}\hspace{10mm}
  \includegraphics[height=3.5cm]{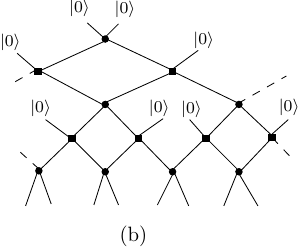}\hspace{10mm}
  \includegraphics[height=3.5cm]{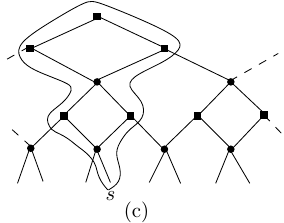}
  \caption{A MERA network on $8$ sites. The circle vertices represent ``disentangling'' unitaries. The square vertices represent isometries. (a) The tensor network view. (b) The quantum circuit view. (c) The causal cone of the site labeled $s$.}\label{fig:MERA1}
\end{figure}

Note that if we assume the bond dimension for each isometric tensor $V_{ij}$ is $d$, then the MERA representation requires $O(nd^4)$ bits of storage; this is because there are $2n-1$ tensors in the network, and each tensor stores at most $O(d^4)$ complex numbers.

\paragraph{The quantum circuit view.} In a sense, we have cheated the reader, because the DMRG view already prescribes the method for the quantum circuit view of MERA. Specifically, imagine we reverse the coarse-graining procedure described above, i.e. instead of working our way from the $n$ sites of $\ket{\psi}$ up to a single site, we go in the opposite direction. Then, intuitively, the DMRG view yields a quantum circuit which, starting from the state $\ket{0}^{\otimes n}$, prepares (an approximation to) the desired state $\ket{\psi}$ via a sequence of the same isometries and unitaries prescribed by the tensor network. This view is depicted in Figure~\ref{fig:MERA1}(b).

\paragraph{Computing with MERA.} A succinct representation of a quantum system would not necessarily be useful without the ability to compute \emph{properties} of the system from this succinct format. A strength of MERA is that, indeed, expectation values of local observables against $\ket{\psi}$ can be efficiently computed. This follows simply because given a MERA network $M$ representing $\ket{\psi}$, the reduced state of $\ket{\psi}$ on $\Theta(1)$ sites can be computed in time $O(\log n)$ (assuming the dimension $D$ of our lattice is considered a constant). To see this, we partition the tensors in our MERA network in terms of horizontal layers or \emph{time slices} from top to bottom. Specifically, in Figure~\ref{fig:MERA1}(b), time slice $0$ is before the top unitary is run, slice $1$ immediately after the top unitary is run and before the following pair of isometries are run, and so forth until slice $5$, which is immediately after the four bottom-most unitaries are run. Then, in each layer, the causal cone $C_s$ for any site $s$ can be shown to have at most \emph{constant} width (more generally, at most $4\cdot3^{D-1}$ width~\cite{V08}). Here, the \emph{causal cone} of $C_s$ is the set of vertices and edges in the network which influence the leg of the network corresponding to site $s$; see Figure~\ref{fig:MERA1}(c). The \emph{width} of $C_s$ in a time slice is the number of edges in $C_s$ in that slice. Thus, by viewing the MERA network in terms of the quantum circuit view, we see that the reduced state on site $s$ is given by a quantum circuit with $O(\log n)$ gates. Moreover, at any point in the computation, this circuit needs to keep track of the state of only $\Theta(1)$ qubits. Such a circuit can be straightforwardly simulated classically in $O(\log n)$ time via brute force (i.e. multiply the unitaries in the circuit and trace out qubits which are no longer needed), yielding the claim.

\section{Area laws}\label{sscn:arealaw_overview}

We now discuss {area laws} in further depth, which were first mentioned briefly at the end of \S\ref{sscn:physicshistory}. Recall that, roughly, an \emph{area law} states that for certain interesting classes of quantum many-body systems, the amount of entanglement
between a subsystem and its complement grows as the \emph{surface area or the boundary} rather than the \emph{volume} of the subsystem. In literature, the amount of entanglement is typically formulated as entanglement entropy, i.e.~the von Neumann entropy of the reduced density matrix on the subsystem of interest.

For example, suppose that a region $L$ of a quantum lattice system is carved out as in Figure~\ref{fig:arealaw1}. An arbitrary quantum state on such a lattice follows a volume law instead of an area law; to see this, simply consider the case in which every particle in $L$ forms an EPR pair with some other particle in $\overline L$. In this case, the amount of entanglement between $L$ and $\overline L$ is proportional to the number of particles in $L$, i.e.~the \emph{volume} of $L$, rather than the size of the boundary $|\partial L|$.  However, it is widely believed in the physics community that many interesting classes of quantum states do satisfy an area law, most notably the ground states of gapped local Hamiltonians. Note, however, that for critical/gapless systems, there are known constructions which violate the area law~\cite{AGIK09,I10,GI09}.

\begin{figure}\centering
  \includegraphics{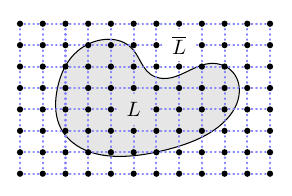}
  \caption{A quantum lattice system partitioned into two parts $L$ and $\overline L$. If the represented quantum state obeys an area law, the amount of entanglement between $L$ and $\overline L$ is bounded above by a quantity that is proportional to the surface area $|\partial L|$, rather than the volume $|L|$.}
  \label{fig:arealaw1}
\end{figure}

\paragraph{Motivations.} The idea that the information content of a region depends on its surface area rather than its volume is not foreign to physics. This intuition, often dubbed the \emph{holographic principle}, is inspired by black hole thermodynamics, where the entropy of a black hole is believed to scale as the surface area of the event horizon rather than the volume of the black hole. In fact, the original interest in area laws in quantum systems also sprang from this intriguing analogy, in connection with the conjecture that the origin of the black hole entropy is the quantum entanglement between the inside and outside of the black hole~\cite{Blackhole1, Blackhole2, Blackhole3}.

Aside from black holes, however, there are other fundamental reasons to study area laws. For instance, area laws are a means for characterizing the structure of entanglement occurring in naturally arising systems. For example, it is well known that a generic state in Hilbert space obeys a volume law rather than an area law~\cite{HLW:aspects}. Therefore, the existence of an area law for most physically relevant systems would imply that much of nature's ``interesting physics'' takes place in a small corner of Hilbert space.

Last but not least, a significant motivation lies in  the classical representability of quantum many-body systems. In contrast to general quantum states that require exponentially many parameters to describe, states that satisfy an area law may be expected to admit an efficient classical description via (e.g.) tensor networks of small bond dimension. In fact, area laws are closely related to tensor networks. For instance, any tensor network with constant bond dimension automatically satisfies an area law in the following sense; if we pick any subset $L$ of vertices, the Schmidt rank between $L$ and $\overline L$ is upper-bounded by $D^{|\partial L|}$, where $D$ represents the bond dimension. This immediately implies that the entanglement entropy is bounded by $|\partial L| \log D$, giving rise to an area law. This bound also shows that tensor networks such as MPS or PEPS, even with polynomially large bond dimension, cannot describe entanglement that scales via a volume law.

\paragraph{Why should area laws hold?} Area laws are a way of formulating (albeit not exactly) the physics intuition that entanglement in ``natural'' quantum systems should roughly live on the boundary. For example, a seminal work of Hastings \cite{H:exp_decay} shows that the ground state of a gapped local Hamiltonian on a lattice of any dimension exhibits an exponential decay of correlations, i.e. two-point correlation functions decay exponentially with respect to spatial separation between the two points. At first glance, this result already seems to suggest that the entanglement between a region and its complement should live mostly on or close to the boundary; unfortunately, it is far from trivial to lift this statement about correlations to one about entanglement.

\paragraph{Known results.} In general, establishing area laws is a challenging task. While a number of specific Hamiltonian models had been shown to obey an area law \cite{ECP10,VLRK03,IJK05,FKR04}, it was not until 2007 that Hastings proved~\cite{Ha07} that the ground states of \emph{all} 1D gapped local Hamiltonians obey an area law. (Hastings~\cite{Ha07} also showed that the ground state of such a system can be well approximated by an MPS of polynomial bond dimension.) Subsequently, Aharonov, Arad, Landau and Vazirani gave an alternative combinatorial proof (whereas Hastings' proof relies on Lieb-Robinson bounds) of Hastings' 1D area law in the frustration-free case~\cite{AALV11, ALV12}. Yet another combinatorial proof was given by Arad, Kitaev, Landau, and Vazirani~\cite{AKLV13}, which also applies in frustrated settings and improves significantly on the area law parameters obtained. The approach of~\cite{AKLV13} was later adapted to the setting of a constant-fold degenerate ground space by Huang~\cite{Hua14}. Here, any state in the ground space obeys an area law. As an aside, we remark that this study of combinatorial proofs of the area law (and in particular, the concept of an \emph{approximate ground state projection}, or AGSP) played an instrumental role in the development of the first rigorous polynomial time algorithm for finding the ground state of a gapped 1D system~\cite{LVV13}. In 2012, Brand\~ao and Horodecki~\cite{BH12} showed that 1D local Hamiltonians that have an exponential decay of correlations satisfy an area law. Together with Hastings' result~\cite{H:exp_decay} that gapped local Hamiltonians exhibit exponential decay of correlations, this yields still another proof of Hastings' 1D area law~\cite{Ha07}. Finally, Wolf~\cite{W06} showed that translationally invariant critical fermionic systems of any spatial dimension satisfy an area law up to a logarithmic correction.

Interestingly, in contrast to the ground state case just discussed, \emph{thermal states} of local Hamiltonians always satisfy an area law regardless of the energy gap or the spatial dimension \cite{MSVC14, WVHC08}.\\

For further details on combinatorial area law proofs, the reader is referred to \S\ref{sscn:arealaw}, where the result of~\cite{AKLV13} is sketched. For a physics-oriented survey of area laws, see Eisert, Cramer, and Plenio~\cite{ECP10}.

\paragraph{Open problems.} Arguably the most important open question in this area is whether an area law holds in dimensions larger than $1$D. Unfortunately, the known approaches do not seem generalize easily to this setting. As for alternate approaches, we list two possible routes. First, Masanes~\cite{masanes09} has shown that the following two criteria are sufficient to establish an area law: (1) There is sufficiently fast decay of correlations, and (2) the number of ``low-energy'' states is not exponentially large. Second, Van Acoleyen, Mari\"en, and Verstraete~\cite{AMV13} have recently shown that if two gapped systems $H_1$ and $H_2$ are adiabatically connected, i.e.~there is a smooth path between $H_1$ and $H_2$ such that the system remains gapped at every point on the path, then the ground state of $H_1$ satisfies an area law if and only if the ground state of $H_2$ also satisfies an area law.

\chapter{Reviews of Selected Results}\label{scn:results}
Having discussed a number of Hamiltonian complexity concepts originating from the physics literature in \S\ref{scn:motivation} and~\ref{scn:physicsforCS}, we now expound on a selected number of central computer science-based results in the area. Specifically, \S\ref{sscn:5LH} reviews Kitaev's original proof that $5$-local Hamiltonian is QMA-complete. \S\ref{sscn:2LH} discusses the ensuing proof by Kitaev, Kempe and Regev using perturbation theory-based gadgets that $2$-local Hamiltonian is QMA-complete. In \S\ref{sscn:CLH}, we review Bravyi and Vyalyi's Structure Lemma and its use in proving that the $2$-local commuting Hamiltonian problem is in NP, and thus unlikely to be QMA-hard. \S\ref{sscn:Q2SAT} gives a quantum information theoretic presentation of Bravyi's polynomial time algorithm for Quantum 2-SAT. Finally, \S\ref{sscn:arealaw} provides an intuitive review of the combinatorial proof of an area law for 1D gapped systems due to Arad, Kitaev, Landau and Vazirani.

\section{5-local Hamiltonian is QMA-complete}\label{sscn:5LH}

One of the cornerstones of classical computational complexity theory is the Cook-Levin theorem, which states that classical constraint satisfaction is NP-complete. The quantum version of this theorem is due to Kitaev, who showed that the $5$-local Hamiltonian problem is QMA-complete~\cite{KSV02}. In this section, we review Kitaev's proof. For a more in-depth treatment, we refer the reader to the detailed surveys of Aharonov and Naveh~\cite{AN02} and Gharibian~\cite{G13}.

We begin by showing that $k$-local Hamiltonian for any $k\in\Theta(1)$ is in QMA.

\paragraph{Local Hamiltonian is in QMA}

\begin{theorem}(Kitaev~\cite{KSV02}) For any constant $k\geq 1$, $\klhh\in\qma$.
\end{theorem}
\noindent \emph{Proof sketch.} The basic idea is that whenever we have a YES instance of $\klhh$, the quantum proof sent to the verifier is essentially the ground state of the local Hamiltonian $H$ in question. The verifier then runs a simple ``local'' version of phase estimation to roughly determine the energy penalty incurred by the given proof.

To begin, suppose we have an instance $(H,a,b)$ of $\klhh$ with $k$-local Hamiltonian $H=\sum_{j=1}^r H_j\in \LL((\B)^{\otimes n})$. We construct an efficient quantum verification circuit $V$ as follows. First, the quantum proof is $\ket{\psi}\in\C^r\otimes(\B)^{\otimes n}\otimes\B$, s.t.
\begin{equation}
    \ket{\psi}=\left(\frac{1}{\sqrt{r}}\sum_{j=1}^{r}\ket{j}\right)\otimes\ket{\eta}\otimes\ket{0},
\end{equation}
for $\set{\ket{j}}_{j=1}^r$ an orthonormal basis for $\complex^r$, and $\ket{\eta}$ an eigenvector corresponding to some eigenvalue $\lambda$ of $H$. We call the first register of $\ket{\psi}$ the \emph{index} register, the second the \emph{proof} register, and the last the \emph{answer} register. The circuit $V$ is defined as $V:=\sum_{j=1}^r\ketbra{j}{j}\otimes W_j$, where $W_j$ is defined as follows. For our Hamiltonian $H=\sum_{j=1}^{r}H_j$, suppose $H_j$ has spectral decomposition $H_j=\sum_{s}\lambda_s \ketbra{\lambda_s}{\lambda_s}$. Then, define $W_j$ acting on the proof and answer registers with action
\begin{equation}
    W_j\left(\ket{\lambda_s}\otimes\ket{0}\right) = \ket{\lambda_s}\otimes\left(\sqrt{\lambda_s}\ket{0} + \sqrt{1-\lambda_s} \ket{1}\right).
\end{equation}

\begin{Question}
    Show that if we apply $V$ to the proof $\ket{\psi}$ and measure the answer register in the computational basis, the probability of obtaining outcome $1$ is $1- \frac{1}{r}\bra{\eta} H\ket{\eta}$. Conclude that since the thresholds $a$ and $b$ are inverse polynomially separated,  $\klhh\in\class{QMA}$.
\end{Question}
\begin{Hint}
    Observe that since we may assume the index register is implicitly measured at the end of the verification, $V$ above can be thought of as using the index register to choose an index $j\in[r]$ uniformly at random, followed by applying $W_j$ to the proof register. As a result, the probability of outputting $1$ can be expressed as
\begin{equation}
    \pr(\text{output 1}) = \sum_{j=1}^r\frac{1}{r}\pr(\text{output 1}\mid W_j\text{ is applied}).
\end{equation}
\end{Hint}
\begin{Hint}
    When considering the action of any $W_j$ on $\ket{\eta}$, rewrite $\ket{\eta}$ in the eigenbasis of $H_j$ as $\ket{\eta}=\sum_s \alpha_s\ket{\lambda_s}$ (the values of the coefficients $\alpha_s$ will not matter).
\end{Hint}

\paragraph{$5$-Local Hamiltonian is QMA-hard}\mbox{ }\\

\noindent To show that $5$-local Hamiltonian is QMA-hard, Kitaev gives a quantum adaptation of the Cook-Levin theorem~\cite{KSV02}. Specifically, he shows a polynomial-time many-one or Karp reduction from an arbitrary problem in QMA to $\flh$, which we now discuss.

Let $P$ be a promise problem in $\qma$, and let $V=V_LV_{L-1}\dots V_1$ be a verification circuit for $P$ composed of unitaries $V_k$. Without loss of generality, we may assume each $V_k$ acts on pairs of qubits, and that $V\in \UU((\B)^{\otimes m}\otimes (\B)^{\otimes N-m})$, where the $m$-qubit register contains the proof $V$ verifies, and the remaining qubits are ancilla qubits. Using $V$, our goal is to define a $5$-local Hamiltonian $H$ that has a small eigenvalue if and only if there exists a proof $\ket{\psi}\in(\B)^{\otimes m}$ causing $V$ to accept with high probability.

We let $H$ act on $(\B)^{\otimes m}\otimes (\B)^{\otimes N-m}\otimes\complex^{L+1}$, which is simply the initial space $V$ acts on, tensored with an $(L+1)$-dimensional \emph{counter} or \emph{clock} register. We label the three registers $H$ acts on as $p$ for proof, $a$ for ancilla, and $c$ for clock, respectively. We now define $H$ itself:
\begin{equation}
    H := \hin + \hprop + \hout,
\end{equation}
with the terms $\hin$, $\hprop$, and $\hout$ defined as follows. Let
\begin{eqnarray}
    \hin&:=&I_p\otimes
    \left(I_a - \ketbra{0\dots0}{0\dots0}_a\right)\otimes \ketbra{0}{0}_c\\
    \hout&:=&\left(\ketbra{0}{0}\otimes I_{(\B)^{\otimes m-1}}\right)_p\otimes
     I_a\otimes \ketbra{L}{L}_c\\
    \hprop &:=& \sum_{j=1}^{L} H_j \text{,~~~~where}\nonumber\\
    H_j &:=& -\frac{1}{2}V_j\otimes\ketbra{j}{j-1}_c -\frac{1}{2}V_j^\dagger\otimes\ketbra{j-1}{j}_c +\\&&\frac{1}{2}I\otimes(\ketbra{j}{j}+\ketbra{j-1}{j-1})_c.\label{eqn:H_j}
\end{eqnarray}

\begin{Question}
Suppose that for any YES-instance of promise problem $P$, $V$ accepts a valid proof $\ket{\psi}$ with \emph{certainty}. Verify that the following state $\ket{\eta}$, known as the \emph{history state}, lies in the null space of $H$. Why do you think $\ket{\eta}$ is called the history state?
\begin{equation}\label{eqn:eta}
    \ket{\eta} := \frac{1}{\sqrt{L+1}}\sum_{j=0}^{L}\left(V_j\dots V_1 \ket{\psi}_p\otimes\ket{0}^{\otimes N-m}_a\right)\otimes \ket{j}_c.
\end{equation}
\end{Question}

In order to ease the analysis of $H$'s smallest eigenvalue, it turns out to be extremely helpful to apply the following change-of-basis operator to $H$:
\begin{equation}
    W=\sum_{j=0}^{L} V_j\dots V_1\otimes\ketbra{j}{j}_c.
\end{equation}

\begin{Question}
    Show that:
    \begin{enumerate}
        \item $\ket{\hat{\eta}}:=W\ket{\eta}=\ket{\psi}_p\otimes\ket{0}^{\otimes N-m}_a\otimes\ket{\gamma}_c$, where we define
$
    \ket{\gamma} := \frac{1}{\sqrt{L+1}}\sum_{j=0}^L\ket{j}.
$
        \item $\hat{H}_{\rm in}:=W^\dagger \hin W=\hin$,
        \item $\hat{H}_{\rm out}:=W^\dagger \hout W= (V^\dagger\otimes I_c)\hout(V\otimes I_c)$,
        \item $\hat{H}_j:=W^\dagger H_j W = I_{p,a}\otimes\frac{1}{2}(\ketbra{j-1}{j-1}-\ketbra{j-1}{j}-\ketbra{j}{j-1}+\ketbra{j}{j})_c$, and hence $\hat{H}_{\rm prop}$ equals
        \begin{equation}
    I_{p}\otimes I_a\otimes\left(
      \begin{array}{cccccc}
        \frac{1}{2} & -\frac{1}{2} & 0 & 0 & 0 & \dots \\
        -\frac{1}{2} & 1 & -\frac{1}{2} & 0 & 0 & \dots \\
        0 & -\frac{1}{2} & 1 & -\frac{1}{2} & 0 & \dots \\
        0 & 0 & -\frac{1}{2} & 1 & -\frac{1}{2} & \dots \\
        0 & 0 & 0 & -\frac{1}{2} & \ddots & \ddots \\
        \vdots & \vdots & \vdots & \vdots & \ddots & \ddots \\
      \end{array}
    \right)=:I_{p}\otimes I_a\otimes E_c,
\end{equation}
where we have let $E$ denote a tridiagonal matrix acting on the clock register.
    \end{enumerate}
\end{Question}

Henceforth, when we refer to $H$, $\hin$, $\hout$, $\hprop$, $H_j$, and $\ket{\eta}$, we implicitly mean $\hat{H}$, $\hat{H}_{\rm in}$, $\hat{H}_{\rm out}$, $\hat{H}_{\rm prop}$, $\hat{H}_j$, and $\ket{\hat{\eta}}$, respectively.\\

\paragraph{The YES case: $H$ has a small eigenvalue}

\begin{Question}
    Suppose that given proof $\ket{\psi}$, $V$ accepts with probability at least $1-\epsilon$ for $\epsilon\geq 0$. Show that
    \begin{equation}
    \bra{\eta}H\ket{\eta} \leq \frac{1}{L+1}\epsilon.
    \end{equation}
    Conclude that if there exists a proof $\ket{\psi}$ accepted with ``high'' probability by $V$, then $H$ has a ``small'' eigenvalue.
\end{Question}

\paragraph{The NO case: $H$ has no small eigenvalues}\mbox{ }\\

\noindent If there is no proof $\ket{\psi}$ accepted by $V$ with high probability, then we wish to show that $H$ has no small eigenvalues. To do so, write $H=A_1+A_2$ for $A_1 := \hin + \hout$ and $A_2 := \hprop$. If $A_1$ and $A_2$ were to commute, then analyzing the smallest eigenvalues of $A_1$ and $A_2$ independently would yield a lower bound on the smallest eigenvalue of $H$. Unfortunately, $A_1$ and $A_2$ do not commute; hence, if we wish to use information about the spectra of $A_1$ and $A_2$ to lower bounds $H$'s eigenvalues, we will need a stronger technical tool, given below.

\begin{lemma}[Kitaev~\cite{KSV02}, Geometric Lemma, Lemma 14.4]\label{l:pluslemma}
    Let $A_1,A_2\succeq 0$, such that the minimum \emph{non-zero} eigenvalue of both operators is lower bounded by $v$. Assume that the null spaces $\nl$ and $\nll$ of $A_1$ and $A_2$, respectively, have trivial intersection, i.e.\ $\nl\cap\nll=\set{\ve{0}}$. Then
    \begin{equation}
        A_1+A_2 \succeq 2v\sin^2\frac{\alpha(\nl,\nll)}{2}I\enspace,
    \end{equation}
    where the \emph{angle} $\alpha(\spa{X},\spa{Y})$ between $\spa{X}$ and $\spa{Y}$ is defined over vectors $\ket{x}$ and $\ket{y}$ as \[\cos\left[\angle(\spa{X},\spa{Y})\right] := \max_{\substack{\ket{x}\in\spa{X},\ket{y}\in\spa{Y}\\\norm{\ket{x}}=\norm{\ket{y}}=1}}\abs{\braket{x}{y}}.\]
\end{lemma}

\begin{Question}
    For complex Euclidean spaces $\spa{X}$ and $\spa{Y}$, is the statement $\spa{X}\cap\spa{Y}=\set{\ve{0}}$ equivalent to $\spa{X}$ and $\spa{Y}$ being orthogonal spaces?
\end{Question}

We use Kitaev's Geometric Lemma with $A_1=\hin + \hout$ and $A_2=\hprop$ to lower bound the smallest eigenvalue of $H$ in the NO case.

\begin{Question}
    For $A_1=\hin + \hout$ and $A_2=\hprop$, what non-zero value of $v$ can we use for the Geometric Lemma?
\end{Question}
\begin{Hint}
    For $A_1$, recall that commuting operators simultaneously diagonalize.
\end{Hint}
\begin{Hint}
    For $A_2$, the eigenvalues are given by $\lambda_k=1-\cos[\pi k/(L+1)]$ for $0\leq k\leq  L$. Use this to show that the smallest \emph{positive} eigenvalue of $A_2$ is at least $1-\cos(\pi/(L+1))\geq c/L^2$.
\end{Hint}
\begin{Question}
    In order to compute $\alpha(\nl,\nll)$, convince yourself first that
    \begin{eqnarray}
    \nl &=& \left[((\B)^{\otimes m})_p\otimes \ket{0}^{\otimes N-m}_a\otimes \ket{0}_c\right]\oplus
    \nonumber\\
                    &&\left[ ((\B)^{\otimes N})_{p,a}\otimes \operatorname{span}(\ket{1},\ldots \ket{L-1})_c\right]\oplus\nonumber\\
                    && \left[
                    V^\dagger (\ket{1}\otimes(\B)^{\otimes N-1})_{p,a}\otimes \ket{L}_c\right],\label{eqn:nl}\\
    \nll &=& ((\B)^{\otimes N})_{p,a}\otimes \ket{\gamma}_c,\label{eqn:nll}
\end{eqnarray}
\end{Question}

\noindent To compute $\sin^2\frac{\alpha(\nl,\nll)}{2}$ for the Geometric Lemma, we now upper bound
\[
        \cos^2\alpha(\nl,\nll) \leq 1-\frac{1-\sqrt{\epsilon}}{L+1}.
\]
\begin{Question}
    Show that     $\cos^2\alpha(\nl,\nll)=
    \operatorname{max}_{\substack{\ket{y}\in\nll\\\norm{\ket{y}}=1}}\bra{y}\Pi_{\nl}\ket{y}.\label{eqn:6}$
\end{Question}
    Observe from Equation~\ref{eqn:nl} that $\nl$ is a direct sum of three spaces, and hence the projector onto $\nl$ can be written as the sum of three respective projectors $\Pi_1+\Pi_2+\Pi_3$.
\begin{Question}
    Observe by Equation~\ref{eqn:nll} that for any $\ket{y}\in\nll$, $\ket{y}=\ket{\zeta}_{p,a}\otimes\ket{\gamma}_c$ for some $\ket{\zeta}\in(\B)^{\otimes m}\otimes(\B)^{\otimes {N-m}}$.
    \begin{enumerate}
        \item Show that $\bra{y}\Pi_{1}\ket{y}=\frac{L-1}{L+1}$.
        \item One can show that \[\bra{y}\Pi_2+\Pi_3\ket{y}\leq \cos^2\varphi(\mathcal{K}_1,\mathcal{K}_2),\] where $\mathcal{K}_1=(\B)^{\otimes m}\otimes \ket{0}^{\otimes N-m}$ and $\mathcal{K}_2=V^\dagger \ket{1}\otimes(\B)^{\otimes N-1}$. Use the fact that in the NO case, any proof is accepted by $V$ with probability at most $\epsilon$ to conclude that
            \[
                \bra{y}\Pi_2+\Pi_3\ket{y}\leq \frac{1}{L+1}(1+\sqrt{\epsilon}).
            \]
    \end{enumerate}
\end{Question}

    \noindent Combining the results of the question above, we have that $\cos^2\alpha(\nl,\nll)\leq1-((1-\sqrt{\epsilon})/(L+1))$. Using the identities $\sin^2 x + \cos^2 x=1$ and $\sin (2x)=2\sin x\cos x$, this yields
    \[
    \sin^2\frac{\alpha(\nl,\nll)}{2}\geq \frac{1}{4}\sin^2\alpha(\nl,\nll)\geq \frac{1-\sqrt{\epsilon}}{4(L+1)}.
    \]
We conclude that in the NO case, the minimum eigenvalue of $H$ scales as $\Omega((1-\sqrt{\epsilon})/L^3)$.
\begin{Question}
    Recall that in the YES case, the smallest eigenvalue of $H$ is upper bounded by ${\epsilon}/(L+1)$. Why do the eigenvalue bounds we have obtained in the YES and NO cases thus suffice to show that $\flh$ is QMA-hard?
\end{Question}

\paragraph{Making $H$ $5$-local}\mbox{ }\\

\noindent We are almost done! The only remaining issue is that we would like $H$ to be $5$-local, but a binary representation of the $(L+1)$-dimensional clock register is unfortunately $\log(n)$-local. To alleviate this~\cite{KSV02}, we switch to a \emph{unary} representation of time. In other words, we now let $H$ act on $(\B)^{\otimes m}\otimes (\B)^{\otimes N-m}\otimes(\B)^{\otimes L}$, where the counter register is now given in \emph{unary}, i.e.\ $\ket{j}\in\complex^{L+1}$ is represented as
\begin{equation}
    |\underbrace{1,\ldots,1}_{j},0,\ldots,0\rangle.\label{eqn:7}
\end{equation}
The operator basis $\ketbra{i}{j}$ for $\LL(\complex^{L+1})$ translates easily to this new representation (omitted here; see Reference~\cite{KSV02}). To enforce the clock register to indeed always be a valid representation of some time $j$ in unary, we add a new fourth penalty term to $H$ which acts only on the clock register, namely
\begin{equation}
    \hstab := I_{p,a}\otimes \sum_{j=1}^{L-1}\ketbra{0}{0}_j\otimes\ketbra{1}{1}_{j+1}.
\end{equation}
Hence, the new $H$ is given by $H=\hin+\hprop+\hout+\hstab$. By using the fact that both $\hin+\hprop+\hout$ and $\hstab$ act invariantly on the original space the old $H$ used to act on, it is a fairly straightforward exercise to verify that the analysis obtained above goes through for this new definition of $H$ as well~\cite{KSV02}. We conclude that $\flh$ is QMA-hard.

\section{2-local Hamiltonian is QMA-complete}\label{sscn:2LH}
In~\cite{KSV02}, Kitaev showed that the $5$-local Hamiltonian problem is QMA-complete. In this section, we review Kitaev, Kempe, and Regev's perturbation theory proof that even $2$-local Hamiltonian is QMA-hard~\cite{KKR06}. Note that Reference~\cite{KKR06} also provides an alternative ``simpler'' proof based on elementary linear algebra and the so-called \emph{Projection Lemma} in the same paper; however, as the Projection Lemma can be derived via perturbation theory, and since Reference~\cite{KKR06}'s idea of using perturbation theory gadgets has since proven useful elsewhere in Hamiltonian complexity~(e.g.~\cite{OT05}), we focus on the latter proof technique. Besides, perturbation theory is a standard tool in a physicist's toolbox, and our goal in this survey is to better understand what goes on in physicists' minds!

The proof that $2$-local Hamiltonian is QMA-complete is quite complicated. To aid in its assimilation, we therefore begin with a high-level overview of how the pieces of the proof fit together. The starting point shall be Kempe and Regev's result~\cite{KR03} that $3$-LH is also QMA-complete. Roughly, the latter result is obtained via a circuit-to-Hamiltonian construction similar to Kitaev's from \S\ref{sscn:5LH}, except that one first uses only a single qubit to refer to the clock in the propagation Hamiltonian $\hprop$ --- this reduces the Hamiltonian's locality from $5$ to $3$. To compensate for this, one next imposes a large energy penalty on the invalid clock space by multiplying $\hstab$ by a polynomial factor.

\paragraph{Overview of the proof.} To prove that $\tlh$ is QMA-hard, the approach is to show a Karp or mapping reduction from an arbitrary instance of $\thlh$ to $\tlh$. To achieve this, given a $3$-local Hamiltonian $H$ acting on $n$ qubits, we map it to a $2$-local Hamiltonian $\tH$ as follows.
\begin{itemize}
    \item (Step 1: Rewrite $H$ by isolating $3$-local terms) Rewrite $H$ in a form which resembles $Y-6B_1B_2B_3$, where $B_1B_2B_3$ is shorthand for $B_1\otimes B_2\otimes B_3$, the $B_i$ are one-local and positive semidefinite, and $Y$ is a $2$-local Hamiltonian.
    \item (Step 2: Construct $\tH$) Define $\tH=Q+P(Y,B_1,B_2,B_3)$, where $P$ is an operator with small norm and which depends on $Y,B_1,B_2,B_3$, and where $Q$ has large spectral gap and depends only on the spectral gap of $H$. We refer to $P$ as the \emph{perturbation} and $\tH$ as the \emph{perturbed Hamiltonian}.
\end{itemize}
This outlines the reduction itself. It now remains to outline the proof of correctness, i.e. to show that the $2$-local Hamiltonian $\tH$ reproduces the low energy spectrum of the input $3$-local Hamiltonian $H$. In order to facilitate understanding, we present the analysis in a backwards fashion compared to the presentation in~\cite{KKR06}.

\begin{itemize}
    \item (Step 3: Define an effective Hamiltonian $\heff$) We first define an \emph{effective} $3$-local Hamiltonian $\heff$ whose low-energy spectrum is (by inspection) identical to that of $H$. We will see next that $\tH$ has been cleverly chosen to simulate $\heff$ with only $2$-local interactions.
    \item (Step 4: Define the self-energy, $\se$) A standard tool in perturbation theory is an operator-valued function known as the \emph{self-energy}, denoted $\se$, for $z\in\complex$. In this step, we show that for an appropriate choice of $z$, we have $\snorm{\heff-\se}\leq \epsilon$ for some small $\epsilon>0$. Intuitively, this relationship will hold because $\heff$ is simply a truncated version of the series expansion of $\se$.
    \item (Step 5: Relate the low energy spectrum of $\se$ to that of $\tH$) This step is where the actual perturbation theory analysis comes in. The outcome of this step will be that, assuming $\snorm{\heff-\se}\leq \epsilon$, the $j$th smallest eigenvalue of $\heff$ is $\epsilon$-close to the $j$th smallest eigenvalue of $\tH$.
\end{itemize}

To recap, once we define the $2$-local Hamiltonian $\tH$, the spectral analysis we perform follows the chain of relationships:
\[
    H\approx \heff \approx\se \approx \tH,
\]
where here $\approx$ roughly indicates that the two operators in question share a similar ground space. We now discuss each of these steps in further detail.

\subsection{Step 1: Rewrite $H$ by isolating $3$-local terms}

\begin{Question}
    Convince yourself that $H$ can be rewritten, up to rescaling by a constant, in the form
    \begin{equation}\label{eqn:rewrite}
        H \propto Y - 6\sum_{i=1}^M B_{i1}\otimes B_{i2}\otimes B_{i3},
    \end{equation}
    where each $B_{ij}$ is a one-local positive semidefinite operator and $Y$ is a $2$-local Hamiltonian whose operator norm is upper bounded by an inverse polynomial in $n$.
\end{Question}
\begin{Hint}
    Rewrite each local term of $H$ in the local Pauli operator basis, i.e. as a linear combination of terms of the form $\sigma_1\otimes \sigma_2\otimes \sigma_3$ for $\sigma_i\in\set{ I,\sigma^x, \sigma^y, \sigma^z}$. Then, for each such term involving Pauli operators, try to add 1-local multiples of the identity in order to obtain positive semidefinite terms $B_{i1}\otimes B_{i2}\otimes B_{i3}$. You will then have to subtract off certain terms to make up for this addition; these subtracted terms will form $Y$. Think about why $Y$ must indeed be $2$-local.
\end{Hint}

\subsection{Step 2: Construct $\tH$}
Using the decomposition of $H$ in Equation~\ref{eqn:rewrite}, we can now construct our desired $2$-local Hamiltonian, $\tH$. As done in~\cite{KKR06}, for simplicity, we assume that $M=1$ in Equation~\ref{eqn:rewrite}, i.e. that $H=Y-6B_1B_2B_3$. The extension to arbitrary $M$ follows similarly~\cite{KKR06}.

To construct $\tH$, suppose $H\in\LL ((\B)^{\otimes n})$. Then, we introduce three auxiliary qubits and define $\tH\in\LL((\B)^{\otimes n}\otimes(\B)^{\otimes 3})$ as follows~\cite{KKR06}.
\begin{eqnarray}
    \tH&:=&Q+P,\label{eqn:tildeQ}\\
    Q &:=& -\frac{1}{4\delta^3}I\otimes (\sigma_1^z\sigma_2^z+\sigma_1^z\sigma_3^z+\sigma_2^z\sigma_3^z - 3I),\\
    P &:=& (Y+\frac{1}{\delta}(B_1^2+B_2^2+B_3^2))\otimes I -\nonumber \\&&\frac{1}{\delta^2}(B_1\otimes \sigma_1^x+B_2\otimes \sigma_2^x + B_3\otimes \sigma_3^x),
\end{eqnarray}
where $\delta>0$ is some small constant, and $\sigma_j^i$ denotes the $i$th Pauli operator applied to qubit $j$. Notice that the ``unperturbed'' Hamiltonian $Q$ contains no information about $H$ itself, whereas the term that \emph{does} contain the information about $H$, $P$, is thought of as the ``perturbation''. The reason for this is that $Q$ is thought of as a \emph{penalty term} with a large spectral gap (compared to $\snorm{P}$), so that intuitively, the ground space of $\tH$ will be forced to be a subspace of the ground space of $Q$ (since $Q$ will enforce a large penalty on any vector not in this space).

\begin{Question}\label{q:GS1}
    Show that $Q$ has eigenvalues $0$ and $1/\delta^3$. Conclude that for ``small'' constant $\delta$, $Q$ has a ``large'' constant-sized spectral gap.
\end{Question}
\begin{Question}\label{q:GS2}
    Show that associated with the null space of $Q$ is the space
    \[
        L_{-} =(\B)^{\otimes n}\otimes \Span{\ket{000},\ket{111}}.
    \]
    For simplicity, we henceforth define $C:=\Span{\ket{000},\ket{111}}$. Conclude that we can think of $C$ as a \emph{logical} qubit, and that we can define logical Pauli operators $\sigma_C^i$ acting on $C$.
\end{Question}


\subsection{Step 3: Define an effective Hamiltonian $\heff$}
\begin{Question}
    Show that the effective Hamiltonian
    \begin{equation}\label{eqn:heff}
        \heff := Y\otimes I_C - 6B_1B_2B_3\otimes \sigma_C^x
    \end{equation}
    has the same ground state energy as $H=Y-6B_1B_2B_3$. Conclude that it suffices to show that the ground state energy of $\heff$ approximates that of $\tH$.
\end{Question}

\subsection{Step 4: Define the self-energy $\se$}

Let $\delta>0$ be a small constant. In this section, we define an operator-valued function $\se$, known as the self-energy, and show that for certain values of $z$, we have $\snorm{\heff-\se}\in O(\delta)$. At a high-level, the claim follows by using the series expansion of $\se$ to observe that for appropriate $z$, one has $\se=\heff + O(\delta)$.

\paragraph{Definition of $\se$.} To begin, suppose $\tH=Q+P$ acts on Hilbert space $\spa{H}$, for $Q$ the unperturbed Hamiltonian and $P$ the perturbation. Let $\lambda_*\in\reals$ be some cutoff value. Then, let $\spa{L}_-$ ($\spa{L}_+$) denote the span of $Q$'s eigenvectors with eigenvalue strictly less than $\lambda_*$ (at least $\lambda_*$), and let $\Pi_-$ ($\Pi_+$) denote the projector onto $\spa{L}_-$ ($\spa{L}_+$). For notational convenience, for any operator $A$, we define
\begin{eqnarray*}
    A_{+}&:=&\Pi_+ A\Pi_+\\
    A_{+-}&:=&\Pi_+ A\Pi_-\\
    A_{-+}&:=&\Pi_- A\Pi_+\\
    A_{-}&:=&\Pi_- A\Pi_-.
\end{eqnarray*}
To define the self-energy operator $\se$, we now require the notion of the \emph{resolvent} of an operator $A$, defined $R(z,A):=(zI-A)^{-1}$. Intuitively, the resolvent will be the vital link allowing us to connect the low-energy eigenvalues of our constructed $2$-local Hamiltonian $\tH$ to those of $\heff$ (and hence to those of our original $3$-local Hamiltonian $H$). Roughly, the range of $z$ we will be interested in for the resolvent is $z\in[c-\epsilon, d+\epsilon]$, where $c$ and $d$ are defined such that the eigenvalues of $\heff$ lie in the range $[c,d]$ for some $c<d$. Further details can be found in \S\ref{ssscn:low}, where explicit use of the resolvent comes into play. Returning to our discussion here, the self-energy is now defined as $\se:= zI_--R^{-1}_-(z,\tH)$. In this section, we use the fact~\cite{KKR06} that $\se$ has a simple and useful series expansion, given by
\begin{eqnarray}
    \se &=& Q_- + P_-+P_{-+}R_+P_{+-}+P_{-+}R_+P_{++}R_+P_{+-}+\nonumber\\&&P_{-+}R_+P_+R_+P_+R_+P_{+-}+\cdots\label{eqn:series}
\end{eqnarray}

\paragraph{$\se$ is close to $\heff$.} For our specific definition of $\tH$, to show that $\se$ is close to $\heff$, we simply show that $\heff$ is the series expansion of $\se$ up to the third order. Define $\Delta:=1/\delta^3$, and consider Equation~\ref{eqn:series}.

\begin{Question} We now compute the expression for $\se$ for $\tH$. Recall that our goal is to show that the low order terms of $\se$ are precisely our desired effective Hamiltonian, $\heff$.
    \begin{enumerate}
        \item Show that the zeroth order term of $\se$ is zero, i.e. $Q_-=0$.
        \item Show that $R_+=(z-\Delta)^{-1}I_{\spa{L}_+}$.
        \item Use parts 1 and 2 above to conclude that in our case,
        \begin{eqnarray*}
            \se &=& P_- + \frac{1}{z-\Delta}P_{-+}P_{+-} + \frac{1}{(z-\Delta)^2}P_{-+}P_{++}P_{+-} + \\&& \frac{1}{(z-\Delta)^3}P_{-+}P_+P_+P_{+-}+\cdots
        \end{eqnarray*}
        \item Show that
        \begin{eqnarray*}
            P_{-+} &=& -\frac{1}{\delta^{2}}(B_1\otimes \ketbra{000}{100}+B_2\otimes\ketbra{000}{010}+\\&&B_3\otimes\ketbra{000}{001}+B_1\otimes\ketbra{111}{011}+\\&&B_2\otimes\ketbra{111}{101}+B_3\otimes\ketbra{111}{110}).
        \end{eqnarray*}
        Derive similar expressions for $P_{+-}$, $P_-$, and $P_+$.
        \item For ease of notation, define $X:= (Y+\frac{1}{\delta}(B_1^2+B_2^2+B_3^2))$, and let $I_C$ denote the projector onto space $C$ from Question~\ref{q:GS2}. Show that
        \begin{eqnarray*}
            P_-&=& X\otimes I_C,\\
            P_{-+}P_{+-} &=&\frac{1}{\delta^4}(B_1^2+B_2^2+B_3^2)\otimes I_C,\\
            P_{-+}P_{++}P_{+-} &=& \frac{1}{\delta^4}(B_1XB_1+B_2XB_2+B_3XB_3)\otimes I_C - \\&&\frac{6}{\delta^6}B_1B_2B_3\otimes \sigma^x_C.
        \end{eqnarray*}
        \item By setting $z\ll\Delta$ a constant, show that $(z-\Delta)^{-1}=-\delta^3 + O(\delta^6)$.
        \item Finally, using all parts above, show that, as desired,
            \begin{equation*}
                \se = Y\otimes I_C - 6B_1B_2B_3\otimes \sigma^x_C+O(\delta)=\heff + O(\delta).
            \end{equation*}
            Conclude that $\snorm{\heff-\se}\in O(\delta)$.
    \end{enumerate}
\end{Question}

\subsection{Step 5: Relate the low energy spectrum of $\se$ to that of $\tH$}\label{ssscn:low}

In this section, we plug in Theorem 3 of \cite{KKR06}, which allows us to conclude that the small eigenvalues of $\se$ are close to the small eigenvalues of $\tH$.

\begin{theorem}\label{thm:3}{(Kitaev, Kempe, Regev~\cite{KKR06}, Theorem 3)} Let $\lambda_*$ be the cutoff on the spectrum of $Q$, as before. Assume the eigenvalues of $Q$ lie in the range $(-\infty,\lambda_*+\alpha/2]\cup[\lambda_*+\alpha/2,\infty)$ for some $\alpha\in\reals$, and that $\snorm{P}< \alpha/2$. Fix arbitrary $\epsilon>0$. Then, if there exists operator $\heff$ whose eigenvalues lie in the range $[c,d]$ for some $c<d<\lambda_*-\epsilon$, and such that $\snorm{\se-\heff}\leq \epsilon$ for all $z\in[c-\epsilon, d+\epsilon]$, then the $j$th eigenvalue of $\widetilde{\Pi}_-\tH\widetilde{\Pi}_-$ is $\epsilon$-close to the $j$th eigenvalue of $\heff$. Here, we assume eigenvalues are ordered for each operator in non-decreasing order, and we define $\widetilde{\Pi}_-$ as the projector onto the span of the eigenvectors of $\tH$ of eigenvalue strictly less than $\lambda_*$.
\end{theorem}

To sketch at a high-level the idea behind the proof of Theorem~\ref{thm:3}, recall that our goal in this section is to approximate the low-energy spectrum of an input $3$-local Hamiltonian $H$ with a $2$-local Hamiltonian $\tH$. To this end, we first observed that the low-energy spectrum of $H$ is identical to that of $\heff$. Next, we showed that $\heff$ is well-simulated in this low-energy space by the self-energy, $\Sigma_-(z)$, since the former is a truncation of the series expansion of the latter. Finally, it remains to link the spectrum of $\Sigma_-(z)$ to that of $\tH_-$ (where recall $\tH_-$ is $\tH$ projected onto its low-energy space), which is precisely the task of Theorem~\ref{thm:3}. Here, we are interested in the eigenvalues of $\tH_-$, and it is here that the resolvent plays a crucial role. Specifically, to show Theorem~\ref{thm:3}, one first shows that the eigenvalues of $\tH_-$ are encoded in the poles of $R(z,\tH)$. Moreover, projecting $R(z,\tH)$ onto the low energy space of $\tH$ preserves these poles, i.e. the poles of $R(z,\tH)_-$ \emph{also} encode the eigenvalues of $\tH$. But now we are in good shape, since the self-energy $\Sigma_-(z)$ is defined directly in terms of $R(z,\tH)_-$; it can thus be shown that the poles of $R(z,\tH)_-$ are related to the eigenvalues of $\Sigma_-(z)$. Finally, since $\Sigma_-(z)$ is ``close'' to $\heff$ by assumption in Theorem~\ref{thm:3}, the claim follows by applying a (special case of) Weyl's inequalities relating the distance between two operators in spectral norm to distances between their respective eigenvalues. This completes the proof sketch.

\begin{Question}
    Apply Theorem~\ref{thm:3} with $c=-\snorm{\heff}$, $d=\snorm{\heff}$, $\lambda_*=\Delta/2$ to complete the proof of correctness for the reduction.
\end{Question}

\section{Commuting $k$-local Hamiltonians and the Structure Lemma}\label{sscn:CLH}
A natural case of the $k$-local Hamiltonian problem whose complexity remains open (for arbitrary $k$ and local dimension $d$) is that of \emph{commuting} local Hamiltonians, i.e. where the local constraints pairwise commute. This class of Hamiltonians is particularly interesting, in that it intuitively seems ``closer'' to the classical world of constraint satisfaction (in which all constraints are diagonal in the computational basis and hence commute), and yet such Hamiltonians are nevertheless rich enough to give rise to highly entangled ground states, such as the toric code Hamiltonian~\cite{Kit03}. The main complexity theoretic question in this area is to characterize the complexity of the problem for various values of $k$ and $d$: Is it in NP? QCMA (i.e. QMA with a classical prover)? Or could it be QMA-complete?

To date, it is known that the commuting cases of $2$-local Hamiltonian for $d\geq 2$~\cite{BV05}, 3-local Hamiltonian with $d=2$ (as well as $d=3$ with a ``nearly Euclidean'' interaction graph)~\cite{AE11}, and 4-local Hamiltonian with $d=2$ on a square lattice are in NP~\cite{s11}. At the heart of these results is Bravyi and Vyalyi's  \emph{Structure Lemma}~\cite{BV05}, which is a powerful tool for dissecting the structure of commuting local Hamiltonians. The primary focus of this section is to prove and discuss this lemma.

We remark that often the commuting $\klhh$ problem is phrased with each local term being an orthogonal projection; this is without loss of generality, as since all terms simultaneously diagonalize, the ground state lies completely in some eigenspace of each constraint.

\subsection{Statement of the Structure Lemma}

Intuitively, the Structure Lemma says the following. Suppose we have two Hermitian operators $A\in\herm{\sX\otimes \sY}$ and $B\in\herm{\sY\otimes \sZ}$ for complex Euclidean spaces $\sX,\sY,\sZ$, such that $A$ and $B$ commute. Then, the space $\sY$ can be sliced up in such a way, that if we focus on just one slice $\sYi$ of the space (which is claimed by the NP prover to contain the ground state), then in this subspace $A$ and $B$ are completed \emph{decoupled}. Specifically, the lemma says we can write $\sY=\bigoplus_i \sYi$, such that if we restrict $A$ and $B$ to any one space $\sYi$, the resulting operators can be seen to act on disjoint parts of the space $\sYi$, hence eliminating their overlap. We now state the lemma more formally.

\begin{lemma}[Structure Lemma~\cite{BV05}]\label{l:structure}
    Suppose $A\in\herm{\sX\otimes \sY}$ and $B\in\herm{\sY\otimes \sZ}$ for complex Euclidean spaces $\sX,\sY,\sZ$, and such that $[A,B]=0$. Then, one can write $\sY=\bigoplus_i \sYi$, such that for any $i$, the following two properties hold:
    \begin{enumerate}
        \item $A$ and $B$ act invariantly on $\sYi$, and
        \item $\sYi = \sYio\otimes\sYit$ for some Hilbert spaces $\sYio$ and $\sYit$, such that $A|_{\sYi}\in\herm{\sX\otimes \sYio}$ and $B|_{\sYi}\in\herm{\sYit\otimes \sZ}$. In other words, within the subspace $\sYi$, $A$ and $B$ act non-trivially only on $\sYio$ and $\sYit$, respectively.
    \end{enumerate}
\end{lemma}

The strength of the Structure Lemma in proofs placing variants of commuting Hamiltonian in NP is as follows. First, note that by property 1 above, when looking for the joint ground state $\ket{\psi}$ of $A$ and $B$, we can safely restrict our attention to one appropriately chosen slice $i$. But which slice $i$ does $\ket{\psi}$ live in? This is not obvious \emph{a priori}; hence, we ask the NP prover to send us the correct choice of $i$. Then, by restricting $A$ and $B$ to space $\sYi$, by property 2 above the resulting Hamiltonians are decoupled. We can hence now easily diagonalize this system and determine the ground state energy, thus confirming whether it is indeed zero or bounded away from zero. Applying this idea repeatedly allows one to place the commuting $2$-local Hamiltonian problem in NP.

\subsection{Proof of the Structure Lemma}

In this section, we prove the Structure Lemma. The proof cleverly uses basic $C^*$ algebraic techniques. We remark that readers unfamiliar with $C^*$ algebras should not be put off; the Structure Lemma is a powerful tool worth understanding, and to be clear, once the terminology barrier of the $C^*$ formalism is overcome, the underlying proof is rather intuitive and simple. For this reason, we begin by defining the basic terminology required for the proof.\\

\noindent\textbf{Algebra:} Let $A$ be a vector space over $\complex$. Then $A$ is an algebra if it is endowed with a bilinear operation $\cdot:A\times A\mapsto A$. In our setting, $A$ will be some subset of linear operators taking $\complex^k$ to itself, and $\cdot$ is simply matrix multiplication. A {Banach} algebra is an associative algebra over the real or complex numbers which is also normed and complete, i.e. is a Banach space.\\

\noindent\textbf{C* Algebra:} To get a $C^*$ algebra, we start with a Banach algebra $A$ over $\complex$, and add a $*$-operation which has the following properties:
\begin{enumerate}
    \item For all $x\in A$, $x=(x^*)^*=x^{**}$.
    \item For all $x,y\in A$, $(xy)^*=y^*x^*$ and $(x+y)^*=x^*+y^*$.
    \item For all $c\in \complex$ and $x\in A$, $(cx)^*=\overline{c}x^*$.
    \item For all $x\in A$, $\norm{xx^*}=\norm{x}^2$.
\end{enumerate}
In our setting, the $*$-operation is simply the conjugate transpose of a matrix.\\

\noindent\textbf{Commutant:} Let $S\subseteq\lin{\complex^k}$. Then, the commutant of $S$ is defined as
\[
    S':=\set{x\in \lin{\complex^k}: xs=sx \text{ for all }s\in S}.
\]
A few facts about commutants come in handy: $S'$ is a $C^*$ algebra, $S\subseteq S^{''}$ (known as the closure of $S$), and $S^{'''}=S'$. Further, the following holds.
\begin{lemma}\label{l:commute}
    Let $A, B\subseteq\lin{\complex^k}$, and suppose for all $a\in A$ and $b\in B$, we have $ab=ba$. Then, for all $a\in A^{''}$ and $b\in B^{''}$, we have $ab=ba$.
\end{lemma}
\begin{proof}
    Observe that $B\subseteq A^{'}$. Thus, the elements in $A^{''}$ pair-wise commute with all elements of $B$, as well as those in $A^{'}\backslash B$. This implies $A^{''}\subseteq B^{'}$. But the elements of $B^{''}$ pairwise commute with those of $B^{'}$, which in turn implies they commute with the elements of $A^{''}$.
\end{proof}

\noindent\textbf{Center:} The center of algebra $A$ is the set of all elements in $A$ which commute with everyone else in $A$, i.e. $C(A):=A\cap A'$ (recall $A'$ is the commutant of $A$). A \emph{trivial} center is one which satisfies $C(A)=\set{cI\mid c\in\complex}$.

A simple application of the definition of the center yields some very useful well-known properties for all $x\in A$, as stated in the following lemma.

\begin{lemma}\label{l:directsum}
     Let $A$ be a $C^*$ algebra such that $A\subseteq\lin{\spa{X}}$ for complex Euclidean space $\spa{X}$. Suppose there exists $M\in C(A)$ with diagonalization $M=\sum_i\lambda_i \Pi_i$, where $\lambda\in \reals$ and $\Pi_i$ are (not necessarily one-dimensional) orthogonal projections. Then, any $N\in A$ has a block diagonal structure with respect to basis $\set{\Pi_i}$, i.e. can be written
     \[
        N = \bigoplus_i N_i
     \]
     where $N_i$ is an operator acting on the space $\Pi_i$ projects onto.
\end{lemma}
\begin{proof}
    We claim that it is true that for all $i$, $\Pi_i\in C(A)$; assuming this is true, the statement of the lemma holds simply because any $N\in A$ must now commute with each $\Pi_i$. To thus see that $\Pi_i\in C(A)$, note that if $M\in C(A)$, then for any polynomial $p$, $p(M)\in C(A)$. Then, defining for all $j$ a polynomial $p_j$ such that  $p_j(\lambda_i)=\delta_{ij}$ completes the proof.
\end{proof}
\begin{corollary}\label{cor:directsum}
    If $C^*$ algebra $A\subseteq\lin{\spa{X}}$ has a non-trivial center, then there exists a direct sum decomposition $\spa{X}=\bigoplus_i \spa{X}_i$ such that any $M\in A$ acts invariantly on each subspace $\spa{X}_i$.
\end{corollary}

\paragraph{Proof of Structure Lemma.} We can now proceed with the proof of the Structure Lemma (Lemma~\ref{l:structure}). As in the statement of the claim, let $A\in\herm{\sX\otimes \sY}$ and $B\in\herm{\sY\otimes \sZ}$ such that $[A,B]=0$. For an appropriate choice of operators $\set{A_{ij}},\set{B_{kl}}\subseteq\lin{\sY}$, one can write
\begin{eqnarray}
    A&=&\sum_{ij}\ketbra{i}{j}_{\sX}\otimes (A_{ij})_{\sY}\otimes I_{\sZ}\label{eqn:A}\\
    B&=&\sum_{kl}I_{\sX}\otimes (B_{kl})_{\sY}\otimes \ketbra{k}{l}_{\sZ}\label{eqn:B}.
\end{eqnarray}
\begin{Question}
    Use Equations~\ref{eqn:A} and~\ref{eqn:B} and the fact that $[A,B]=0$ to conclude that for all $i,j,k,l$, we have $[A_{ij},B_{kl}]=0$.
\end{Question}

Consider now algebras $\tA$ and $\tB$ generated by $\set{A_{ij}}$ and $\set{B_{kl}}$, respectively, i.e. $\tA=\set{A_{ij}}^{''}$ and $\tB=\set{B_{kl}}^{''}$, whose elements pairwise commute by Lemma~\ref{l:commute}. Focusing first on $\tA$, assume that $\tA$ has non-trivial center. We reduce this case to one with trivial center, which can then be solved directly. Specifically, by Corollary~\ref{cor:directsum}, we can first decompose $\sY=\bigoplus_i \sY_i$ such that $\tA$ acts invariantly on each $\sY_i$. In order to decouple $A$ and $B$, recall that our goal is to split $\sY_i=\sY_{i1}\otimes \sY_{i2}$ such that $A$ and $B$ act non-trivially only on $\sY_{i1}$ and $\sY_{i2}$, respectively. To this end, let $\tA_{i}$ denote $\tA$ restricted to space $\sY_i$, and assume without loss of generality that the subalgebra $\tA_i$ has trivial center (otherwise, we can decompose the space further). We now apply the following lemma, which is a standard result in the representation theory of $C^*$ algebras.
\begin{lemma}[\cite{KLV00,BV05}]\label{l:tensorprod}
    Let $A\subseteq \lin{\spa{Y}}$ be a $C^*$ algebra with trivial center. Then one can decompose $\spa{Y}=\spa{Y}_1\otimes \spa{Y}_2$ such that $A = \lin{\spa{Y}_1}\otimes I_{\spa{Y}_2}$.
\end{lemma}
\noindent We hence obtain a decomposition $\sY_i=\sY_{i1}\otimes \sY_{i2}$ such that $\tA_i$ acts non-trivially only on $\sY_{i1}$. In sum, we have thus far shown the following.

\begin{obs}\label{obs:1}
    The algebra $\tA$ is precisely the set of all operators $W\in\lin{\spa{Y}}$ which can be written as $W=\bigoplus_i (W_i)_{\spa{Y}_i}=\bigoplus_i (W_i)_{\sY_{i1}}\otimes I_{\sY_{i2}}$.
\end{obs}

We now use Observation~\ref{obs:1} to uncover the structure of $\tB$.

\begin{Question}\label{q:invar}
    For any $V\in \tB$ and $\spa{Y}_i$ in the decomposition of $\spa{Y}$ above, show that $V$ acts invariantly on $\spa{Y}_i$. Conclude that $V$ can also be written as a direct sum over spaces $\spa{Y}_i$, i.e. that $V=\bigoplus_i (V_i)_{\spa{Y}_i}$. (Hint: Use the fact that all operators in $\tA$ and $\tB$ pairwise commute, and Observation~\ref{obs:1}.)
\end{Question}

Having answered Question~\ref{q:invar}, we can now let $\tB_{i}$ denote $\tB$ restricted to space $\sY_i$. The answer to the following question completes the proof.

\begin{Question}
    Prove that any $V\in \tB_{i}$ has the form $V=I_{\sY_{i1}}\otimes (V_i)_{\sY_{i2}}$. (Hint: Use the fact that all operators in $\tA$ and $\tB$ pairwise commute, and Observation~\ref{obs:1}.)
\end{Question}

\section{Quantum 2-SAT is in P}\label{sscn:Q2SAT}

In this section, we discuss Bravyi's polynomial time algorithm~\cite{B06} for the Quantum 2-SAT problem. To begin, recall that in $k$-SATISFIABILITY ($k$-SAT), one is given as input a set of $k$-local constraints $\set{\Pi_{i}}$ acting on subsets of $k$ binary variables out of a total of $n$ variables $x_1,\ldots, x_n$. Each clause has the form
\[
    \Pi_i = x_{i,1}\vee\cdots\vee x_{i,k},
\]
where each $x_{i,j}$ is a \emph{literal} corresponding to either a variable or its negation. The question is whether there exists an assignment to the variables $x_1,\ldots, x_n$ such that all $\Pi_i$ evaluate to $1$.

The study of $k$-SAT has a long and rich history, which we shall not attempt to survey here. However, as is well-known, SAT is historically the first problem to be proven NP-complete~\cite{C72,L73}. Further, its restricted version $k$-SAT remains NP-complete for $k\geq 3$~\cite{K72}, but is polynomial-time solvable for $k=2$~\cite{K67,EIS76,APT79}. This raises the natural question: Can one define an appropriate quantum generalization of $k$-SAT, and could this generalization also lie in P when $k=2$?

In 2006, Bravyi answered both these questions in the affirmative~\cite{B06}. Here, we define Quantum $k$-SAT as follows.
\begin{definition}[Quantum $2$-SAT (\tq)~\cite{B06}]
        Given a set of $2$-local orthogonal projections $\set{\Pi_{ij}\mid 1\leq i,j,\leq n}$ acting on $n$ qubits, does there exist a satisfying quantum assignment, i.e. does there exist a state $\ket{\psi}\in(\complex^2)^{\otimes n}$ such that $\Pi_{ij}\ket{\psi}=0$ for all $1\leq i,j,\leq n$?
\end{definition}

\noindent Note that unlike in SAT, which would correspond to rank $1$ projections $\Pi_{ij}$, here the projections are allowed to be arbitrary rank. Thus, Bravyi's definition is more accurately a generalization of $2$-CSP, where arbitrary Boolean $2$-local constraints are allowed.

\subsection{The algorithm}

We now discuss Bravyi's algorithm for \tq. We remark that Bravyi's original exposition involved heavy use of tensors, which are perhaps not a typical tool in the computer scientist's toolkit. In contrast, we give a different description of the algorithm in terms of \emph{local filters}~\cite{Gi96} from entanglement theory, which, in our opinion, is arguably more accessible to the quantum computing community.

To begin, Bravyi's algorithm consists of three subroutines (defined subsequently), and can be described at a high level as follows.

\begin{enumerate}
    \item While there exists a constraint $\Pi_{ij}$ of rank at least $2$, run \emph{rankReduction}$(\Pi_{ij})$. If the call fails, reject.
    \item Run \emph{generateConstraints}. If the call succeeds, return to Step 1.
    \item Accept and return the output of \emph{solveSaturatedSystem}.
\end{enumerate}

\noindent Roughly, \emph{rankReduction} outputs a 2-QSAT instance (on a possibly smaller number of qubits) in which all constraints are rank $1$. Once all constraints are rank $1$, \emph{generateConstraints} attempts to add ``new'' constraints which are already implicit in the present constraint system in an attempt to ``saturate'' the system. If the call succeeds, we return to Step 1 to try to once again simplify the system. If, on the other hand, \emph{generateConstraints} fails, then we have arrived at a ``saturated'' system of constraints~\cite{B06}. At this point, we conclude the system is satisfiable; indeed, \emph{solveSaturatedSystem} outputs a satisfying assignment. We shall shortly discuss each of these three procedures in greater detail.

Before doing so, however, let us compare Bravyi's algorithm with known classical algorithms for $2$-SAT. In particular, Bravyi's algorithm can be thought of as a quantum generalization of Krom's 1967 algorithm~\cite{K67}, which we briefly sketch now. Specifically, one first runs a classical version of \emph{generateConstraints} repeatedly, which acts as follows: Given a pair of clauses overlapping on a bit with conflicting literals, say $(x_1\vee x_2)$ and $(\overline{x_2} \vee x_3)$, it produces a new redundant clause $(x_1\vee x_3)$. If we are able to generate a pair of conflicting clauses $(x_i\vee {x_i})$ and $(\overline{x_i}\vee \overline{x_i})$ for some $i$ (checking for this can be seen as a classical version of \emph{rankReduction}), we conclude the instance is unsatisfiable. Otherwise, as in Step 3 of Bravyi's algorithm, we conclude the instance is ``saturated'' and hence satisfiable, and we run a classical version of \emph{solveSaturatedSystem}, which extracts the satisfying assignment. We now discuss the components of Bravyi's algorithm in further depth.\\

\noindent\textbf{\emph{rankReduction}$(\Pi_{ij})$.} Given a constraint $\Pi_{ij}$ of rank at least $2$, act as follows:
\begin{itemize}
    \item If $\rankpi=4$, return \emph{fail}, as no assignment could satisfy this clause. This is analogous to checking in Krom's algorithm whether conflicting clauses $(x_i\vee {x_i})$ and $(\overline{x_i}\vee \overline{x_i})$ have been generated.
    \item If $\rankpi=3$, the assignment to qubits $i$ and $j$ is forced to be $I-\Pi_{ij}$. Set this as their assignment and remove them from the system, updating any $2$-local clauses acting on $i$ or $j$ as necessary.
    \item If $\rankpi=2$, qubits $i$ and $j$ are allowed to live in a $2$-dimensional subspace. Hence, combine $i$ and $j$ into a single merged qubit via an appropriate isometry. Update clauses acting on $i$ or $j$ as necessary.
\end{itemize}

\noindent\textbf{\emph{generateConstraints}.} We now give Bravyi's quantum generalization of Krom's iterative procedure for generating new redundant constraints from existing ones. Specifically, this subroutine examines all triples of qubits $\set{a,b,c}$ on which there exist clauses $\Pi_{ab}=\ketbra{\phi}{\phi}_{ab}$ and $\Pi_{bc}=\ketbra{\phi}{\phi}_{bc}$, and attempts to generate a new clause $\Pi_{ac}$.

To motivate the idea~\cite{B06}, suppose $\phiab=\phibc=\ket{\psi^-}$, for $\ket{\psi^-}=\ket{01}-\ket{10}$ the singlet (we omit normalization for simplicity). Then, since for two qubits $I-\ketbra{\psi^-}{\psi^-}$ projects onto the symmetric space, it follows that any assignment $\ket{\psi}$ must live in the symmetric space on qubits $\set{a,b}$ and $\set{b,c}$, and hence also on $\set{a,c}$. Thus, we can safely add the new (implicit) constraint $\Pi_{ac}=\ketbra{\phi}{\phi}_{ac}$.

Now, what if (say) $\ket{\phi}_{ab}$ is not the singlet? Here, we use the fact that any pure state on two qubits can be produced from the singlet via a \emph{local filter}~\cite{Gi96}. Specifically, there exist linear operators $A,C\in\lin{\complex^2}$ such that
\begin{equation}\label{eq:filters}
    \ket{\phi}_{ab}=A_a\otimes I_b \ket{\psi^-}_{ab}\quad\quad\text{and}\quad\quad    \ket{\phi}_{bc}=I_b\otimes C_c\ket{\psi^-}_{bc}.
\end{equation}
Note that while local filters were originally introduced to \emph{increase} the entanglement of a previously entangled state~\cite{Gi96} (in which case the filter must be invertible, as otherwise one can create entanglement via a local operation from a product state), in this setting, we are \emph{reducing} entanglement using a filter; thus, $A$ and $C$ in general will not be invertible. The following is our analogue of Lemma 1 in Reference~\cite{B06}.

\begin{lemma}\label{l:generate}
    For $\ket{\phi}_{ab}$ and $\ket{\phi}_{bc}$ as in Equation~\ref{eq:filters}, suppose $\ket{\psi}\in(\complex^2)^{\otimes n}$ satisfies $\braket{\psi}{\phi_{ab}}=\braket{\psi}{\phi_{bc}}=0$. Then, the constraint $\ket{\phi}_{ac}=A_a\otimes C_c\ket{\psi^-}$ on $\set{a,c}$ satisfies $\braket{\psi}{\phi_{ac}}=0$.
\end{lemma}
\begin{proof}
    Suppose for assignment $\ket{\psi}\in(\complex^2)^{\otimes n}$ that
\[
    \braket{\psi}{\phi_{ab}}=\bra{\psi}A_a\otimes I_b \ket{\psi^-}_{ab}=0.
\]
This implies that $A_a^\dagger\otimes I_b\ket{\psi}$ lives in the symmetric subspace on qubits $a$ and $b$. An analogous argument implies that $I_b\otimes C^\dagger_c\ket{\psi}$ lives in the symmetric subspace on $b$ and $c$. It follows that $A_a^\dagger\otimes C_c^\dagger\ket{\psi}$ lives in the symmetric subspace of both $\set{a,b}$ and $\set{b,c}$, and hence also of $\set{a,c}$. Thus, $\braket{\psi}{\phi_{ac}}=\bra{\psi}A_a\otimes C_c\ket{\psi^-}=0$, as desired.
\end{proof}

\noindent\textit{Remark.} For the reader interested in comparing Lemma~\ref{l:generate} above directly with Lemma 1 of Reference~\cite{B06}, the correspondence is given by $A=\phi\epsilon^\dagger$ and $C=\theta^T\epsilon$. This is easily seen by using the $\operatorname{vec}$ mapping~\cite{W08_2} (where roughly, $\operatorname{vec}$ ``reshuffles'' matrices into vectors) and its property that $A\otimes B\operatorname{vec}(X)=\operatorname{vec}(AXB^T)$.

In sum, \emph{generateConstraints} applies Lemma~\ref{l:generate} on triples of qubits until it generates a new clause on some pair of qubits $\set{a,c}$ which is linearly independent from existing clauses on $\set{a,c}$. If no such clause is found, the subroutine returns \emph{fail}.\\

\noindent\textbf{\emph{solveSaturatedSystem}.} A \emph{saturated} system of constraints is one in which (1) all constraints are rank $1$, and (2) for any triple of qubits $\set{a,b,c}$, Lemma~\ref{l:generate} fails to produce a new linearly independent constraint on $\set{a,c}$. We now give our analogue of Lemma 2 of Reference~\cite{B06}, which is in turn a quantum analogue of Krom's classical procedure for extracting a satisfying assignment from a ``saturated'' classical $2$-SAT system.

\begin{lemma}\label{l:homog}
    For any saturated system of constraints $\set{\Pi_{ij}}$, there exists an efficiently computable product state $\ket{\psi}=\bigotimes_{i=1}^n\ket{\psi_i}$ with $\ket{\psi_i}\in\complex^2$ such that $\Pi_{ij}\ket{\psi}=0$ for all $i,j$.
\end{lemma}

In order to prove Lemma~\ref{l:homog}, we first require the following.
\begin{lemma}\label{l:a}
    Let $E:=iY$ for Pauli operator $Y$. Then, for any $A\in\lin{\complex^2}$,
    \[
        A\otimes I\ket{\psi^-}=I\otimes E^\dagger A^T E\ket{\psi^-}.
    \]
\end{lemma}
\begin{proof}
    Let $\ket{\phi^+}=(\ket{00}+\ket{11})/\sqrt{2}$. Then,
    \begin{eqnarray*}
        A\otimes I\ket{\psi^-}&=&(I\otimes E^\dagger)(A\otimes E)\ket{\psi^-}\\&=&(I\otimes E^\dagger)(I\otimes A^T)\ket{\phi^+}\\&=&I\otimes E^\dagger A^T E\ket{\psi^-},
    \end{eqnarray*}
    where the first equality follows since $E^\dagger E=I$, the second since $I\otimes E\ket{\psi^-}=\ket{\phi^+}$ and $A\otimes I\ket{\phi^+}=I\otimes A^T\ket{\phi^+}$, and the third again since $I\otimes E\ket{\psi^-}=\ket{\phi^+}$.
\end{proof}

\begin{proof}[{Proof of Lemma~\ref{l:homog}}]
    We give a simple deterministic polynomial time algorithm which outputs $\ket{\psi}$. Pick an arbitrary qubit, $q_1$, and set its assignment to $\ket{0}$, i.e. set $\ket{\psi_1}=\ket{0}$. Now consider the neighbor set of $q_1$, denoted $N(q_1)$. For any $q_i\in N(q_1)$, suppose the forbidden space is spanned by $\ket{\phi_{1,i}}$. Then, observing that $\ket{\phi_i}:=\braket{\psi_1}{\phi_{1,i}}\in\complex^2$ and that $\bra{v}E\overline{\ket{v}}=0$ for any $\ket{v}\in \complex^2$ and where $\overline{\ket{v}}$ denotes the entry-wise complex conjugate of $\ket{v}$, it follows that $\bra{\phi_{1,i}}(\ket{\psi_1}\otimes E\overline{\ket{\phi_i}})=0$. Thus, setting $\ket{\psi_i}=E\overline{\ket{\phi_i}}$ satisfies all clauses between $q_1$ and its neighbors. Moreover, by Lemma~\ref{l:generate}, we  have that any clause between distinct qubits $q_i,q_j\in N(q_1)$ is also satisfied by this assignment.

    Let $S$ denote the set of qubits whose shortest path from $q_1$ is precisely $2$ in the interaction graph, i.e. $S=N(N(q_1))\backslash (N(q_1)\cup \set{q_1})$. If we can now show that for all $q_i\in N(q_1)$ and $q_j\in S$, the clause $\ket{\phi_{i,j}}$ is satisfied by the current assignment regardless of the assignment on $q_j$, then note that the proof is complete, as we can simply iterate the argument above by discarding all clauses which act on qubits in $\set{q_1\cup N(q_1)}$ and choosing a new starting vertex $q_j \in S$.

    Thus, we now prove that for all $q_i\in N(q_1)$ and $q_j\in S$, $\braket{\phi_{i,j}}{\psi_i}=0$. Let the clauses on $(1,i)$ and $(i,j)$ be given by
    \[
        \ket{\phi_{1,i}}=A_1\otimes I_i \ket{\psi^-}\quad\quad\text{and}\quad\quad\ket{\phi_{i,j}}=I_i\otimes C_j \ket{\psi^-}.
    \]
    Then, analogous to Reference~\cite{B06}, the key observation is that since $(1,j)$ is not an edge, then by Lemma~\ref{l:generate}, we must have $A\otimes C\ket{\psi^-}=0$. This, along with Lemma~\ref{l:a}, together imply:
    \begin{eqnarray*}
        \braket{\phi_{i,j}}{\psi_i} &=& (\bra{\psi^-}I_i\otimes C^\dagger_j )((\overline{\bra{\psi_1}}\otimes E_i)\overline{\ket{\phi_{1,i}}})\\
        &=&(\bra{\psi^-}I_i\otimes C^\dagger_j )((\overline{\bra{\psi_1}}\otimes E_i)(\overline{A}_1\otimes I_i \ket{\psi^-}))\\
        &=&(\bra{\psi^-}I_i\otimes C^\dagger_j )((\overline{\bra{\psi_1}}\otimes E_i)(I_1\otimes E_i^\dagger A_i^\dagger E_i\ket{\psi^-}))
        \\
        &=&(\bra{\psi^-}A_i^\dagger\otimes C^\dagger_j )(\overline{\bra{\psi_1}}\otimes E_i)\ket{\psi^-}\\
        &=&0.
    \end{eqnarray*}
\end{proof}

\section{Area laws for one-dimensional gapped quantum systems}\label{sscn:arealaw}

In this section, we review Arad, Kitaev, Landau and Vazirani's combinatorial proof~\cite{ALV12,AKLV13} of Hastings'~\cite{Ha07} 1D area law for gapped systems. Specifically, we consider a $1$D chain of $n$ quantum systems governed by a gapped local Hamiltonian $H$. By \emph{gapped}, we mean that
the difference between the smallest and second-smallest eigenvalues of $H$ is lower bounded by a constant $\epsilon>0$, and that the ground state is unique.

Let us begin by clarifying the statement regarding $H$ to be proven.  First, note that the statement of an area law is simple in the $1$D case, since by definition the surface area of any contiguous region in a 1D system must be constant (i.e. either $1$ or $2$). It follows that proving an area law for $1$D systems is equivalent to proving that for any $1\leq i\leq n$, the entanglement entropy of the ground state between particles $1,\ldots,i$ and particles $i+1,\ldots,n$ is bounded above by a constant. Here, the entanglement entropy is defined as the von Neumann entropy of the reduced density matrix on particles $1,\ldots,i$. Without loss of generality, it suffices to fix an arbitrary cut and prove the statement there; let $C$ denote this cut.

\paragraph{High-level overview.} Conceptually, the proof consists of two steps. First, it is shown that there exists a product state with respect to cut $C$ which has reasonably good, i.e.~constant, overlap with the ground state of $H$. Second, we show how one can ``transform'' such a product state to a much better approximation of the ground state \emph{without} increasing the entanglement entropy across $C$ too much. (Recall that our goal is to prove that in the ground state, the entanglement entropy across cut $C$ is constant.) In the proof of~\cite{ALV12, AKLV13}, both of these steps depend crucially on a theoretical construct called an \emph{Approximate Ground-Space Projection (AGSP)}, and the bulk of the work goes into the construction of an AGSP with sufficiently good parameters. In contrast, we remark that Hastings' original proof~\cite{Ha07} also follows the broad outline of the above two steps, but realizes them using different physics-inspired techniques such as the Lieb-Robinson bound or ``monogamy of entanglement''-type arguments.\\

For simplicity, in this review we focus on the simpler case of frustration-free Hamiltonians. Readers interested in the proof of the frustrated case are referred to \S6 of \cite{AKLV13}.

\paragraph{Organization.} In \S\ref{ssscn:AGSP}, we first define an AGSP. Then, we show that if a good AGSP and product state with non-trivial overlap onto the ground state of $H$ exist, then an area law holds (Lemma~\ref{l:arealaw1}). In \S\ref{ssscn:AGSPprod}, we show that the existence of a good AGSP already implies a good product state (Lemma~\ref{l:arealaw2}). Finally, \S\ref{ssscn:AGSPconstruct} constructs a good AGSP.

\subsection{Approximate Ground-Space Projection (AGSP)}\label{ssscn:AGSP}

We now motivate and define an AGSP. Specifically, returning to the second step of the high-level proof overview discussed above, we ask: Given a product state $\psiprod$ with non-trivial overlap with the ground state of $H$, how can we map $\psiprod$ to a good approximation of the ground state? One obvious idea is to apply to $\psiprod$ the projection onto the ground space, thus obtaining a scaled-down version of the ground state. Unfortunately, this approach does not give us a way to bound the amount of entanglement generated across cut $C$ when the projection is applied. The main idea behind an AGSP is to only \emph{approximately} project onto the ground space; in return, we obtain a rigorous bound on how the entanglement grows with each application of the AGSP.

\begin{definition}\label{def:agsp}
    An operator $K$ is said to be a $(D,\Delta)$-AGSP if the following conditions hold:
\begin{enumerate}
\item \textbf{Ground space invariance:} For any ground state $\ket\Gamma, K\ket\Gamma=\ket\Gamma$.
\item \textbf{Shrinking:} If $\ket{\Gamma^\perp}$ is any state orthogonal to the ground space, then $K\ket{\Gamma^\perp}$ is also orthogonal to the ground space, and moreover $||K\ket{\Gamma^\perp}||^2\leq \Delta$.
\item \textbf{Entanglement:} The Schmidt rank of $K$ across the given cut is at most $D$. Note that the Schmidt rank of an operator $A$ is defined as the smallest integer $m$ such that $A$ can be written in the form $A=\sum_{1\leq i\leq m} L_i\otimes R_i$.
\end{enumerate}
\end{definition}
\noindent In other words, if there exists an AGSP with good parameters, i.e.~small $D$ and $\Delta$, we can repeatedly apply it to $\psiprod$ until we obtain a good approximation to the ground state. Note that although the action of an AGSP will converge to that of the exact ground state projection as the number of iterations goes to infinity, a \emph{finite} number of iterations reveals a delicate tradeoff between proximity to the ground state and a bound on the entanglement entropy.

\paragraph{Using an AGSP with a good product state.} The high-level lemma which stitches together the various parts of the proof is the following. In words, it says that the existence of a good AGSP and a good product state suffices to establish an area law.

\begin{lemma}[Arad, Landau, and Vazirani~\cite{ALV12}]\label{l:arealaw1}
Let $\ket\Gamma$ be the ground state of $H$ and $\ket\phi$ a product state such that $|\braket{\phi}{\Gamma}|=\mu$. Then, the existence of a $(D,\Delta)$-AGSP $K$ implies that the entanglement entropy $S$ of $\ket\Gamma$ is bounded by
\[
S\leq O(1)\cdot \frac{\log \mu^2}{\log \Delta}\log D.
\]
\end{lemma}
\begin{proof}
Let $\ket\Gamma=\sum_i \lambda_i \ket {L_i}\otimes \ket{R_i}$ be the Schmidt decomposition of the ground state $\ket\Gamma$, where the Schmidt coefficients $\lambda_i$ are in decreasing order. Then, the entanglement entropy is defined as $-\sum_i \lambda_i^2\log\lambda_i^2$. To bound this quantity, we consider the family of states $K^\ell\ket\phi$. By Definition~\ref{def:agsp}, we know that these states satisfy the following two properties:
\begin{enumerate}
\item The Schmidt rank of $K^\ell\ket\phi$ is at most $D^\ell$.
\item The inner product between $K^\ell\ket\phi$ and $\ket\Gamma$ is at least $\mu/\sqrt{\mu^2+\Delta^\ell(1-\mu^2)}$.
\end{enumerate}

We now use these facts to bound the entropy in two steps. We show the first of these steps, and guide the reader through the second step via a sequence of questions. Our starting point is the Eckart-Young theorem \cite{EY36}, which implies that the magnitude of the inner product between $\ket\Gamma$ and any normalized state with Schmidt rank $r$ is upper-bounded by the Euclidean norm of the vector of the first $r$ Schmidt coefficients, i.e. by $(\sum_{i=1}^r \lambda_i^2)^{1/2}$. Therefore,
\[
\sum_{i\leq D^\ell} \lambda_i^2\geq \frac{\mu^2}{\mu^2+\Delta^\ell (1-\mu^2)}\geq \frac{\mu^2}{\mu^2+\Delta^\ell},
\]
where the first inequality follows from the Eckart-Young theorem and point $2$ above. This implies that
\[
\sum_{i> D^\ell} \lambda_i^2\leq 1-\frac{\mu^2}{\mu^2+\Delta^\ell}=\frac{\Delta^\ell}{\mu^2+\Delta^\ell}\leq \frac{\Delta^\ell}{\mu^2} =: p_\ell.
\]
If we now choose $\ell_0=\lceil\frac{\log \mu^2}{\log \Delta}\rceil$ so that $p_{\ell_0}<1$, it follows that the contribution of entropy from the first $D^{4\ell_0}$ Schmidt coefficients is at most
\[
-\sum_{i\leq D^{4\ell_0}} \lambda_i^2\log\lambda_i^2 \leq \log D^{4\ell_0}=4\ell_0\log D=O(1)\cdot \frac{\log\mu^2}{\log \Delta}\log D,
\]
where the first inequality follows since the entropy of a (sub)normalized $d$-dimensional vector is at most $\log d$.

The next step is to bound the contribution to the entropy of the remaining Schmidt coefficients, which the following two questions guide the reader through.

\begin{Question}\label{q:arealaw1}
Show that if $\ell> \ell_0$, the contribution of entropy from the $(D^{2\ell}+1)$-th to $D^{2(\ell+1)}$-th Schmidt coefficients is at most $\Delta^{\ell-\ell_0}(\ell-\ell_0)\log \frac{D^{2(\ell+1)/(\ell-\ell_0)}}{\Delta}$.
\end{Question}
\begin{Hint}
Use the fact that $\sum_{i=D^{2\ell}+1}^{D^{2(\ell+1)}}\lambda_i^2\leq p_{2\ell}\leq p_\ell$.
\end{Hint}
\begin{Question}
Using the answer to Question \ref{q:arealaw1}, show that
\[
-\sum_{i> D^{4\ell_0}} \lambda_i^2\log\lambda_i^2 \leq \frac{\Delta}{(1-\Delta)^2}\log\frac{D^6}{\Delta}.
\]
\end{Question}
\begin{Hint}
Use the series equality $\sum_{j\geq 1} jr^j=\frac{r}{(1-r)^2}$.
\end{Hint}

Combining the two bounds we have derived, it follows that the entanglement entropy of the ground state satisfies
\[
S\leq O(1)\cdot \frac{\log\mu^2}{\log \Delta}\log D+\frac{\Delta}{(1-\Delta)^2}\log\frac{D^6}{\Delta}.
\]
Finally, note that for any $k\geq 1$, $K^k$ is a $(D^k,\Delta^k)$-AGSP. Moreover, if we choose $k=\lceil -\frac{1}{\log \Delta}\rceil$, we can ensure that $\frac{1}{4}\leq \Delta^k \leq \frac{1}{2}$. Substituting $D^k$ for $D$ and $\Delta^k$ for $\Delta$, we obtain
\[
S\leq O(1)\cdot k\cdot (-\log\mu^2\log D+\log D-1)\leq O(1)\cdot \frac{\log \mu^2}{\log \Delta}\log D.
\]

\end{proof}

\subsection{Good AGSP implies a good product state}\label{ssscn:AGSPprod}

Lemma~\ref{l:arealaw1} stated that to prove the area law, it suffices to have a good AGSP along with a product state with constant overlap with the ground state. It turns out that the former criterion actually implies the latter, as we now show in this section via the following lemma.

\begin{lemma}\label{l:arealaw2}
If there exists a $(D,\Delta)$-AGSP $K$ such that $D\cdot \Delta\leq \frac{1}{2}$, then there exists a product state $\ket \phi=\ket L\otimes\ket R$ whose overlap with the ground state $\ket \Gamma$ is $\mu=|\braket{\Gamma}{\phi}|\geq1/\sqrt{2D}$.
\end{lemma}
\begin{proof}
We proceed by contradiction. Let $\ket{\phi'}$ be a product state with the maximum possible overlap $\mu$ with the ground state, and assume for sake of contradiction that $\mu<1/\sqrt{2D}$. Then, consider the state $\ket \phi:={K\ket{\phi'}}/{||K\ket{\phi'}||}$. By definition of an AGSP, the Schmidt rank of $\ket\phi$ is at most $D$, and therefore it can be written as
$\ket\phi=\sum_{i=1}^D \lambda_i\ket {L_i}\otimes\ket{R_i}$. It follows that
\begin{eqnarray*}
\frac{\mu^2}{||K\ket{\phi'}||^2} &=&|\braket{\Gamma}{\phi}|^2\\&\leq& \left(\sum_{i=1}^D \lambda_i \left|\bra \Gamma (\ket {L_i}\otimes \ket{R_i})\right|\right)^2\\&\leq&  \sum_{i=1}^D\left|\bra \Gamma (\ket {L_i}\otimes \ket{R_i})\right|^2,
\end{eqnarray*}
where the last inequality uses the Cauchy-Schwarz inequality and the fact that $\sum_i \lambda_i^2=1$. Therefore, there exists some $i$ such that
\[
\left|\bra \Gamma (\ket {L_i}\otimes \ket{R_i})\right|^2 \geq \frac{\mu^2}{D||K\ket{\phi'}||^2}\geq \frac{\mu^2}{D(\mu^2+\Delta)},
\]
where the last inequality follows from the fact that $K$ is a $(D,\Delta)$-AGSP. However, $D(\mu^2+\Delta)=D\mu^2+D\Delta< \frac{1}{2}+\frac{1}{2}=1$, which implies that $\ket{L_i}\otimes \ket{R_i}$ has a larger overlap with the ground state than $\ket{\phi'}$ does. This is a contradiction, and therefore we have shown that $\mu\geq 1/\sqrt{2D}$.

\end{proof}

To recap, by combining Lemmas \ref{l:arealaw1} and \ref{l:arealaw2}, we have the following theorem, which sets the stage for our final area laws section on constructing a good AGSP (\S\ref{ssscn:AGSPconstruct}).

\begin{theorem}\label{thm:arealaw}
If there exists a $(D,\Delta)$-AGSP such that $D\cdot\Delta \leq \frac{1}{2}$, the ground state entanglement entropy is bounded by $O(1)\cdot \log D$.
\end{theorem}

\subsection{Constructing a good AGSP}\label{ssscn:AGSPconstruct}

Theorem~\ref{thm:arealaw} implies that in order to show an area law, it suffices to construct a good AGSP. In this section, we do precisely this. For simplicity in exposition, in this review, our exposition will attempt to deliver the main ideas behind the construction, without overwhelming the reader with technical details.

To begin, how does one go about constructing an AGSP, i.e.~an operator that leaves the ground state invariant and shrinks every vector that is orthogonal to the ground state? Since the only information regarding the ground state in our possession is the Hamiltonian $H$ itself, we use $H$ as our starting point. Consider the first attempt of $K=I-H/||H||$; it is easy to check that $K$ leaves the ground state invariant and cuts the norm of any orthogonal state to at most $1-\frac{\epsilon}{||H||}$ (recall that $\epsilon$ is the spectral gap of the given Hamiltonian). Thus, $K^k$ for large values of $k$ yields an AGSP with a good value of $\Delta$. Unfortunately, however, in general with this construction $D$ will also grow exponentially in $k$. Thus, this candidate AGSP does not satisfy our required condition that $D\cdot\Delta\leq\frac{1}{2}$.

However, this first attempt has taught us something --- that we can manipulate the spectrum of $H$ via matrix polynomials. Namely, if we have a polynomial $f(x)$ that maps $0$ to $1$ and $[\epsilon,||H||]$ to $[0,\Delta]$, $f(H)$ will be an AGSP with the shrinking factor $\Delta$. The remaining task then would be to show that $f(H)$ has a small Schmidt rank. Naturally one would expect that $f$ must have a small degree, because the naive bound on the Schmidt rank grows at least exponentially in the degree of the polynomial.

\begin{figure}\centering
  \includegraphics[height=5.3cm]{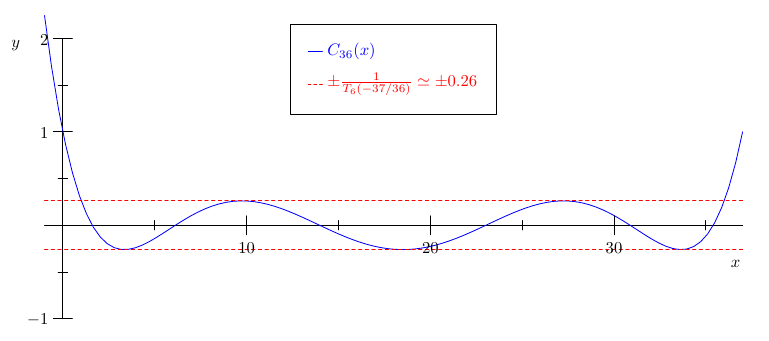}
  \caption{The polynomial $C_{36}(x)$, constructed based upon the Chebyshev polynomial of the first kind.}
  \label{fig:chebyshev}
\end{figure}

But which $f$ should we choose? In approximation theory, there is a well-known family of polynomials $T_\ell$ called the \textbf{Chebyshev polynomial of the first kind}, which has the following properties~\cite{Chebyshev_Poly}:
\begin{enumerate}
\item The degree of $T_\ell(x)$ is $\ell$.
\item For $x\in[-1,1]$, $|T_\ell(x)|\leq 1$.
\item For $x>1$, $T_\ell(x)\geq\frac{1}{2}e^{2\ell\sqrt{(x-1)/(x+1)}}$.
\end{enumerate}
The main point is that once $x$ passes the value $1$, the polynomial begins to increase very rapidly, giving rise to a threshold-like behavior. We will use this behavior to address the challenge that the AGSP should have an eigenvalue of $1$ for the ground state while it should have a very small eigenvalue for the first excited state which is only $\epsilon$ away from the ground state in the energy spectrum. The strategy is to make the first excited state correspond to $T_\ell(1)$ and the largest eigenvector to $T_\ell(-1)$ so that all excited states have eigenvalues at most $1$, while the ground state, which will then correspond to $T_\ell(1+y)$ for some small $y$, has a much greater eigenvalue. Then we can renormalize the operator so that the eigenvalue of the ground state becomes $1$, as required by the definition of an AGSP.

In other words, we construct a family of polynomials $C_\ell$ by scaling and translating $T_\ell$ as follows:
\[
C_\ell(x) = \left.T_\ell\left(\frac{||H||+\epsilon-2x}{||H||-\epsilon}\right)\middle/T_\ell\left(\frac{||H||+\epsilon}{||H||-\epsilon}\right)\right..
\]
It is then straightforward to check that $C_\ell(H)$ is an AGSP with $\Delta=4e^{-4\ell\sqrt{\epsilon/||H||}}$.

\begin{figure}\centering
  \includegraphics{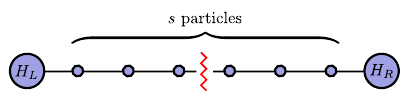}
  \caption{Hamiltonian truncation. We focus on a segment of $s$ particles around the cut, denoting the multiparticle Hamiltonians to the left and right of the segment by $H_L$ and $H_R$ respectively.}
  \label{fig:truncation}
\end{figure}

Unfortunately, an immediate difficulty in using the above bound on $\Delta$ is that $||H||$ can be very large, significantly slowing down the decay of $\Delta$. This issue can be addressed by a technique called \emph{Hamiltonian truncation}. The technique consists in picking out $s$ particles around the cut that we are concerned with and truncating the upper spectrum of the other ``less important'' particles. More precisely, we write the Hamiltonian as $H=H_L+H_1+\cdots+H_s+H_R$ as shown in Figure \ref{fig:truncation}, and then replace $H_L$ and $H_R$ by their respective truncated versions $H_L^{\leq t}$ and $H_R^{\leq t}$, where $A^{\leq t}$ denotes the matrix obtained from $A$ by replacing all eigenvalues greater than $t$ with $t$ and leaving the other eigenvalues unchanged. It is clear that the resulting truncated Hamiltonian $H'$ has norm bounded by $s+2t$, and it can also be shown that there exists some constant $t$ such that the spectral gap of $H'$ remains to be $\Omega(1)\cdot\epsilon$. Moreover, since the Hamiltonian is frustration-free, $H$ and $H'$ have the same (unique) ground state.\footnote{Note that this is the only use of the frustration-free assumption in our presentation. In fact, the proof of the frustrated case also follows exactly the same outline except that there we need a more delicate argument to prove that the ground states of $H$ and $H'$, which may now be distinct, are very close to each other. The interested reader can refer to \S6 of \cite{AKLV13}.} This means that we can use $H'$ in place of $H$ in our construction of the AGSP.

Finally, we need to show that the resulting AGSP $C_\ell(H')$ has a small Schmidt rank across the given cut. Intuitively speaking, this follows from the fact that each term in the expansion of $C_\ell(H')$ is a product of at most $\ell$ terms of $H'$, and therefore on average contributes Schmidt rank of $d^{\ell/(s+1)}$ to a cut, where $d$ is the local dimension of the particles (there are $s+1$ possible cuts in $H'$). The challenge here is that there are exponentially many terms in the expansion of $C_\ell(H')$ and thus we cannot simply sum these up. Fortunately, it is possible to circumvent this problem by cleverly grouping the exponential number of terms into a small number of large sums, with each sum again having small entanglement. The argument uses polynomial interpolation (see Lemma 4.2 of \cite{AKLV13}) and we end up with the bound of $D=(d\ell)^{O(\max\{\ell/s,\sqrt{\ell}\})}$.



Comparing this bound with our previously derived expression $\Delta=4e^{-4\ell\sqrt{\epsilon/||H||}}$ reveals that a suitable choice of parameters $\ell=O(s^2)$ and $s=O((\log^2 d)/\epsilon))$ gives us $D\cdot\Delta\leq \frac{1}{2}$. We conclude that, by Theorem~\ref{thm:arealaw}, there exists an area law for 1D gapped systems.

\begin{acknowledgements}
\addcontentsline{toc}{chapter}{Acknowledgements}
This Simons Institute Monograph has its roots in the UC Berkeley Quantum Reading Group of Spring 2013, during which many discussions took place. There were a number of participants involved in these discussions, whom we wish to thank: Piyush Srivastava, who contributed to our discussion on \emph{Thermal equilibrium, Boltzmann distribution \& Gibbs state} in \S\ref{sscn:terms}, Guoming Wang, as well as Vamsi K. Devabathini for remote participation and contributions from IIT Madras in India. SG would like to thank Niel de Beaudrap and Or Sattath for very helpful discussions regarding \S\ref{sscn:Q2SAT}, Maris Ozols for insightful comments on the first version of this document, Tomotoshi Nishino for helpful comments on the history of tensor network states, and David Gosset.

One of the initial aims of this survey was to act as an introduction to the field for participants of  the Spring 2014 program entitled ``Quantum Hamiltonian Complexity'', held at the Simons Institute for the Theory of Computing at UC Berkeley. Over the course of this program, the survey matured and was improved greatly due to feedback from and exposure to the field. In this regard, SG would like to especially thank Ignacio Cirac for a pair of wonderful talks and ensuing discussions (on which Chapter~\ref{scn:motivation} is based) at the Simons Institute. We also thank the Simons Institute for hosting both the authors and the Quantum Hamiltonian Complexity program itself.

SG acknowledges financial support from a Government of Canada NSERC Banting Postdoctoral Fellowship and the Simons Institute for the Theory of Computing at UC Berkeley. YH is supported by DARPA OLE. ZL is supported by ARO Grant W911NF-12-1-0541, NSF Grant CCF-0905626, Templeton Foundation Grant 21674, and the Simons Institute for the Theory of Computing. SWS is supported by NSF Grant CCF-0905626 and ARO Grant W911NF-09-1-0440.
\end{acknowledgements}

\backmatter

\bibliographystyle{plain}
\bibliography{Sevag_Gharibian_Central_Bibliography_Abbrv,Sevag_Gharibian_Central_Bibliography}

\end{document}